\documentclass[10pt,a4paper]{article}
\usepackage{latexsym,amssymb,amsmath,amsthm,amsfonts, amscd, enumerate,verbatim,xspace,exscale}

\input xy
\xyoption{all} \CompileMatrices \UseComputerModernTips

\usepackage{graphicx,tabularx}
\usepackage{amsmath,amsbsy,amsfonts,amssymb}
\usepackage{color}
\usepackage{amsthm}
\usepackage{mathtools}
 \def\bm#1{\boldsymbol{#1}}

\title{Conserved quantities and Hamiltonization of nonholonomic systems}

\author{{\sc{Paula Balseiro}\thanks{
         Universidade Federal Fluminense, Instituto de Matem\'atica e Estat\'\i stica,  Rua Prof. Marcos Waldemar de Freitas Reis S/N (Campus do Gragoat\'a), CEP 24210-201, Niteroi, Rio de Janeiro, Brazil. \newline{\texttt{E-mail: pbalseiro@vm.uff.br}}}} \ \  and \
{\sc{Luis P. Yapu}\thanks{
        Universidade Federal Fluminense, Instituto de Matem\'atica e Estat\'\i stica, Rua Prof. Marcos Waldemar de Freitas Reis S/N (Campus do Gragoat\'a), CEP 24210-201,  Niteroi, Rio de Janeiro, Brazil. \newline{\texttt{E-mail: luis.yapu@gmail.com}}}} }

\theoremstyle{plain}
\newtheorem{theorem}{Theorem}[section]
\newtheorem{lemma}[theorem]{Lemma}
\newtheorem{proposition}[theorem]{Proposition}
\newtheorem{corollary}[theorem]{Corollary}
\newtheorem{remarkth}[theorem]{Remark}

\theoremstyle{definition}
\newtheorem{definition}[theorem]{Definition}
\newenvironment{remark}{\begin{remarkth}\upshape}{\hfill$\diamond$\end{remarkth}}

\newcommand{\g}{\mathfrak{g}}
\def\W{\mathcal{W}}
\def\M{\mathcal{M}}

\def\V{\mathcal{V}}
\def\S{\mathcal{S}}
\def\C{\mathcal{C}}

\def\Ham{\mathcal{H}}

\def\R{\mathbb{R}}

\def\red{{\mbox{\tiny{red}}}}
\def\nh{{\mbox{\tiny{nh}}}}
\def\B{{\mbox{\tiny{$B$}}}}
\def\subW{{\mbox{\tiny{$\W$}}}}
\def\subC{{\mbox{\tiny{$\C$}}}}
\def\subS{{\mbox{\tiny{$\S$}}}}
\def\subM{{\mbox{\tiny{$\M$}}}}

\def\vecOm{\boldsymbol{\Omega}}
\def\I{\mathbb{I}}
\def\a{\alpha}

\def\vecgamma{\boldsymbol{\gamma}}
\def\vecalpha{\boldsymbol{\alpha}}
\def\vecbeta{\boldsymbol{\beta}}

\date{}
\begin{document}
\maketitle

\begin{abstract}


This paper studies hamiltonization of nonholonomic systems using geometric tools.  By making use of symmetries and suitable first integrals of the system,  we explicitly define a global 2-form for which the gauge transformed nonholonomic bracket gives rise to a new bracket on the reduced space codifying the nonholonomic dynamics and carrying an almost symplectic foliation (determined by the common level sets of the first integrals). 
In appropriate coordinates, this 2-form is shown to agree with the one previously introduced locally in \cite{LGN-JM-16}. We use our  coordinate-free viewpoint to study various geometric features of the reduced brackets. We apply our formulas to obtain a new geometric proof of the hamiltonization of a homogeneous ball rolling without sliding in the interior side of a convex surface of revolution using our formulas.




\end{abstract}

\tableofcontents

\section{Introduction} \label{S:Intro}


It is well known that nonholonomic systems are not hamiltonian; instead, they are geometrically described by
{\em almost} Poisson structures, the so-called {\it nonholonomic brackets}  \cite{IdLMM, Marle, SchaftMaschke1994}.
The {\it hamiltonization} problem studies whether a nonholonomic system admits a hamiltonian formulation after a reduction by symmetries, and much work has been done on this problem in recent years, see e.g.
\cite{Bal-16, PL2011,  BorMa2001, BorMa2002a, BorMa2002b,EhlersKoiller, FedorovJova-2004, LGN-JM-16, Jova-2003, Jova-2010, Kozlov-2002, Veselov-a-1988}.   The possibility of writing the reduced equations of motion in hamiltonian form is
central in the study  of various aspects of nonholonomic systems, such as integrability, Hamilton-Jacobi theory and even numerical methods (e.g. variational integrators), see \cite{Bol-Bor-Ma-11, Bol-Kil-Kaz-14, Cortes-2002, Cortes-Martinez-01, FGS-05, Hall-Leok-2015, Hall-Leok-2017, dL-MdD-SM-04, OhsawaBloch-09, Oscar-11}.

In the approach developed in \cite{PL2011,Naranjo2008}, the hamiltonization problem is translated into finding a suitable 2-form $B$ that is used to modify the nonholonomic bracket in a way that preserves its dynamical properties, while the corresponding reduced bracket is Poisson, or has at least an underlying almost symplectic foliation. In a recent paper \cite{LGN-JM-16}, the authors proved the existence of a 2-form $B$  through a local construction, so that certain types of first integrals become Casimirs of the modified reduced bracket. This article has two main goals. First, building on \cite{Bal-14} and under the dimension assumption \eqref{dimension_assumption}, we provide a more intrinsic and coordinate-free viewpoint to this 2-form,  proving  that the corresponding reduced bracket is a simple modification (or, more precisely, a ``gauge transformation'') of a genuine Poisson bracket with leaves determined by the level set of the first integrals. Our proofs are independent of those in \cite{LGN-JM-16}, with local formulas being recovered once appropriate coordinates are chosen. Second, we use our construction to give a new geometric proof of the hamiltonization of the problem of a ball rolling on a surface of revolution \cite{BKM-02, FGS-05, Hermans-95, Ramos-04, Routh-55, Zenkov-95}.

We now pass to a more detailed description of our results.  We start by recalling our geometric framework.



\noindent{\bf Nonholonomic systems and hamiltonization}.
Nonholonomic systems are mechanical systems with constraints in the velocities defining a nonintegrable distribution on a configuration manifold $Q$. Using the Legendre transform, we obtain a constraint submanifold $\M \subset T^*Q$ which inherits an almost Poisson bracket $ \{ \cdot , \cdot \}_\nh$, the {\it nonholonomic bracket}.
The nonholonomic dynamics is governed by the vector field $X_\nh = \{ \cdot , \Ham_\subM \}_\nh$ on $\M$, where $\Ham_\subM$ is the associated hamiltonian function (see Section~\ref{Sec:geometric_approach_nh} for details).
When the system admits symmetries given by a Lie group $G$, 
the reduced dynamics is given by the projection of $X_\nh$  to the quotient space $\M/G$, denoted by
$X_\red$. The bracket $\{\cdot, \cdot\}_\nh$ and hamiltonian $\Ham_\subM$ also descend to the quotient, giving rise to an almost Poisson manifold  $(\M/G, \{\cdot, \cdot\}_\red)$  and hamiltonian $\Ham_\red$ so that
\begin{equation}
\label{reduced_dynamics_intro}
X_\red = \{ \cdot, \Ham_\red \}_\red.
\end{equation}
When the action is not free, we endow the quotient $\M/G$ with a differential structure \cite{Cushman-Bates-15}.

While it is known that the nonholonomic bracket is never a Poisson bracket, since it has a nonintegrable characteristic distribution, the reduced bracket $\{\cdot, \cdot\}_\red$ may or may not be Poisson, depending on how the ``Jacobiator'' of $\{\cdot, \cdot\}_\nh$ --the cyclic sum that vanishes when the Jacobi identity holds-- interacts with the symmetries, see \cite{Bal-14}.
But even when $\{\cdot, \cdot\}_\red$ is not Poisson, it has been observed in many examples that 
there could be other brackets on $\M/G$ relating $X_\red$ and $\Ham_\red$ as in \eqref{reduced_dynamics_intro}, and 
which are Poisson (or close to being Poisson, e.g. conformally Poisson, or {\it twisted Poisson} \cite{KlimStro-2002,SeveraWeinstein}), see \cite{Bal-14, Bal-16,PL2011,BorMa2001,BorMa2002a,Naranjo2008,LGN-JM-16,Ramos-04}. 
The issue is then finding systematic ways to produce such brackets on $\M/G$.


The technique of using a 2-form $B$ to modify the nonholonomic bracket and produce new reduced brackets 
 first appeared in \cite{Naranjo2008}, and it was later formalized in \cite{PL2011} using the concept of {\it gauge transformation} by $2$-forms from \cite{SeveraWeinstein}. In the context of hamiltonization, one is interested in {\it dynamical gauge transformations}, which have the additional property that the modified bracket $\{\cdot, \cdot\}_\B$ still describes the nonholonomic dynamics, in the sense that $X_\nh = \{\cdot, \Ham_\subM\}_\B$, see \cite{PL2011}. In \cite{Bal-14,Bal-16,Bal-Fer,GNThesis,Naranjo2008}, and specially in \cite{LGN-JM-16}, it has been observed
that the existence of such a 2-form $B$ producing Poisson-type brackets on $\M/G$ is related to the presence of conserved quantities of the system called {\it horizontal gauge momenta} \cite{BGM-96,FGS-08,FGS-2012}.



\noindent{\bf Main result.}
In this paper, under the presence of suitable first integrals (horizonal gauge momenta), we present a coordinate-free formulation of a 2-form $B$ which defines a reduced bracket $\{\cdot, \cdot\}_\red^\B$ that describes the reduced nonholonomic dynamics and
admits an almost symplectic foliation, with leaves given by the common level sets of the first integrals; see Theorem \ref{T:GlobalB}. As we observe, these leaves actually form the symplectic foliation of a genuine Poisson bracket on $\M/G$, and $\{\cdot, \cdot\}_\red^\B$ arises as a gauge transformation of it.


The coordinate-free formulation of the 2-form $B$ relies on the choice of a suitable distribution $W$ complementing the constraint distribution $D$ in $TQ$ (as in \cite{Bal-14}) as well as a principal connection whose horizontal distribution lies in $D$. 
The 2-form $B$ is then written in terms of the 2-form $\langle J,\mathcal{K}_\subW\rangle$, already defined in \cite{Bal-14} to measure the failure of the Jacobi identity of the nonholonomic bracket (see Section \ref{section_splittings}), the first integrals,  the principal curvature and the kinetic energy metric; see \eqref{Eq:GlobalB}. We also verify that $B$ agrees with the local formulas of \cite{LGN-JM-16} in suitable coordinates.

It is insightful to write the 2-form $B$ as a sum
$$
B = B_1 + \mathcal{B},
$$
with each term playing a different role. First,
the gauge transformation of the nonholonomic bracket by $B_1$ produces, in the reduced space, a genuine Poisson bracket $\{\cdot , \cdot\}_\red^1$ whose symplectic leaves are the common level sets of the first integrals. However, this 2-form $B_1$ is not necessarily {\it dynamical}, that is, the induced bracket might not describe the nonholonomic dynamics as in \eqref{reduced_dynamics_intro}.  This is fixed by a second gauge transformation by $\mathcal{B}$; moreover, $\mathcal{B}$ is basic, so it induces
a 2-form $\bar{\mathcal B}$ on the quotient space $\M/G$ which gauge relates $\{ \cdot , \cdot \}_\red^1$ and $\{ \cdot , \cdot \}_\red^\B$. In particular, these brackets have the same characteristic distributions, their leafwise 2-forms differ by $\bar{\mathcal B}$, and $\{ \cdot , \cdot \}_\red^\B$ is a twisted Poisson bracket by the 3-form $(-d\bar{\mathcal B})$.

The following diagram illustrates these two steps:
\begin{equation}\label{Diagram}
   \xymatrix{ (\M, \{\cdot , \cdot\}_\nh, \Ham_\subM)  \ar@/^2pc/[rr]^{B_1+\mathcal{B} \mbox{\scriptsize{ (dynamical gauge)}} }  \ar[d]^{\rho}  \ar[r]^{B_1}  & (\M, \{\cdot , \cdot\}_1)  \ar[d]^{\rho} \ar[r]^{\mathcal{B}} & (\M, \{\cdot , \cdot\}_\B, \Ham_\subM)  \ar[d]^{\rho} \\
               (\M/G, \{\cdot , \cdot\}_\red, \Ham_\red)  & (\M/G, \{\cdot , \cdot\}_\red^1) \ar[r]^{\bar{\mathcal{B}}} & (\M/G, \{\cdot , \cdot\}_\red^\B, \Ham_\red) }
\end{equation}
In this diagram, we omit the hamiltonian functions $\Ham_\subM$ and $\Ham_\red$ in the middle column 
to indicate that these brackets do not necessarily describe the dynamics. If $\textup{rank}(TQ) - \textup{rank}(V) =1$ then  $\mathcal{B} = 0$ and thus we need only one step to hamiltonize the problem: $B= B_1$ and the resulting Poisson bracket $\{\cdot ,\cdot\}_\red^1$ (with 2-dimensional symplectic leaves) describes the dynamics.


Following our construction, we revisit three examples that were known to be hamiltonizable via gauge transformations by 2-forms.  We show how the previous diagram explains the hamiltonization procedure in each situation, explaining why we get a twisted Poisson bracket in the cases of the Chaplygin ball \cite{Naranjo2008,PL2011,Bal-14} and the snakeboard \cite{Bal-14,Bal-Fer}, and a Poisson bracket in the case of the solids of revolution rolling on a plane \cite{Bal-16,LGN-JM-16}. We then present a new example of hamiltonization via gauge transformations: a homogeneous ball rolling without sliding on a convex surface of revolution \cite{BKM-02, FGS-05, Hermans-95, Ramos-04, Routh-55, Zenkov-95}, showing that the dynamics is described by the Poisson bracket $\{\cdot ,\cdot\}_\red^1$.




\noindent{\bf Outline of the paper}.
In Section \ref{Sec:geometric_approach_nh} we recall the geometric set-up  of nonholonomic systems with symmetries.
In Section \ref{Sec:Gauge_transformations_conserved_quantities}, we define the 2-form $B$ and state the main results in Theorem \ref{T:GlobalB}. The proofs are presented in two steps (Prop.~\ref{L:B_1} and Prop.~\ref{Prop:gauge2}), following Diagram \eqref{Diagram}.   
In Section \ref{ExamplesRevisited} we illustrate the theory revisiting some examples. 
Finally, in Section \ref{Example_BallOnSurface}, we present the hamiltonization of the ball rolling on a surface of revolution.


\noindent {\bf Acknowledgements:} The authors would like to thank to Nicola Sansonetto for useful conversations and to Luis Garcia-Naranjo for his comments on this manuscript. P.B. thanks CNPq (Brasil) and L.P.Y. thanks Faperj (Brasil) for financial support. 



\section{The geometric approach to nonholonomic systems with symmetries}
\label{Sec:geometric_approach_nh}

In this section we present the geometric framework to write the nonholonomic equations of motion before and after reduction. We follow the previous works in \cite{Bal-14, PL2011, BS-93, BKMM-96, Naranjo2008, IdLMM, SchaftMaschke1994}.

\subsection{The nonholonomic bracket}
\label{Sec:Preliminaries_Nonholonomic}

Consider a mechanical system on the $n$-dimensional configuration manifold $Q$ with a Lagrangian $L:TQ\to \R$ of mechanical type, i.e., $L$ is of the form $L = \frac{1}{2} \kappa - \tau^*_{TQ} U$, where $\kappa$ is the kinetic energy metric, $U: Q \rightarrow \mathbb{R}$ the potential energy and $\tau_{TQ}:TQ\to Q$ the canonical projection.

Let $D$ be a (constant rank) non-integrable distribution on $Q$ representing the nonholonomic constraints, that is, at the configuration point $q \in Q$ the permitted velocities belong to a linear subspace $D_q$ of $T_qQ$.
Under the assumption of a mechanical-type Lagrangian, the Legendre transform $Leg = \kappa^\flat:TQ \to T^*Q$ is a global diffeomorphism, where $\kappa^\flat (X)(Y) = \kappa (X, Y)$ for $X, Y\in \mathfrak{X}(Q)$. Hence, following \cite{BS-93}, the distribution $D$ on $Q$ induces a submanifold $\M$ of $T^*Q$,
\begin{equation}
\label{definition_M}
\M := Leg(D) \subset T^*Q,
\end{equation}
called the \textit{constraint manifold}. Since $Leg:TQ\to T^*Q$ is linear on the fibers, $\M\to Q$ is also a subbundle of $\tau_{T^*Q}: T^*Q \rightarrow Q$ of rank $n-k$. 
Let us denote by $\iota_\subM : \M \hookrightarrow T^*Q$ the inclusion and by $\tau_\subM : \M \rightarrow Q$ the canonical projection.

The distribution $D$ on $Q$ also induces a non-integrable (constant rank) distribution $\C$ on $\M$ given, at a point $m \in \M$, by
\begin{equation}
	\label{definition_C}
	\C_m = \{ v_m \in T_m \M \  | \ T \tau_\subM(v_m) \in D_{\tau_\subM(m)} \}.
\end{equation} 

We denote by $\Omega_\subM$ the pull back to $\M$ of the canonical $2$-form $\Omega_Q$ in $T^*Q$ and by $\Omega_\subC$ the point-wise restriction of $\Omega_\subM$ to $\C$, i.e.
\begin{equation}
\label{def:Omega_M_and_Omega_C}
\Omega_\subM := \iota^*_\subM \Omega_Q \quad \text{and} \quad\Omega_\subC:= \Omega_\subM |_\C. 
\end{equation}
Since the $2$-section $\Omega_\subC$ is nondegenerate (see \cite{BS-93}), it induces the so-called {\it nonholonomic bracket} $\{\cdot, \cdot\}_\nh$ on $\M$ (see \cite{IdLMM, Marle, SchaftMaschke1994}) given, at each $f,g \in C^\infty(\M)$, by
\begin{equation}
\label{definition_nh_bracket}
 \{f , g\}_\nh = - X_f(g),
\end{equation}
where $X_f$ is the (unique) vector field on $\M$ taking values in $\C$  so that $\mathbf{i}_{X_f} \Omega_\subC = df |_\C$.  In particular, $X_f$ is a section of the bundle $\C\to \M$ and we denote it by $X_f\in \Gamma(\C)$.  The nonholonomic bracket is an almost Poisson bracket, that is, it is bilinear, skew-symmetric, it verifies the Leibniz property but does not necessarily satisfy the Jacobi identity \cite{CdLMD}. Moreover, the distribution on $\M$ generated by the hamiltonian vector fields $X_f$ --called {\it characteristic distribution}-- is the non-integrable distribution $\C$ in \eqref{definition_C}. 

The nonholonomic bracket defines a bivector field $\pi_\nh$ on $\M$ by: 
\begin{equation}
\label{definition_pi_nh}
\pi_\nh(df,dg) = \{f,g\}_\nh, \qquad \mbox{for } f,g\in C^\infty(\M).
\end{equation}
We denote by $\pi_\nh^\sharp:T^*\M \to T\M$ the map given by $\beta(\pi_\nh^\sharp(\alpha) )= \pi_\nh(\alpha,\beta)$ for each $\alpha,\beta \in T^*\M$.  

The dynamics of the nonholonomic system is given by the integral curves of the {\it nonholonomic vector field} $X_\nh$ on $\M$, that takes values on $\C$, and is defined by 
\begin{equation}
\label{nonholonomic_dynamics}
{\bf i}_{X_\nh} \Omega_\subC = d \Ham_\subM |_\C ,
\end{equation}
where the function $\Ham_\subM \in C^\infty (\M)$ is the restriction to $\M$ of the Hamiltonian function $H:T^*Q\to \R$ induced by the Lagrangian $L:TQ\to \R$, \cite{BS-93}. Equivalently, the vector field $X_\nh$ is given by 
$$
X_\nh =  - \pi_{\nh}^\sharp (d \Ham_\subM) = \{\cdot , \Ham_\subM\}_\nh.
$$

We say that a nonholonomic system is {\it described by} the triple $(\M, \pi_\nh, \Ham_\subM)$, or that the nonholonomic dynamics is {\it described by} the bivector field $\pi_\nh$.

\subsection{Symmetries and reduction for proper actions}
\label{symmetries_reduction_proper_actions}

Consider a $m$-dimensional connected Lie group $G$ acting {\it properly} on the manifold $Q$ such that the tangent lift of the action leaves invariant the Lagrangian $L$ and the distribution $D$. As a consequence, the corresponding cotangent lift of the action leaves invariant the constraint manifold $\M \subset T^*Q$ and the restricted hamiltonian $\Ham_\subM$. Considering the $G$-action on $\M$, the reduced nonholonomic dynamics is determined by the integral curves of the vector field $X_\red$ on $\M/G$ given by $X_\red = T\rho (X_\nh)$ where $\rho:\M \to \M/G$ is the orbit projection.  In this section, we study the geometry behind this reduction process emphasizing the fact that the $G$-action is not necessarily free.   
In particular, the quotient space $\M / G$ is not necessarily a smooth manifold but it can be understood as a {\it stratified differential space}, see for instance \cite{Cushman-Bates-15, CDS-2010, Sniatycki-13}. Indeed, firstly, the quotient space $ \M / G$ is a topological space with the quotient topology and the properness of the action implies that $\M / G$ is Hausdorff. Secondly, the space $ \M / G$ can be endowed with a differential space structure by declaring the ring of smooth functions on $\M / G$ as being the ring of smooth $G$-invariant functions on $\M$, denoted by $C^\infty(\M)^G$. That is, 
a smooth function $\bar{f} \in C^\infty(\M /G)$ is identified with the corresponding $G$-invariant function $f$ on $\M$ such that $\rho^* \bar{f} = f$.

The $G$-action on $\M$ leaves also the nonholonomic bracket $\{ \cdot, \cdot \}_\nh$ invariant and thus we say that the nonholonomic system $( \M , \pi_\nh , \Ham_\subM )$ admits a $G$-{\it symmetry}. Hence, the differential space inherits an almost Poisson structure given, at each $\bar{f}, \bar{g} \in C^\infty(\M / G)$, by
   \begin{equation}
   \label{definition_reduced_bracket}
      \{ \bar{f} , \bar{g} \}_\red \circ \rho = \{ f,g \}_\nh,
   \end{equation}
where $f= \rho^* \bar{f}$ and $g= \rho^* \bar{g}$ belong to $C^\infty(\M)^G$.

Since the Hamiltonian $\Ham_\subM$ is $G$-invariant, the reduced bracket $\{ \cdot , \cdot \}_\red$ describes the reduced dynamics:   
\begin{equation}\label{Eq:ReducedDynamics}
X_\red = \{ \cdot , \Ham_\red \}_\red,
\end{equation}
where $\Ham_\red : \M / G \rightarrow \R$ is the reduced Hamiltonian. 

\begin{remark}
For a proper $G$-action, the space $\M / G$ has more properties in addition of being a differential space. 
First, it is a {\it subcartesian space}, meaning that it is a Hausdorff differential space, locally diffeomorphic (as a differential space) to an open subset of the Cartesian space $\R^n$. On the other hand $\M / G$ is a {\it stratified space} given by the {\it orbit type stratification} associated to the $G$-action. 
A derivation $X$ on a subcartesian space is called a {\it vector field} if the unique maximal integral curve passing by a given point induces a local one-parameter group of local diffeomorphisms.
Sniatycki \cite{Sniatycki-03} has
shown that the orbits of the family of {\it all} vector fields on $\M / G$ coincide with the strata of the orbit type stratification of $\M / G$.
In the case of a nonholonomic system $(\M , \pi_\nh, \Ham_\subM )$ with a proper $G$-symmetry, the space $\M / G$ is an almost Poisson differential space with bracket $\{ \cdot, \cdot \}_\red$. It is stratified by orbit type, where each stratum is an almost Poisson manifold and $X_\red$ is a vector field in $\M / G$ inducing  smooth vector fields on each strata and preserving the orbit type stratification of $\M / G$, see Chapter 8 of \cite{Sniatycki-13}. For more details on the structure of the quotient $\M / G$ see the books \cite{Cushman-Bates-15, CDS-2010, DuisKo2000:Book, Sniatycki-13} and the paper \cite{Sniatycki-03}.
\end{remark}

We are interested in the integrability properties of the reduced bracket $\{ \cdot, \cdot \}_\red$ in $\M / G$, which are cast by the failure of the Jacobi identity. In order to compute the Jacobiator,  we use the formulas proven in \cite{Bal-14} which are based on certain splittings of the tangent bundle $TQ$ explained in the next section.

\subsection{The failure of the Jacobi identity}
\label{section_splittings}

Following Section \ref{symmetries_reduction_proper_actions}, let us consider a nonholonomic system $(\M , \pi_\nh, \Ham_\subM)$ with a  $G$-symmetry. Let us denote by $V$ the (generalized) smooth distribution on $Q$, called {\it vertical distribution}, given, at each $q \in Q$,  by the tangent space to the orbit by $G$ passing through the point $q$, i.e. $V_q := T_q Orb(q)$. In this paper we suppose that the \textit{dimension assumption} \cite{BKMM-96} holds, i.e., at each $q \in Q$,
\begin{equation}
\label{dimension_assumption}
T_q Q = D_q + V_q.
\end{equation}

Since the action is not necessarily free, the rank of the vertical distribution $V$ may vary.
Now let us define the (generalized) smooth distribution $S$ in $Q$ given, at each $q \in Q$, by
\begin{equation}
\label{S_definition}
S_q := D_q \cap V_q.
\end{equation}

It was shown in \cite[Prop.~2.2]{Bal-16} that the dimension assumption implies the existence of a constant rank smooth distribution $W$ on $Q$ such that, for all $q \in Q$,
\begin{equation}
\label{splitting_TQ_DW}
T_q Q = D_q \oplus W_q \quad \mbox{and} \quad W_q \subset V_q,
\end{equation}
which is equivalent to the following splitting of the vertical distribution: $V_q = S_q \oplus W_q$.
In fact, let $\g$ denote the ($m$-dimensional) Lie algebra of $G$. Following  \cite{BKMM-96}, for each $q \in Q$, let us define the vector subspace
$(\g_S)_q$ of $\g$ given by
\begin{equation}
\label{defbundle_gS}
\zeta_q \in (\g_S)_q \Leftrightarrow (\zeta_q)_Q (q) \in S_q, 
\end{equation}
where $(\zeta_q)_Q$ denotes the infinitesimal generator of the $G$-action on $Q$ associated to the element $(\zeta_q) \in \g$.
First, it was shown in \cite[Prop.~2.2]{Bal-16} that the dimension assumption implies that $\g_S \rightarrow Q$ is a vector subbundle of the trivial bundle $\g \times Q \rightarrow Q$.  
Then a section $\zeta$ of the bundle $\g\times Q \to Q$ is a section of the subbundle $\g_S\to Q$ if $\zeta_Q \in \Gamma(S)$, where for $q\in Q$, $\zeta_Q (q) := (\zeta |_q)_Q (q)$.
Second, observe that the bundle $\g_S \rightarrow Q$ admits a bundle complement $\g_W \rightarrow Q$ such that, for $q \in Q$,
\begin{equation}
\label{splitting_g_times_Q}
(\g \times Q)_q = (\g_S)_q \oplus (\g_W)_q.
\end{equation}
Hence, a distribution $W \subset V$ as in (\ref{splitting_TQ_DW}) is induced by the choice of such a subbundle $\g_W \to Q$, so that
\begin{equation}
\label{definition_W_q}
W_q = \textup{span} \{ (\xi_q)_Q (q) : (\xi_q) \in (\g_W)_q \},
\end{equation}
and it has constant rank and is smooth. 

Next, we will see that the bundle $\g_W\to Q$ in \eqref{splitting_g_times_Q} can be chosen so that the distribution $W$ is $G$-invariant.  First, observe that if the $G$-action is free and proper, then $W$ can be defined so that $W= S^\perp \cap V$ where $S^\perp$ is the orthogonal complement of $S$ with respect to the kinetic energy metric.  In this case, $W$ is automatically constant rank, smooth and $G$-invariant.


\begin{proposition}
\label{prop:W_G_invariant}  
Let $G$ be a Lie group so that its Lie algebra $\g$ admits an $Ad$-invariant scalar product. If the dimension assumption \eqref{dimension_assumption} is satisfied, then there exists a $G$-invariant complement $W \subset V$ verifying \eqref{splitting_TQ_DW} (or equivalently
there exists an $Ad$-invariant subbundle $\g_W \to Q$ of $\g \times Q \to Q$ verifying 
\eqref{splitting_g_times_Q}).
\end{proposition}


\begin{proof}
For any $g \in G$, let us consider the bundle map $Ad_g : \g \times Q \rightarrow \g \times Q$, over the action diffeomorphism $\psi_g : Q \rightarrow Q$, given at each $\eta \in \Gamma(\g \times Q)$ by $Ad_g (\eta) |_{\psi_g(q)} := Ad_g (\eta |_q)$.
By the dimension assumption $\g_S \to Q$ is a subbundle of the trivial bundle $\g \times Q \to Q$ \cite{Bal-16}, and 
since the distribution $S$ is $G$-invariant then $\g_S$ is $Ad$-invariant and thus $Ad_g$ restricts to a bundle map $Ad_g : \g_S \rightarrow \g_S$, see \cite[Lemma.~4.4.8]{Cortes-2002}. By $Ad$-invariance of the bundle metric on $\g \times Q \to Q$, if $\g_W$ in (\ref{splitting_g_times_Q}) is chosen to be the orthogonal complement of the bundle $\g_S$ with respect to this metric, then $\g_W \rightarrow Q$ is an $Ad$-invariant subbundle of $\g \times Q \rightarrow Q$. Defining the distribution $W$ as in (\ref{definition_W_q}), we obtain that it is $G$-invariant.
\end{proof}

\begin{remark}
Recall that if $G$ is connected and  isomorphic to the Cartesian product of a compact group and a vector space, then the Lie algebra $\g$ admits an $Ad$-invariant scalar product.
Therefore, Prop.~\ref{prop:W_G_invariant} can be used in examples in Sec.~\ref{ExamplesRevisited} and Sec.~\ref{Example_BallOnSurface}.
\end{remark}

Following \cite{Bal-14}, the splitting (\ref{splitting_TQ_DW}) on $TQ$ induces a splitting on $T\M$. More precisely, the $G$-action on $\M$ defines a (generalized) vertical distribution $\V$ on $\M$ given by $\V_m := T_m Orb(m)$, for $m \in \M$. Then the distribution $\W$ given, at each $m\in\M$, by
\begin{equation}
\label{distribution_W_in_M}
\W_m = \textup{span} \{ (\xi_q)_\subM (m) : (\xi_q) \in (\g_W)_q , \: q = \tau_\subM(m) \},
\end{equation}
is $G$-invariant as long as $\g_W \to Q$ is $Ad$-invariant, as in Prop.~\ref{prop:W_G_invariant}. Moreover, observe that $\C\cap \W =\{0\}$ since, if there exists a non-zero element $\xi_q\in (\g_W)_q$ so that $(\xi_q)_\subM (m) \in \C_m$, for $\tau_\subM(m)=q$,  then $(\xi_q)_Q(q) \in D_q\cap W_q$ which contradicts \eqref{splitting_TQ_DW}. Therefore, we obtain the following definition:

\begin{definition} \label{Def:W}
 The distribution $W$ on $Q$ is called a {\it vertical complement of the constraints} $D$ if $W$ is $G$-invariant and satisfies that $TQ = D\oplus W$ with $W\subset V$.
In this case, the distribution $\W$ on $\M$ defined in \eqref{distribution_W_in_M} is a {\it vertical complement of the constraints} $\C$ since 
\begin{equation}
\label{splitting_TM_CW}
T \M = \C \oplus \W \quad \text{and} \quad \W \subset \V .
\end{equation}
\end{definition}

Following \eqref{S_definition}, we define the (generalized) distribution $\S$ on $\M$ given, at each point $m \in \M$, by
\begin{equation}
\label{def:S_script}
\S_m := \C_m \cap \V_m,
\end{equation}
and we see that $\V_m = \S_m \oplus \W_m$.

Now, we recall from \cite{Bal-14} the formulas that cast the failure of the Jacobi identity of the nonholonomic bracket $\{ \cdot, \cdot \}_\nh$ and the reduced bracket $\{ \cdot, \cdot \}_\red$.  First, we define the so-called {\it $\W$-curvature} and study some properties to then write down the formulas of the Jacobiator.

A  vertical complement $W$ of the constraints $D$ induces a $\g$-valued 1-form $A_W$ on $Q$ given, at each $X \in \mathfrak{X}(Q)$, by
\begin{equation}
\label{definition_A_W}
A_W(X) = \xi \in \Gamma(Q\times \g) \quad \mbox{if and only if} \quad P_W(X) = \xi_Q,
\end{equation}
where $P_W : TQ \rightarrow W$ is the projection on the second factor in the splitting (\ref{splitting_TQ_DW}). Analogously we denote by $P_D : TQ \rightarrow D$ 
the projection onto the first factor in the splitting (\ref{splitting_TQ_DW}) and we define the $\g$-valued $2$-form $K_W$ on $Q$ given by
$$
K_W (X,Y) = d A_W (P_D (X), P_D(Y)), \quad X,Y \in \mathfrak{X}(Q).
$$
Using the natural projection $\tau_\subM :\M\to Q$, we denote by $\mathcal{A}_\subW$ the $\g$-valued 1-form on $\M$ given by $\mathcal{A}_\subW := \tau_\subM^* A_W$ and observe that $\textup{Ker}\mathcal{A}_\subW = \C$.
\begin{definition}\label{Def:definition_K_W}\cite{Bal-14}
The {\it $\W$-curvature} ${\mathcal K}_\subW$ is the $\g$-valued $2$-form on $\M$ given by $\mathcal{K}_\subW :=  \tau_\subM^* K_W$, that is, 
$$
\mathcal{K}_\subW (X,Y):=   d\mathcal{A}_\subW(P_\C(X), P_\C(Y)),
$$
where $P_\C: T\M\to \C$ the projection associated to the splitting \eqref{splitting_TM_CW}.
\end{definition}

It will be useful to set the following notation.  For a $k$-form $\alpha$ on $\M$, we denote by $d^\C$ the differential of $\alpha$ on elements of $\C$, that is, at each $ X_1, ..., X_{k+1} \in \mathfrak{X}(\M)$,
\begin{equation}\label{Eq:dC}
d^\C\alpha(X_1,..., X_{k+1}) := d\alpha(P_\C(X_1),..., P_\C( X_{k+1})).
\end{equation}
  In other words,  $d^\C\alpha|_\C = d\alpha|_\C$ while ${\bf i}_Zd^\C\alpha = 0$ for all $Z\in \W$.  Observe that $\Omega_\C = - d^\C\Theta_\subM |_\C$ and $\mathcal{K}_\subW = d^\C \mathcal{A}_\subW$.

\begin{lemma}\label{L:dK=0}
 The $\g$-valued 2-form ${\mathcal K}_\subW$ satisfies that $d^\C \mathcal{K}_\subW=0$.
\end{lemma}

\begin{proof}
 Following \eqref{splitting_g_times_Q}, let us consider a basis of sections $\{ \zeta_i, \xi_a \}$ of $\g \times Q \rightarrow Q$ where $\zeta_i \in \Gamma(\g_S)$ and $\xi_a \in \Gamma(\g_W)$, for $i=1,..., l$; $a = l+1,...,m$ and $\textup{dim} (\g) = m$. 
For each $a=l+1,..., m$, let us define the 1-forms $\epsilon^a$ on $Q$ such that $\epsilon^a((\xi_b)_Q) = \delta_b^a$ (where $\delta_b^a = 1$ if $a=b$ and $\delta_b^a=0$ if $a\neq b$) and $\epsilon^a|_D = 0$ for all $a$.  Then let us denote $\tilde{\epsilon}^a := \tau_\subM^* \epsilon^a$, where as usual $\tau_\subM : \M \rightarrow Q$ is the canonical projection.
Now consider a basis of $\g$  given by $\{ \chi_1, ... , \chi_m \}$ (i.e., $\chi_I$ are constant sections of the bundle $Q\times \g\to Q$) and thus
\begin{equation}
\label{basis_xi_eta_chi}
	\zeta_i = h_i^I \chi_I, \quad \xi_a = h_a^I \chi_I, \quad \mbox{for } I=1,...,m, 
\end{equation}
for functions $h_i^I$, $h_a^I \in C^\infty(Q)$.  Therefore, the $\g$-valued $1$-form $\mathcal{A}_\subW$ is written as 
$\mathcal{A}_\subW = \tilde \epsilon^a \otimes \xi_a = h_a^I \tilde \epsilon^a \otimes \chi_I$ and thus $\mathcal{K}_\subW = d^\C(h_a^I \tilde \epsilon^a) \otimes \chi_I$. Let us denote by $\mathcal{A}_{\subW}^I = h_a^I \tilde \epsilon^a$ and by $\mathcal{K}_{\subW}^I =d^\C \mathcal{A}_\subW^I = d^\C( h_a^I \tilde\epsilon^a)$.  In order to show that $d\mathcal{K}_\subW|_\C = 0$ it is enough to prove that $d\mathcal{K}_{\subW}^I|_\C = 0$ for all $I=1,...,m$. 
Let us consider $X_1, X_2, X_3\in \Gamma(\C)$ and then
\begin{equation*}
 \begin{split}
  d\mathcal{K}_{\subW}^I(X_1, X_2, X_3) & = cyclic [ X_1(\mathcal{K}_{\subW}^I(X_2,X_3)) -\mathcal{K}_{\subW}^I([X_1, X_2],X_3)  ] \\
  & = cyclic [ d (d\mathcal{A}_{\subW}^I)(X_1,X_2,X_3) + d\mathcal{A}_{\subW}^I(P_\subW[X_1, X_2],X_3)  ] = cyclic [d\mathcal{A}_{\subW}^I(P_\subW[X_1, X_2],X_3)  ] 
 \end{split}
\end{equation*}
where $cyclic$ indicates cyclic permutations of the parameters.  Next, we show that $d\mathcal{A}_{\subW}^I(Z, X)=0$ for all $Z\in \Gamma(\W)$ and $X\in \Gamma(\C)$.  First, observe that, $[(\chi_I)_\subM, X] \in \Gamma(\C)$ by the $G$-invariance of $\C$. Now, denoting  $Z_b = (\xi_b)_\subM$ we get
\begin{equation*}
\begin{split}
  d\mathcal{A}_{\subW}^I(Z_b, X) & =  - X(\mathcal{A}_{\subW}^I(Z_b) )  -\mathcal{A}_{\subW}^I([Z_b,X])  = - X(h_b^I)- h_a^I\tilde \epsilon^a([h_b^J (\chi_J)_\subM, X] ) \\ & = - X(h_b^I) + h_a^I X(h_b^J) \tilde \epsilon^a((\chi_J)_\subM)  =  - X(h_b^I) + h_a^I X(h_b^J) \bar{h}^a_J = 0,
 \end{split}
\end{equation*}
where $\bar{h}_J^a$ are the functions such that $\chi_J = \bar{h}_J^a \xi_a + \bar{h}_J^i \zeta_i$.  Then, we obtain that $d^\C \mathcal{K}_{\subW}^I = 0$ showing that $d^\C \mathcal{K}_\subW=0$. 
\end{proof}

On the other hand,  let $J : \M \rightarrow \g^*$ be the restriction to $\M$ of the canonical momentum map on $T^* Q$, i.e., for all $m\in \M$,  $\langle J(m), \eta \rangle = \mathbf{i}_{\eta_{\M}(m)}\Theta_\subM (m)$, where $\eta \in \g$ and $\Theta_\subM$ is the Liouville $1$-form restricted to $\M$. 
Then, following \cite{Bal-14},  we define the $G$-invariant 2-form on $\M$ given by 
\begin{equation}\label{Def:JK}
\langle J, \mathcal{K}_\subW\rangle
\end{equation}
where $\langle \cdot, \cdot \rangle$ denotes the pairing between $\g^*$ and $\g$.  By Lemma \ref{L:dK=0}, we also have the 3-form on $\M$ given, on $X,Y,Z \in \Gamma(\C),$ by 
\begin{equation*}
d\langle J, \mathcal{K}_\subW \rangle (X,Y,Z) = dJ \wedge {\mathcal K}_\subW (X,Y,Z) := cyclic [ \langle dJ(X), \mathcal{K}_\subW(Y,Z) \rangle ],
\end{equation*}
(see \cite{Bal-14} and \cite{Bal-15} for more details). In other words,  $d^\C\langle J, \mathcal{K}_\subW \rangle  = d^\C J \wedge {\mathcal K}_\subW$.

\medskip

\begin{proposition}\textup{\cite{Bal-14}} \label{Prop:Jac} The almost Poisson manifold $(\M, \{ \cdot , \cdot \}_{\emph\nh})$ associated to a nonholonomic system with symmetries satisfying the dimension assumption \eqref{dimension_assumption} verifies that, for all $f, g, h \in C^\infty(\M)$,
\begin{equation*}
cyclic[\{ f ,\{ g , h \}_{\emph\nh} \}_{\emph\nh}] = d\langle J , \mathcal{K}_\subW\rangle(\pi_{\emph\nh}^\sharp (d f), \pi_{\emph\nh}^\sharp (d g), \pi_{\emph\nh}^\sharp (d h )) - \psi_{\pi_{\emph\nh}} (df,dg,dh), 
\end{equation*}
where 
$\psi_{\pi_{\emph\nh}}$ is the trivector given by $\psi_{\pi_{\emph\nh}} (\alpha,\beta,\gamma) = cyclic \left[ \gamma \left(  ( \mathcal{K}_\subW(\pi_{\emph\nh}^\sharp (\alpha),\pi_{\emph\nh}^\sharp (\beta))\, )_\subM \right) \right] $, for $\alpha$, $\beta$, $\gamma$ $1$-forms on $\M$.
Moreover, the induced reduced nonholonomic bracket $\{ \cdot, \cdot \}_{\emph\red}$ on $\M / G$ satisfies, for all $\bar{f}, \bar{g}, \bar{h} \in C^\infty(\M / G)$,
$$
cyclic [ \{ \bar{f} ,\{ \bar{g} ,\bar{h} \}_{\emph\red} \}_{\emph\red} \circ \rho ] = d\langle J , {\mathcal K}_\subW\rangle (\pi_{\emph\red}^\sharp (d\rho^* \bar{f}), \pi_{\emph\red}^\sharp (d \rho ^* \bar{g}), \pi_{\emph\red}^\sharp (d \rho^* \bar{h})).
$$
\end{proposition}

\medskip

We observe that Prop.~\ref{Prop:Jac} also holds in the case where the quotient $\M/G$ is a differential space.

It is also important to note here that the reduced bracket $\{\cdot, \cdot\}_\red$ is not necessarily Poisson. However, in some particular situations, it could become Poisson even when $\pi_\nh$ is not: if
\begin{equation}\label{Eq:Jacnh=0}
d\langle J, \mathcal{K}_\subW\rangle |_{\mathcal U} = 0,
\end{equation}
where $\mathcal{U}$ is the distribution on $\M$ given by 
$$
\mathcal{U}=\textup{span} \{\pi_\nh^\sharp(df)  \ : \ f \in C^\infty(\M)^G \},
$$
then the reduced bracket $\{\cdot, \cdot\}_\red$ on $\M/G$ is Poisson.

Recall that a form $\alpha$ on $\M$ is a {\it semi-basic} form with res\-pect to $\rho: \M\to \M/G$ if ${\bf i}_{X}\alpha = 0$ when $X\in T\M$ such that $T\rho (X) = 0$.  In other words, a section in the distribution $\mathcal{U}$ can be written as $\pi_\nh^\sharp(\alpha)$ where $\alpha$ is a semi-basic 1-form on $\M$.  
On the other hand, we may consider the almost Poisson manifold $(\M_{reg}/G, \{\cdot, \cdot\}_\red)$ where $\M_{reg}$ denotes the submanifold of $\M$ where the $G$-action is free. 
In this case, the distribution $\mathcal{U}$ on $\M_{reg}$ was originally defined in \cite{BS-93} from where it can be concluded that the characteristic distribution $\C_\red$ of $\{\cdot, \cdot \}_\red$ is given by $\C_\red = T\rho({\mathcal U})$.  
If $\{\cdot, \cdot\}_\red$ is a Poisson bracket on $\M_{reg}/G$ then it carries a symplectic foliation as it is the case, for example, of the nonholonomic particle (in this case, $\M_{reg}=\M$, see \cite{Bal-14,BS-93}).    

However, there is a class of brackets defined on manifolds whose characteristic distribution is integrable but the foliation carries an almost symplectic structure on each leaf.  These are the so-called {\it twisted Poisson structures} that were defined in \cite{KlimStro-2002,SeveraWeinstein}. We are interested in this class of brackets since, as we will see, they can appear in the reduction process and they have stronger integrability properties than a general almost Poisson bracket.  

\begin{definition}\label{Def:twisted} \cite{KlimStro-2002, SeveraWeinstein}
Given an almost Poisson bracket $\{\cdot, \cdot\}$ on a manifold $M$, we say that $\{\cdot, \cdot\}$ is $\phi${\it-twisted Poisson} if there is closed 3-form $\phi$ on $M$ such that, for all $f,g,h\in C^\infty(M)$,  
$$
cyclic [ \{ f ,\{g ,h \} \} ] = \phi (\pi^\sharp (df), \pi^\sharp (dg), \pi^\sharp (dh)),
$$
where $\pi$ is the bivector field on $M$ associated to $\{\cdot, \cdot\}$.  Equivalently,  $\frac{1}{2}[\pi,\pi] = \pi^\sharp(\phi)$ where $[\cdot, \cdot]$ here represents the Schouten bracket.
\end{definition}

A twisted Poisson bracket has an integrable characteristic distribution endowed with an almost symplectic foliation. 

\noindent {\bf Examples.}
\begin{enumerate}
 \item[$(i)$] If $\Omega$ is a nondegenerate 2-form on $M$, then the almost Poisson bracket given, at each $f,g\in C^\infty(M)$, by $\{f, g\} = \Omega(X_f,X_g)$ is a $(-d\Omega)$-twisted Poisson bracket.  
 \item[$(ii)$] If an almost Poisson bivector field $\pi$ on $M$ has an almost symplectic foliation $(L_\mu, \Omega_\mu)$ associated, then $\pi$ is $\phi$-twisted Poisson for $\phi$ any 3-form that coincides with $-d\Omega_\mu$ on the leaves.
 \item[$(iii)$] A ``gauge transformation'' of a Poisson bracket $\pi$ by a non-closed 2-form $B$ gives rise to a twisted Poisson bracket with $\phi = (-dB)$ (see Sec.~\ref{Sec:Gauge_transformations_conserved_quantities}, in particular Rmk.~\ref{R:Gauge}$(i)$). 
\end{enumerate}

Returning to our framework, we see that $\pi_\nh$ is not a twisted Poisson bracket since its characteristic distribution $\C$ is not integrable. However, there are some cases in which, after a reduction by symmetries, $\pi_\red$ becomes twisted Poisson on $\M_{reg}/G$, for example the case of the snakeboard (where we also have $\M_{reg} = \M$, see \cite[Sec.~5]{Bal-14} and Sec.\ref{Example_Snake}  ).  What occurs in these cases is that the 3-form $d\langle J, \mathcal{K}_\subW\rangle$ on $\M_{reg}$ is the pull-back of a well defined 3-form $\phi$ on $\M_{reg}/G$ and thus the failure of the Jacobi identity in Prop.~\ref{Prop:Jac} becomes 
\begin{equation*}
cyclic [ \{ \bar{f} ,\{ \bar{g} ,\bar{h} \}_\red \}_\red  ] = \phi (\pi_\red^\sharp (d\bar{f}), \pi_\red^\sharp (d \bar{g}), \pi_\red^\sharp (d \bar{h})).
\end{equation*}

We remark that finding a twisted Poisson bracket on the reduced space enlightens the process of searching for a conformal factor for the bracket.  The characterization of the almost symplectic foliations allows us to compute a conformal factor on each leaf and then, under suitable regularity conditions, to extend the conformal factor to the entire foliation \cite{Bal-Fer} (for the study of conformal factors on almost symplectic manifolds, see also \cite{Chaplygin,LGN-Ma-18,LGN-18}.

However, in general, the reduced nonholonomic bracket $\{\cdot, \cdot\}_\red$ does not admit an almost symplectic foliation, i.e., it is not twisted Poisson.  In this paper, we state a formula to produce a new bracket on the reduced space that describes the dynamics \eqref{Eq:ReducedDynamics} and that is twisted Poisson with its almost symplectic foliation given by the level sets (of certain types) of first integrals of the nonholonomic vector field $X_\nh$.

\bigskip


\section{Gauge transformations and conserved quantities}
\label{Sec:Gauge_transformations_conserved_quantities}

Consider a nonholonomic system $(\M,\pi_\nh,\Ham_\subM)$ with a proper $G$-symmetry satisfying the dimension assumption \eqref{dimension_assumption}.  The question we address in this paper is whether there is a Poisson (or twisted Poisson) bracket on $\M/G$ describing the reduced dynamics $X_\red$ as in \eqref{Eq:ReducedDynamics}.  
Following \cite{Bal-14, PL2011, Naranjo2008}, we use {\it gauge transformations by 2-forms} of the nonholonomic bivector $\pi_\nh$ in order to generate a new almost Poisson bivector $\pi_\B$ on $\M$ describing the dynamics so that the reduced bracket $\{\cdot, \cdot \}_\red^\B$ has the desired properties. 

\subsection{Gauge transformations}
\label{GaugeTransformations}

We start by recalling the concept of a {\it gauge transformation of $\pi_{\emph\nh}$ by a $2$-form}, which is a ``deformation" of the bivector field $\pi_\nh$ that generates a new bivector $\pi_\B$.

Recall that a regular almost Poisson manifold $(M,\pi)$ is determined by a distribution $F$ on $M$ (its characteristic distribution) and a nondegenerate $2$-section $\Omega$ on $F$ such that
\begin{equation*}
\pi^\sharp (\alpha) = -X \Leftrightarrow \mathbf{i}_X \Omega |_F = \alpha |_F , \qquad \mbox{for } \alpha \in T^*M.
\end{equation*}

Let us now consider a $2$-form $B$ on $M$ such that the $2$-section $(\Omega + B)|_F$ is nondegenerate. The {\it gauge transformation of $\pi$ by the $2$-form $B$} \cite{SeveraWeinstein} induces a new bivector field $\pi_\B$ on $M$ defined by 
\begin{equation}
\label{def:pi_B}
\pi_\B^\sharp (\alpha) = -X \Leftrightarrow \mathbf{i}_X (\Omega + B ) |_F = \alpha |_F , \qquad \mbox{for } \alpha \in T^*M.
\end{equation}

So, the new bivector $\pi_\B$ is defined by the same distribution $F$ on $M$ and the $2$-section $(\Omega + B)|_F$ and we say that $\pi$ and $\pi_\B$ are {\it gauge related}.

\begin{remark}\label{R:Gauge}

\begin{enumerate}
  \item[$(i)$] If $\pi$ is a Poisson (or twisted Poisson) bracket, then a gauge transformation by $B$ produces a new bracket with the same characteristic foliation but now the 2-form on each leaf will have the additional term $B$ --restricted to the leaf--.  In particular, if $\pi$ is Poisson, then the gauge transformation of $\pi$ by a 2-form $B$ produces a $(-dB)$-twisted Poisson bivector field $\pi_\B$. 
  \item[$(ii)$]  If $\pi_1$ and $\pi_2$ are two bivector fields on $\M$ with different characteristic distributions, then there is no 2-form $B$ that relates them. 
  \item[$(iii)$] If $\pi$ is a $G$-invariant bivector field on $M$ and $B$ is also $G$-invariant, then the gauge related bivector field $\pi_\B$ is $G$-invariant as well.  However, the reduced bivector fields $\pi_\red$ and $\pi_\red^\B$ on the quotient $M/G$ might not be gauge related. Indeed their characteristic distributions may be different. The reduced bivector fields $\pi_\red$ and $\pi_\red^\B$ will be gauge related only when the 2-form $B$ is basic with respect to the orbit projection $\rho: M \to M/G$. In that case,  the gauge transformation related $\pi_\red$ and $\pi_\red^\B$ is given by $\bar{B}$, the (unique) 2-form on $M/G$ such that $\rho^*\bar{B} = B$. 
 \end{enumerate}
\end{remark}

Consider now the scenario given by a nonholonomomic system $(\M, \pi_\nh, \Ham_\subM)$ as in Sect.~\ref{Sec:Preliminaries_Nonholonomic}. From  \eqref{definition_nh_bracket} we see that $\pi_\nh$ is determined by the (non-integrable) distribution $\C$  and the nondegenerate $2$-section $\Omega_\subC$ on $\C$ given in (\ref{definition_C}) and (\ref{def:Omega_M_and_Omega_C}) respectively. 
Now, using gauge transformations by 2-forms, our goal is to generate new bivector fields $\pi_\B$ describing the nonholonomic dynamics on $\M$, i.e. $\pi_\B^\sharp(d\Ham_\subM) = - X_\nh$. From \cite{PL2011}, we recall the concept of a {\it dynamical gauge transformation}.

\begin{definition}\label{Def:DynGauge} \cite{PL2011} We say that a $2$-form $B$ defines a {\it dynamical gauge transformation} if
\begin{equation*}
\begin{split}
 & (i) \ \ \Omega_\subC + B|_\C \mbox{ is a nondegenerate $2$-section,}\\
 & (ii) \ \mathbf{i}_{X_\nh} B = 0.
\end{split}
\end{equation*}
Then, the {\it dynamically gauge related bracket} $\pi_\B$ is given by 
 \begin{equation}\label{Eq:GaugedNH}
 \pi_\B^\sharp(\alpha) = -X \Leftrightarrow \mathbf{i}_X (\Omega_{\subM} +B ) |_\C = \alpha |_\C, \qquad \mbox{for } \alpha \in T^*\M,
\end{equation}
and satisfies that $\pi_\B^\sharp (d\Ham_\subM)= - X_\nh$.  We call condition $(ii)$ the {\it dynamical condition} for $B$.  
\end{definition}

\begin{remark}
\label{R:remarks_def_of_B}
\begin{enumerate} 
\item[$(i)$] If we consider a $2$-form $B$ that is semi-basic with respect to the bundle $\tau_\subM : \M \rightarrow Q$, then $\Omega_\subC + B|_\C$ is automatically nondegenerate, see \cite{PL2011}.
 \item[$(ii)$] Let $B_1$ and $B_2$ be semi-basic 2-forms (with respect to $\tau_\subM:\M\to Q$).  Performing a gauge transformation of $\pi_\nh$ by the 2-form $B_1$ and subsequently by $B_2$, is the same as performing a gauge transformation of $\pi_\nh$ by the 2-form $B = B_1 + B_2$.  
\item[$(iii)$] From \eqref{Eq:GaugedNH}, we see that the bivector field $\pi_\B$ only depends on the restriction of $B$ to the characteristic distribution $\C$. Therefore, following \cite{Bal-14}, once we choose a vertical complement $\W$ so that $T \M = \C \oplus \W$, we will also ask that 
\begin{equation}
\label{condition_B_W}
\mathbf{i}_X B \equiv 0, \mbox{ for } \quad X \in \Gamma(\W) .
\end{equation}
\end{enumerate}
\end{remark}

As observed in Remark \ref{R:Gauge} $(iii)$, if the nonholonomic system $(\M, \pi_\nh, d\Ham_\subM)$ has a $G$-symmetry then a $G$-invariant 2-form $B$ defining a gauge transformation generates a $G$-invariant bivector field $\pi_\B$ as in \eqref{Eq:GaugedNH}.  
So, the  gauge related bivector $\pi_\B$ induces a reduced almost Poisson bracket $\{ \cdot , \cdot \}_\red^\B$ on the differential space ${\mathcal M}/G$ given, for $\bar{f}, \bar{g} \in C^\infty(\mathcal{M}/G)$, by
\begin{equation}
\label{eq:reduced_bracket_B}
\{ \bar{f}, \bar{g}\}_\red^\B \circ \rho = \{ \bar{f} \circ \rho, \bar{g}\circ \rho\}_\B,
\end{equation}
where, as usual, $\rho : \M \rightarrow \M / G$ is the orbit projection and $\{ \cdot , \cdot \}_\B$ is the bracket associated to $\pi_\B$. Moreover, if $B$ defines a dynamical gauge transformation, then the reduced bracket $\{ \cdot , \cdot \}_\red^\B$ describes the nonholonomic reduced dynamics:
\begin{equation}
\label{eq:ReducedDescriptionDynamics}
X_\red = \{ \cdot , \Ham_\red \}_\red^\B \in \mathfrak{X}(\M / G),
\end{equation}
where $\Ham_\red : \M / G \rightarrow \R$ is the reduced Hamiltonian.

Following \cite{Bal-14}, in order to analyze the failure of the Jacobi identity of $\{ \cdot , \cdot \}_\red^\B$ we use an analogous formulation as in Prop.~\ref{Prop:Jac} but now considering the gauge transformation. 

\begin{proposition}\label{Prop:JacB}\textup{\cite{Bal-14}} Consider the nonholonomic system $(\M, \pi_{\emph\nh}, d\Ham_\subM)$ with a $G$-symmetry acting properly satisfying the dimension assumption. Let $\W$ be a $G$-invariant vertical complement of the constraints $\C$ and let $\pi_\B$ be the bivector field on $\M$ gauge related  to $\pi_{\emph\nh}$ by a $G$-invariant $2$-form $B$ satisfying condition \eqref{condition_B_W}.
Then:
\begin{enumerate}
\item[$(i)$] The gauge related bracket $\{ \cdot , \cdot \}_\B$ on $\M$ satisfies that, for $f, g, h \in C^\infty(\M)$,
\begin{equation*}
cyclic[\{ f ,\{ g , h \}_\B \}_\B] = (d\langle J , \mathcal{K}_\subW \rangle - d B)(\pi_\B^\sharp (d f), \pi_\B^\sharp (d g), \pi_\B^\sharp (d h )) 
 - \psi_{\pi_\B} (df,dg,dh), 
\end{equation*}
where $\psi_{\pi_\B}$ is the trivector given by $\psi_{\pi_\B} (\alpha,\beta,\gamma) = cyclic \left[ \gamma \left(  ( {\mathcal K}_\subW(\pi_{\B}^\sharp (\alpha),\pi_{\B}^\sharp (\beta))\, )_\subM \right) \right]$, for $\alpha$, $\beta$, $\gamma$ 1-forms on $\M$.
\item[$(ii)$] The reduced bracket $\{\cdot, \cdot\}_{\emph\red}^\B$ on the differential space $\M/G$ induced by $\{\cdot, \cdot\}_\B$, as in \eqref{eq:reduced_bracket_B}, satisfies that, for $\bar{f}, \bar{g}, \bar{h} \in C^\infty(\M / G)$,
$$
cyclic [ \{ \bar{f} ,\{ \bar{g} ,\bar{h} \}_{\emph\red}^\B \}_{\emph\red}^\B \circ \rho ] = (d \langle J ,\mathcal{K}_\subW \rangle - d B)(\pi_\B^\sharp (d\rho^* \bar{f}), \pi_\B^\sharp (d \rho ^* \bar{g}), \pi_\B^\sharp (d \rho^* \bar{h})).
$$
\end{enumerate}
\end{proposition}


It follows that the reduced bracket $\{ \cdot, \cdot \}_\red^\B$ is Poisson if
$$
(d\langle J ,  {\mathcal K}_\subW\rangle - dB)|_{{\mathcal U}_\B} = 0 ,
$$ 
where $\mathcal{U}_\B$ is the distribution on $\M$ given by
\begin{equation}
\label{eq:def_U_B}
{\mathcal U}_\B = \textup{span} \{ \pi_\B^\sharp(d f) \: : \: f \in C^\infty(\M)^G \}.
\end{equation}
Recall that, on $\M_{reg}/G$, the characteristic distribution of $\{\cdot, \cdot \}_\red^\B$ is given by $\C_\red^\B := T\rho({\mathcal U}_\B)$.  Moreover, if the 3-form $d\langle J ,  {\mathcal K}_\subW\rangle - dB$ is semi-basic with respect to the orbit projection $\rho:\M_{reg}\to \M_{reg}/G$, then $\{\cdot,\cdot \}_\red^\B$ is a twisted Poisson bracket (Def.~\ref{Def:twisted}) on the regular strata $\M_{reg}/G$. In fact, since $d\langle J ,  {\mathcal K}_\subW\rangle - dB $ is $G$-invariant, then being semi-basic implies that it is basic and we get that  
$$
cyclic [ \{ \bar{f} ,\{ \bar{g} ,\bar{h} \}_\red^\B \}_\red^\B ] = \phi((\pi_\red^\B)^\sharp (d \bar{f}), (\pi_\red^\B)^\sharp (d  \bar{g}), (\pi_\red^\B)^\sharp (d\bar{h})),
$$
where $\phi$ is the 3-form on $\M_{reg}/G$ such that $\rho^*\phi = d\langle J ,  {\mathcal K}_\subW\rangle - dB$.  

Here, we also remark the importance of finding a twisted Poisson bracket on the reduced manifold $\M_{reg}/G$ describing the dynamics in order to compute a conformal factor for the bracket.  In particular, in \cite{Bal-Fer}, it was studied the Chaplygin ball using a gauge transformation to obtain a twisted Poisson bracket and afterwards, it was computed a conformal factor on each almost symplectic leaf to then extend it to the foliation and conclude that the twisted Poisson bracket was, in fact, a conformally Poisson bracket.

\subsection{Horizontal gauge symmetries and the choice of a gauge transformation}
\label{NonholonomicMomentMap}

In order to find an appropriate 2-form so that the reduced bracket is (twisted) Poisson,  we study properties of the first integrals of the system that are $G$-invariant {\it horizontal gauge momenta} \cite{BGM-96,FGS-05}, see \cite{Bal-16, LGN-JM-16}.

Consider a nonholonomic system $(\M, \pi_\nh, \Ham_\subM)$ with a proper $G$-symmetry satisfying the dimension assumption \eqref{dimension_assumption}.
The {\it nonholonomic momentum map} \cite{BKMM-96} is the map $J^\nh : \M \rightarrow \g_S^*$ defined, at each $m \in \M$ and $\eta \in \Gamma(\g_S)$, by
\begin{equation}
\label{NH-moment_map}
J_\eta (m) = \langle J^\nh, \eta \rangle (m) = \langle J^\nh (m), \eta (\tau_\subM(m)) \rangle := \mathbf{i}_{\eta_\subM} \Theta_\subM (m),
\end{equation}
where $\Theta_\subM$ is the restriction to $\M$ of the Liouville $1$-form $\Theta_Q$ on $T^*Q$.
Observe that the function $J_\eta$ on $\M$ is linear on the fibers of the bundle $\tau_\subM : \M \rightarrow Q$. However,
contrarily to the canonical momentum map for Hamiltonian systems, $J_\eta$ is not necessarily a first integral of the dynamics $X_\nh$. In fact, using the $G$-invariance of the restricted Hamiltonian $\Ham_\subM$, we have that for any $\eta\in\Gamma(\g_S)$, 
\begin{equation*}
X_\nh ( J_\eta  ) = ( \pounds_{\eta_\subM} \Theta_\subM ) (X_\nh),
\end{equation*}
where $\pounds$ denotes the Lie derivative. 

\begin{definition}\label{Def:HGM}\cite{BGM-96, FGS-08} A {\it horizontal gauge momentum} of $X_\nh$ is a function of type $J_\eta = \langle J^\nh , \eta \rangle \in C^\infty(\M)$, for $\eta\in\Gamma(\g_S)$, that is a first integral of $X_\nh$. The (global) section $\eta \in \Gamma(\g_S)$ is called a {\it horizontal gauge symmetry}. 
\end{definition}

Even if $J_\eta$ is a horizontal gauge momentum for the nonholonomic dynamics, the vector field $\pi_\nh^\sharp(dJ_\eta)$ might not be equal to the infinitesimal generator $- \eta_\subM$ or, even more, $\pi_\nh^\sharp(dJ_\eta)$ might not be a section of the vertical distribution $\V$. In fact, for $\eta\in\Gamma(\g_S)$ let us denote by $\Lambda_\eta$ the 1-form on $\M$ such that
\begin{equation}
\label{Eq:Lambda}
\Lambda_\eta |_\C =  - \mathbf{i}_{\eta_\subM} \Omega_\subC + d J_\eta |_\C = \pounds_{\eta_\subM} \Theta_\subM |_\C  \quad \mbox{and} \quad \Lambda_\eta |_{\W} = 0.
\end{equation}
Note that  $\pi_\nh^\sharp ( d J_\eta - \Lambda_\eta ) = - \eta_\subM$.  When $J_\eta$ is a first integral, we conclude that $\Lambda_\eta(X_\nh) = 0$ (see also \cite{Bal-San-2016}). 


Following the idea that $\pi_\nh^\sharp(dJ_\eta)$ might not be a vertical vector field (with respect to $\rho:\M\to \M/G$), our goal is to find a 2-form $B$ so that the (dynamically) gauge related bracket $\pi_\B$ satisfies  $\pi_{\B}^\sharp ( d J_\eta) = - \eta_\subM$.  Then, if $J_\eta$ is $G$-invariant, it will become a Casimir of the reduced bracket $\{\cdot, \cdot\}_\red^\B$.

\begin{proposition}
\label{prop:Lambda}
Let $(\M, \pi_{\emph\nh}, \Ham_\subM)$ be a nonholonomic system with a (proper) $G$-symmetry satisfying the dimension assumption and assume that $J_\eta$ is a horizontal gauge momentum with $\eta \in \Gamma(\g_S)$ its associated horizontal gauge symmetry.  If  the 2-form $B$ defines a dynamical gauge transformation so that
$$
\mathbf{i}_{\eta_\subM} B |_\C = \Lambda_\eta |_\C,
$$
where $\Lambda_\eta |_\C$ is given in \eqref{Eq:Lambda}, then $\pi_\B$ is an almost Poisson bracket on $\M$ describing the dynamics satisfying:
\begin{enumerate}
\item[$(i)$] $\pi_\B^\sharp (d J_\eta) = - \eta_\subM.$
\item[$(ii)$]  If $J_\eta$ and the 2-form $B$ are $G$-invariant, then the induced function $\bar J_\eta$ on the quotient space $\M/G$ is a Casimir of the reduced bracket $\{\cdot, \cdot\}_{\emph\red}^\B$.
\end{enumerate}
\end{proposition}

\begin{proof} By definition of the nonholonomic bivector field $\pi_\nh$ and using \eqref{Eq:Lambda} we have that $\mathbf{i}_{\eta_\subM} \Omega_\subC = d J_\eta |_\C - \Lambda_\eta |_\C$.  Therefore, if $\mathbf{i}_{\eta_\subM} B |_\C = \Lambda_\eta |_\C $ then, by \eqref{def:pi_B}, $\pi_{\B}^\sharp ( d J_\eta) = - \eta_\subM$.

Following \eqref{eq:reduced_bracket_B}, since $B$ is $G$-invariant, the reduced bracket  $\{\cdot, \cdot\}_{\red}^\B$ is well defined on $\M/G$ (and, of course, it describes the reduced dynamics). Denoting $\bar{J}_\eta \in C^\infty(\M/G)$ the reduced function such that $\rho^*(\bar{J}_\eta) = J_\eta$,  then for any $\bar{f}\in C^\infty(\M/G)$ we have  
\begin{equation*}
\{ \bar{f} , \bar{J}_\eta \}_\red^\B \circ \rho = \{ \rho^* \bar{f} , J_\eta \}_\B = \pi_\B^\sharp (d J_\eta ) ( \rho^* \bar{f} ) = 0,
\end{equation*}
showing that $\bar{J}_\eta$ is a Casimir.  
\end{proof}

\subsection{The case of $l=\textup{rank}(\g_S)$ horizontal gauge momenta}
\label{section_system_for_B}

In this section we consider a nonholonomic system $(\M, \pi_\nh, \Ham_\subM)$  with a free and proper action on $Q$ given by a connected Lie group $G$ defining a $G$-symmetry.   If the nonholonomic system has a symmetry given by a proper action, then our conclusions can be applied to the submanifold $\M_{reg}$ of $\M$, which is the submanifold where the $G$-action is free, and to the quotient manifold $\M_{reg}/G$. In Sec.~\ref{Sec:Solids} and \ref{Example_BallOnSurface}, we study examples where the action is not free, however the 2-form $B$ can be easly defined on the entire manifold $\M$ and then the reduced bracket is defined on the differential space $\M/G$ as in \eqref{eq:reduced_bracket_B}.  

From now on, we assume that the free and proper $G$-symmetry satisfies the dimension assumption \eqref{dimension_assumption} and, if $l=\textup{rank}(\g_S$), we have $J_1, ..., J_l$ horizontal gauge momenta such that their associated horizontal gauge symmetries $\{\eta_1, ..., \eta_l\}$ form a (global) basis of sections of $\g_S$.  

Following \eqref{Eq:Lambda}, for each $i=1,...,l$ we denote by $\mathrm{Y}_i:= (\eta_i)_\subM$ and define the 1-forms $\Lambda_i$ on $\M$ such that
\begin{equation}\label{Eq:Pinh(dJ)}
\pi_\nh^\sharp(dJ_i) = - \mathrm{Y}_i + \pi_\nh^\sharp(\Lambda_i), 
\end{equation}
Equivalently, each $\Lambda_i$ is given by the condition that $\Lambda_i |_\C = - {\bf i}_{\mathrm{Y}_i} \Omega_\subC + dJ_i|_\C = \pounds_{\mathrm{Y}_i} \Theta_\subM |_\C$. 

In order to obtain a (dynamically gauge related) bivector field $\pi_\B$ so that $\pi_\B^\sharp(dJ_i) = -\mathrm{Y}_i\in \Gamma(\V)$,  we look for a 2-form $B$ defining a dynamical gauge transformation so that
\begin{equation}\label{Eq:BonSi}
 {\bf i}_{\mathrm{Y}_i} B |_\C = \Lambda_i|_\C , \qquad \mbox{for each} \ i=1,...,l,
\end{equation}
as in Prop.~\ref{prop:Lambda}. For that purpose, we consider a $G$-invariant distribution $H$ on $Q$ so that $D= H\oplus S$ (using the kinetic energy metric, the distribution $H$ can be chosen to be $H = S^\perp \cap D$, but it is not necessarily of this form).  Then, we obtain a decomposition of the tangent bundle given by 
\begin{equation}\label{Eq:Splitting:H+S+W}
TQ = H \oplus S \oplus W = H\oplus V.
\end{equation}
and we denote by $P_H :TQ \to H$ and $P_V:TQ \to V$ the projections associated to the splitting $TQ = H\oplus V$.  Now, consider the principal connection 1-form $A_V :TQ \to \g$ given, at each $X\in \mathfrak{X}(Q)$, by
$$
A_V(X) = \xi \in \Gamma(Q\times\g\to Q)  \qquad \mbox{if and only if} \qquad P_V(X) = \xi_Q,
$$
where $H= \textup{Ker}A_V$. Note that the splitting $V=S\oplus W$ induces another $\g$-valued 1-form $A_S$ so that, for $X\in \mathfrak{X}(Q)$, 
$$
A_S (X) = \eta\in \Gamma(\g_S) \qquad \mbox{if and only if} \qquad P_S(X) = \eta_Q,
$$
where $P_S :TQ \to S$ is the projection associated to the splitting $TQ=H\oplus S \oplus W$. Observe that $A_V = A_S + A_W$, for $A_W$ the $\g$-valued 1-form defined in \eqref{definition_A_W}.

We denote by $\mathcal{A}_\V := \tau_\subM^*A_V:T\M\to \g$ the $\g$-valued 1-form on $\M$, and we see that it is a principal connection on $\rho:\M\to \M/G$  inducing a splitting of the tangent bundle
\begin{equation}
\label{eq:spliting_TM_HV}
T\M = \mathcal{H} \oplus {\mathcal V} ,
\end{equation}
where $\mathcal{H} = \textup{Ker}\mathcal{A}_\V$. Let us denote by $\mathcal{A}_\subS := \tau_\subM^*A_S$ the corresponding $\g$-valued 1-form on $\M$.  Since the horizontal gauge symmetries  $\{\eta_1,...,\eta_l\}$ form a global basis of $\g_S$, then 
\begin{equation*}
\mathcal{A}_\subS = \mathcal{A}^i_\subS \otimes \eta_i = \sum_{i=1}^l \mathcal{A}^i_\subS \otimes \eta_i,
\end{equation*}
 for $\mathcal{A}_\subS^i$ standard 1-forms on $\M$ such that $\mathcal{A}_\subS^i(\mathrm{Y}_j) = \delta_{ij}$ and $\mathcal{A}_\subS^i |_{\mathcal H} = 0 = \mathcal{A}_\subS^i|_\subW$. 
Considering again the momentum map as a $\g^*$-valued function $J:\M\to \g^*$, we denote by $\langle J, d^\C\!\mathcal{A}^i_\subS \otimes \eta_i\rangle$ the 2-form given by the natural pairing between $J$ and the $\g$-valued 2-form $(d^\C\!\mathcal{A}^i_\subS )\otimes \eta_i$.  

We now have the ingredients to define the 2-form $B_1$ that will play an important role in the hamiltonization process:
\begin{equation} \label{Eq:B1}
  B_1  :=   \langle J, \mathcal{K}_\subW \rangle +  \langle J, d^\C\!\mathcal{A}^i_\subS \otimes \eta_i \rangle =  \langle J, \mathcal{K}_\subW \rangle +  J_i\, d^\C\!\mathcal{A}^i_\subS . 
 \end{equation}
 
Next, following Diagram \eqref{Diagram}, we define the basic (with respect to the principal bundle $\rho:\M\to \M/G$) 2-form $\mathcal{B}$. 
 First, consider the curvature $\mathcal{K}_\V$ associated to the principal connection $\mathcal{A}_\V$, that is, $\mathcal{K}_\V$ is the $\g$-valued 2-form on $\M$ given, at each $X,Y\in T\M$, by 
 $$
\mathcal{K}_\V (X,Y) = d^{\mathcal H} \mathcal{A}_\V (X,Y) = d\mathcal{A}_\V(P_{\mathcal H} (X), P_{\mathcal H}(Y)),
$$
where $P_{\mathcal H} :T\M \to \mathcal{H}$ is the projection associated to the splitting \eqref{eq:spliting_TM_HV}.
Analogously as it was done in \eqref{Eq:B1}, we may consider the 2-form $\langle J, \mathcal{K}_\V \rangle$ given by the natural pairing between $J$ and the $\g$-valued 2-form $\mathcal{K}_\V$.  
 We also define $ \kappa_\g $ as the $\g^*$-valued 1-form on $\M$ given, at each $X\in T\M$, by
 \begin{equation}\label{Def:kappag}
\langle \kappa_\g (X), \chi\rangle  = \kappa (T\tau_\subM(X), \chi_Q), \qquad \mbox{for } \chi\in\g.
 \end{equation}
 Considering the $\g$-valued 1-form on $\M$ given by ${\bf i}_{P_\V(X_\nh)} [\mathcal{K}_\subW + d^\C \mathcal{A}_\subS^i\otimes \eta_i]$,  where $P_\V:T\M\to \V$ denotes the projection associated to the decomposition \eqref{eq:spliting_TM_HV}, we define the 2-form $(\kappa_\g \wedge {\bf i}_{P_\V(X_\nh)} [\mathcal{K}_\subW + d^\C \mathcal{A}_\subS^i\otimes \eta_i] )_\mathcal{H}$ on $\M$ given, at each $X,Y\in T\M$, by 
 \begin{equation*}
  \begin{split}
   (\kappa_\g \wedge {\bf i}_{P_\V(X_\nh)} [\mathcal{K}_\subW + d^\C \mathcal{A}_\subS^i\otimes \eta_i])_\mathcal{H} (X,Y) =  & \  \kappa_\g \wedge {\bf i}_{P_\V(X_\nh)} [\mathcal{K}_\subW + d^\C \mathcal{A}_\subS^i\otimes \eta_i]\, (P_{\mathcal H}(X),P_{\mathcal H}(Y)) \\  
  =  & \  \langle \kappa_\g(P_\mathcal{H}(X)) , [\mathcal{K}_\subW + d \mathcal{A}_\subS^i\otimes \eta_i ]( P_\V(X_\nh), P_\mathcal{H}(Y) ) \rangle \\
   & - \langle \kappa_\g(P_\mathcal{H}(Y)) , [\mathcal{K}_\subW + d \mathcal{A}_\subS^i\otimes \eta_i] ( P_\V(X_\nh), P_\mathcal{H}(X) ) \rangle.
  \end{split}
 \end{equation*}
Finally, we define the 2-form $\mathcal{B}$ on $\M$ given by
 \begin{equation} \label{Eq:mathcalB}
  \mathcal{B} := - \langle J, \mathcal{K}_\V\rangle - \frac{1}{2}(\kappa_\g \wedge {\bf i}_{P_\V(X_\nh)} [\mathcal{K}_\subW + d^\C \mathcal{A}_\subS^i\otimes \eta_i] )_\mathcal{H}.
 \end{equation}

 Observe that, by construction, the 2-forms $B_1$ and $\mathcal{B}$ are semi-basic forms with respect to the bundle $\tau_\subM:\M\to Q$. 
 
\begin{theorem}\label{T:GlobalB}
 Consider a nonholonomic system $(\M, \pi_{\emph\nh}, \Ham_\subM)$ carrying a free and proper $G$-action by symmetries  satisfying the dimension assumption. Suppose that the system admits $l=\textup{rank} (\g_\subS)$ $G$-invariant horizontal gauge momenta $\{J_1,...,J_l\}$ so that the associated symmetries $\{\eta_1,...,\eta_l\}$ form a basis of $\Gamma(\g_S)$.  Then the gauge transformation of $\pi_{\emph\nh}$ by the 2-form $B$ given by
 \begin{equation}\label{Eq:GlobalB}
 \begin{split}
 B & =  B_1 + \mathcal{B} \\
 & = \langle J, \mathcal{K}_\subW +d^\C\!\mathcal{A}^i_\subS \otimes \eta_i \rangle - \langle J, {\mathcal K}_\V \rangle  - \frac{1}{2}(\kappa_\g \wedge {\bf i}_{P_\V(X_{\emph\nh})} [\mathcal{K}_\subW + d^\C \mathcal{A}_\subS^i\otimes \eta_i] )_\mathcal{H},
 \end{split}
 \end{equation}
 defines a bracket $\pi_\B$ on $\M$ such that:
 \begin{enumerate}
  \item[$(i)$] $\pi_\B$ describes the dynamics: $\pi_\B^\sharp(d\Ham_\subM) = - X_\nh$. 
  \item[$(ii)$] $\pi_\B$ is $G$-invariant and the reduced bracket $\pi_{\emph\red}^\B$ on $\M/G$ is a twisted Poisson bracket by the exact 3-form $(-d\bar{\mathcal B})$, where $\bar{\mathcal B}$ is the 2-form on $\M/G$ such that $\rho^*\bar{\mathcal B} = \mathcal{B}$,  and it verifies $(\pi_{\emph\red}^\B)^\sharp(d\Ham_{\emph\red}) = - X_{\emph\red}$.
  \item[$(iii)$]  $\pi_{\emph\red}^\B$  has an (integrable) characteristic distribution tangent to a regular foliation of dimension  $2 \: \textup{dim}(Q/G)$ defined by the common level sets of the (reduced) horizontal gauge momenta $\{\bar{J}_1,...,\bar{J}_l\}$, where $\bar{J}_i\in C^\infty(\M/G)$ such that $\rho^*\bar{J}_i = J_i$. 
 \end{enumerate}
\end{theorem}

%
%
%

In order to prove Theorem \ref{T:GlobalB} and to understand the role played by the 2-forms $B_1$ and $\mathcal{B}$, we follow Diagram \eqref{Diagram} using the decomposition $B = B_1 + \mathcal{B}$. 
The outline of the proof is as follows: we will see that the gauge transformation of $(\M, \pi_\nh)$  by $B_1$ gives a new $G$-invariant bivector field $\pi_1$ (described in Proposition \ref{L:B_1}).  Its reduced bracket $\pi_\red^1$ on $\M/G$ is Poisson and its characteristic foliation is given by the common level sets of the (reduced) horizontal gauge momenta.  The problem is that $\pi_\red^1$ might not describe the reduced dynamics (even though $X_\red$ is tangent to the symplectic foliation).  
 However, performing a gauge transformation by $\mathcal{B}$ of $\pi_1$ we gain two properties for the resulting gauge related bivector $\pi_\B$. First, Lemma \ref{L:Invariance} proves that $\mathcal{B}$ is basic with respect to the orbit projection $\rho:\M\to \M/G$, and thus $\pi_\red^1$ and $\pi_\red^\B$ are gauge related by the 2-form $\bar{\mathcal B}$ showing that their characteristic distributions coincide.
 Second, we show that $B=B_1 + \mathcal{B}$ satisfies the dynamical condition, so $\pi_\red^\B$ describes the reduced dynamics.

To prepare for the proof of Theorem \ref{T:GlobalB} we study some properties related to the $G$-invariance of the horizontal gauge momenta $J_i$ and the 2-forms $B_1$ and $\mathcal{B}$.

\begin{lemma}\label{L:Invariance}
Denote by $\rho_Q:Q \to Q/G$ and $\rho:\M\to \M/G$ the orbit projections of the $G$-action on $Q$ and $\M$, respectively.  Then
 \begin{enumerate}
  \item[$(i)$] for $\eta\in\Gamma(\g_S)$, the function $J_\eta ={\bf i}_{\eta_\subM} \Theta_\subM$ on $\M$ is $G$-invariant if and only if $\eta$ (regarded as a $\g$-valued function) is $G$-invariant, i.e., $[\eta, \chi] = 0$ for all $\chi\in \g$. 
\item[$(ii)$] If $X\in \Gamma(H)$ is $\rho_Q$-projectable, then $X$ is $G$-invariant, that is $[X, \chi_Q]=0$ for all $\chi\in\g$ and if $\mathrm{X}\in \Gamma(\mathcal{H})$ is $\rho$-projectable, then the function $p_{\mathrm{X}} : = {\bf i}_{\mathrm{X}} \Theta_\subM$ is $G$-invariant. 
\item[$(iii)$]  If the functions $J_i={\bf i}_{(\eta_i)_\subM} \Theta_\subM$ are $G$-invariant, then the 2-forms $B_1$ and $\mathcal{B}$ are $G$-invariant.  Since $\mathcal{B}$ is also semi-basic by construction, we conclude that it is basic and there is a 2-form $\bar{\mathcal{B}}$ on $\M/G$ such that $\rho^*\bar{\mathcal{B}} = \mathcal{B}$. 
\end{enumerate}
\end{lemma}

\begin{proof}
 $(i)$ Denote by $\tilde J : T^*Q \to \g^*$ the canonical momentum map on $T^*Q$. Note that $J = {\bf i}_{\eta_\subM} \Theta_\subM =  \iota^*(\langle \tilde J, \eta \rangle )$, where, as usual, $\iota: \M \to T^*Q$ is the inclusion.  Since $\langle \tilde J, \eta \rangle(m)= J_\eta (m)$ for $m\in\M$ and $\langle \tilde J, \eta \rangle (\bar m) = 0$ for $\bar m \in \kappa^\sharp (D^\perp)$, we see that $J_\eta$ is $G$-invariant if and only if $\langle \tilde J, \eta \rangle$ is $G$-invariant.  
 Now, we observe that, for any $\chi\in \g$, 
 $$
 \pounds_{\chi_{T^*Q}} \langle \tilde J, \eta \rangle = \pounds_{\chi_{T^*Q}} {\bf i}_{\eta_{T^*Q}} \Theta_Q = {\bf i}_{\eta_{T^*Q}} \pounds_{\chi_{T^*Q}}  \Theta_Q +  {\bf i}_{[\chi_{T^*Q}, \eta_{T^*Q}]} \Theta_Q = {\bf i}_{[\chi_{T^*Q}, \eta_{T^*Q}]} \Theta_Q.
 $$  
Then $\pounds_{\chi_{T^*Q}} \langle \tilde J, \eta \rangle = 0$ if and only if $0 = T\tau_\subM ( [\chi_{T^*Q}, \eta_{T^*Q}]) = [\chi_{Q}, \eta_{Q}] = - [\chi, \eta]_Q$, which is equivalent to $[\chi, \eta]=0$.
 
 $(ii)$ It is straightforward to see that if $X\in \Gamma(H)$ is $\rho_Q$-projectable then, for all $g\in G$,  $T\psi_g(X(q)) = X(\psi_g(q))$,  where $\psi_g: Q \to Q$ represents the $G$-action on $Q$.
 Now, for $\mathrm{X}\in \Gamma(\mathcal{H})$, we denote by $T\tau_\subM(\mathrm{X}) = X\in \Gamma(H)$. Then for $\chi\in \g$, we have 
 $ \pounds_{\chi_{\M}} p_{\mathrm{X}} = {\bf i}_{[\chi_{\subM}, \mathrm{X}]} \Theta_\subM$. Since $X$ is $\rho_Q$-projectable,  $[\chi_Q, X]=0$ and thus we get that $p_{\mathrm{X}}$ is $G$-invariant.  

 $(iii)$ In order to see that $B_1$ is $G$-invariant, from \eqref{Eq:B1}, we observe that it only remains to prove that $\langle J , d^\C\mathcal{A}_\subS^i \otimes \eta_i \rangle$ is $G$-invariant since $\langle J, \mathcal{K}_\subW\rangle$ has already been proven to be $G$-invariant.  
Let us show that the $\g$-valued 2-form  $d^\C\mathcal{A}_\subS^i \otimes \eta_i $ is $ad$-equivariant.  First, we observe that the horizonal gauge symmetries $\eta_i$ can be seen as $ad$-equivariant sections of $Q\times \g \to Q$ since, for all $\chi \in \g$, we have $\pounds_{\chi_Q} \eta_i = [\chi, \eta_i] = ad_\chi (\eta_i).$
On the other hand, let us consider a basis of sections $\{X_\alpha, Y_j \}$ of $D$ adapted to the splitting $D=H\oplus S$, where $X_\alpha\in \Gamma(H)$ are $\rho_Q$-projectable and $Y_j= (\eta_j)_Q$. Then, $X \in \Gamma(D)$ can be written as $X = x_\alpha X_\alpha + y_j Y_j$ and thus we have that $[\chi_Q,X] = \chi_Q (x_\alpha) X_\alpha + \chi_Q (y_j) Y_j$, where we used that $[\chi_Q,Y_j]=0$ and $[\chi_Q,X_\alpha]=0$ by items $(i)$ and $(ii)$, respectively. Hence, we get that
$(\pounds_{\chi_Q} A_S^i) (X) = \pounds_{\chi_Q} y_i - \chi_Q (y_i) = 0$  and as a consequence, using that $\C$ is a $G$-invariant distribution, we have 
$\pounds_{\chi_\subM} (d^\C\mathcal{A}_\subS^i \otimes \eta_i ) =  d^\C\mathcal{A}_\subS^i \otimes ad_\chi ( \eta_i)  = ad_\chi (d^\C\mathcal{A}_\subS^i \otimes \eta_i )$.
Finally, for all $\chi\in\g$, we conclude that
$$
\pounds_{\chi_\subM} \langle J , d^\C\mathcal{A}_\subS^i \otimes \eta_i \rangle = \langle - ad^*_\chi J, d^\C\mathcal{A}_\subS^i \otimes \eta_i \rangle + \langle J, ad_\chi (d^\C\mathcal{A}_\subS^i \otimes \eta_i) \rangle = 0.
$$

Next, we see that $\mathcal{B}$ is $G$-invariant.  It is straightforward to see that $\langle J, \mathcal{K}_\V\rangle$ is $G$-invariant since $J$ is $-ad^*$-equivariant while the principal curvature $\mathcal{K}_\V$ is $ad$-equivariant.  On the other hand, since $\kappa$ is the $G$-invariant kinetic energy metric, we conclude that the $\g^*$-valued 1-form $\kappa_\g$ is $-ad^*$-equivariant. Using that $\mathcal{K}_\subW$ and $d^\C \mathcal{A}_\subS^i\otimes \eta_i$ are   $ad$-equivariant, it is straightforward to see that ${\bf i}_{P_\V(X_{\nh})} [\mathcal{K}_\subW + d^\C \mathcal{A}_\subS^i\otimes \eta_i]$ is also $ad$-equivariant, and we conclude that $\kappa_\g \wedge{\bf i}_{P_\V(X_{\nh})} [\mathcal{K}_\subW + d^\C \mathcal{A}_\subS^i\otimes \eta_i]$ is $G$-invariant. 
 
 \end{proof}

\begin{proposition} \label{L:B_1}
 Under the hypothesis of Theorem \ref{T:GlobalB}, the gauge transformation of $\pi_{\emph\nh}$ by the 2-form $B_1= \langle J, \mathcal{K}_\subW \rangle +  \langle J, d^\C\mathcal{A}_\subS^i \otimes \eta_i \rangle$  induces a bivector field $\pi_1$ on $\M$ such that 
 \begin{enumerate}
  \item[$(i)$] $\pi_1^\sharp(dJ_i) = -(\eta_i)_\subM$,  for all $i=1,...,l$, 
  \item[$(ii)$] $\pi_1$ is $G$-invariant and the reduced bracket $\pi_{\emph\red}^1$ on $\M/G$ is Poisson with regular symplectic foliation of dimension  $2 \: \textup{dim}(Q/G)$, given by the common level sets of the functions $\{\bar{J}_1,...,\bar{J}_l\}$ where $\bar{J}_i \in C^\infty(\M/G)$ such that $\rho^*\bar{J}_i=J_i$. 
 \end{enumerate}

\end{proposition}

\begin{proof}
 First, observe that $B_1$ is a semi-basic 2-form with respect to the bundle $\tau_\subM:\M \to Q$ and thus, by Remark \ref{R:remarks_def_of_B}, the gauge transformation of $\pi_\nh$ by $B$ gives another bivector field denoted by $\pi_1$. 
 
 $(i)$ In virtue of Prop.~\ref{prop:Lambda}, it is enough to check that, for each $\mathrm{Y}_i = (\eta_i)_\subM$,   $\mathbf{i}_{\mathrm{Y}_i} B_1 |_\C = \Lambda_i |_\C$ as in \eqref{Eq:BonSi}. 
 Then, for $\mathrm{X}\in \Gamma(\C)$, we have
 \begin{equation*} 
\begin{split}
 B_1(\mathrm{Y}_i,\mathrm{X}) &= \langle J, \mathcal{K}_\subW (\mathrm{Y}_i,\mathrm{X}) \rangle +
 \langle J, d\mathcal{A}_\subS^k (\mathrm{Y}_i,\mathrm{X})\otimes \eta_k \rangle\\ 
 & = -\langle J, \mathcal{A}_\subW ([\mathrm{Y}_i,\mathrm{X}]) \rangle     +  \langle J, (\mathrm{Y}_i (\mathcal{A}_\subS^k (\mathrm{X})) -  \mathcal{A}_\subS^k( [\mathrm{Y}_i, \mathrm{X}]) ) \otimes \eta_k\rangle = \langle J ,  \mathrm{Y}_i (\mathcal{A}_\subS^k (\mathrm{X})) \eta_k \rangle - \langle J , \mathcal{A}_\V([\mathrm{Y}_i, \mathrm{X}]) \rangle,
\end{split}
\end{equation*} 
where we used that $\mathcal{A}_\V = \mathcal{A}_\subS + \mathcal{A}_\subW$.
Following \eqref{eq:spliting_TM_HV}, the vector field $\mathrm{X}$ on $\M$ can be decomposed as $\mathrm{X} = P_{\mathcal H}(\mathrm{X}) + P_{\mathcal V}(\mathrm{X})$ where $P_{\mathcal H}: T\M \to  \mathcal{H}$ and $P_{\mathcal V}: T\M \to  \mathcal{V}$ are the corresponding projections. Since $\mathcal{H}$ is $G$-invariant, then one can write $P_{\mathcal H} (\mathrm{X}) = x_n X_n$,  where $X_n$ are $\rho$-projectable vector fields taking values in $\mathcal{H}$ and thus, from Lemma \ref{L:Invariance}$(ii)$, we get that $[\mathrm{Y}_i, \mathrm{X}] = P_{\mathcal V}([\mathrm{Y}_i, \mathrm{X}]) + \mathrm{Y}_i(x_n) X_n $.   
Moreover,  using again Lemma \ref{L:Invariance} $(i)$ and $(ii)$ we get that 
$$
\mathrm{Y}_i(\Theta_\subM(\mathrm{X}) ) = \mathrm{Y}_i (\langle J, \mathcal{A}_\subS(\mathrm{X})\rangle + \mathrm{Y}_i(\Theta_\subM(x_nX_n)) =  J_k \mathrm{Y}_i (\mathcal{A}_\subS^k (\mathrm{X}))  + \mathrm{Y}_i(x_n) \Theta_\subM( X_n).
$$
Then, we conclude that 
\begin{equation*} 
\begin{split}
 B_1(\mathrm{Y}_i,\mathrm{X})   &= \langle J ,  \mathrm{Y}_i (\mathcal{A}_\subS^k (\mathrm{X})) \eta_k \rangle - {\bf i}_{P_{\mathcal V}([\mathrm{Y}_i, \mathrm{X}]) } \Theta_\subM =  J_k  \mathrm{Y}_i (\mathcal{A}_\subS^k (\mathrm{X})) - {\bf i}_{[\mathrm{Y}_i, \mathrm{X}] } \Theta_\subM +  \mathrm{Y}_i(x_n) \Theta_\subM(X_n) \\
 & =  \mathrm{X}(\Theta_\subM(\mathrm{Y}_i) ) + d\Theta_\subM (\mathrm{Y}_i, \mathrm{X}) = \Lambda_i(\mathrm{X}). 
\end{split}
\end{equation*} 

$(ii)$ By Lemma \ref{L:Invariance}$(iii)$,  $B_1$ is $G$-invariant and thus following Remark \ref{R:Gauge}$(iii)$, $\pi_1$ is also a $G$-invariant bivector field inducing a reduced bivector field $\pi_\red^1$ on $\M/G$. In order to show that $\pi_\red^1$ is Poisson, by  Prop. \ref{Prop:JacB}$(ii)$, we will check that $ ( d\langle J, {\mathcal K}_\subW\rangle - dB_1 ) |_{\mathcal{U}_{\B_1}} = 0$ for $\mathcal{U}_{\B_1} = \textup{span}\{\pi_1^\sharp(df) \ : \ f\in C^\infty(\M)^G\}$. 
First, using that $\pi_1^\sharp(dJ_i) = \mathrm{Y}_i$ and consequently that the reduced functions $\bar J_i$ on $\M/G$ are Casimirs of $\pi_\red^1$ (see Prop.\ref{prop:Lambda}), we see that, for any $f,g\in C^\infty(\M/G)$,  
$$
{\bf i}_{\mathrm{Y}_i} ( d\langle J, {\mathcal K}_\subW\rangle - dB_1) (\pi_1^\sharp(\rho^*f), \pi_1^\sharp(d\rho^*g)) \circ \rho =  cyclic[\{ \bar{J}_i ,\{ f , g \}_\red^1 \}_\red^1 ] = 0.
$$
Second, since $\mathcal{U}_{\B_1} \subset \C = \mathcal{H} \oplus \S$ (see \eqref{eq:spliting_TM_HV}), it remains to see that $d\langle J, {\mathcal K}_\subW\rangle - dB_1 |_{{\mathcal H} \cap \mathcal{U}_{\B_1}} = 0$. From \eqref{Eq:B1}, observe that $d\langle J, {\mathcal K}_\subW\rangle - dB_1 = - d \langle J , d^\C\mathcal{A}^i_\subS \otimes \eta^i \rangle = - d( J_i\, d^\C\!\mathcal{A}^i_\subS)$.   Hence for  $X_1, X_2, X_3\in \Gamma(\mathcal{H} \cap \mathcal{U}_{B_1}),$ we obtain that 
\begin{equation*}
 \begin{split}
d \langle J , d^\C\!\mathcal{A}^i_\subS \otimes \eta^i \rangle (X_1, X_2, X_3) & = cyclic \left[ X_1( J_i\, d\mathcal{A}^i_\subS (X_2, X_3)) - J_i\, d\mathcal{A}^i_\subS(P_\C([X_1, X_2]), X_3)  \right] \\
& = d(J_i\, d \mathcal{A}^i_S )(X_1, X_2, X_3) + cyclic \left[ J_i\, d \mathcal{A}_\subS^i (P_\subW ([X_1, X_2]), X_3 ) \right] \\
& = cyclic \left[ J_i\, d \mathcal{A}_\subS^i (P_\subW ([X_1, X_2]), X_3 ) \right], 
 \end{split}
\end{equation*}
where we used the fact that $dJ_i(X)=0$ for all $X\in \Gamma(\mathcal{U}_{\B_1})$. 

Note that  $d \mathcal{A}_\subS^i (Z, X)=0$ for all $Z\in \Gamma(\W)$ and $X\in \Gamma(\mathcal{H})$. In fact, on hte one hand, since $\mathcal{A}_\V$ is a principal connection, it is well known that $d\mathcal{A}_\V(Z, X)=0$ for all $Z\in \Gamma(\V)$ and $X\in \Gamma(\mathcal{H})$. On the other hand, it was already proven in Lemma \ref{L:dK=0} that $d\mathcal{A}_\subW(Z,X)=0$ for all $Z \in \Gamma(\W)$ and $X\in \Gamma(\mathcal{H})$.  Then, for $Z\in \Gamma(\W)$ and $X\in \Gamma(\mathcal{H})$,   $d\mathcal{A}_\subS(Z,X) =(d\mathcal{A}_\V - d\mathcal{A}_\subW)(Z,X) = 0$.  Finally, since $\mathcal{A}_\subS^i(X)= \mathcal{A}_\subS^i(Z) = 0$, we get that $0=d\mathcal{A}_\subS(Z,X) = (d \mathcal{A}_\subS^i \otimes \eta_i + \mathcal{A}^i_\subS \otimes d\eta^i)(Z,X) = d \mathcal{A}_\subS^i(Z,X) \otimes \eta_i$ which implies that $d \mathcal{A}_\subS^i (Z, X)=0$. Therefore, we proved that $d \langle J , d^\C\!\mathcal{A}^i_\subS \otimes \eta^i \rangle (X_1, X_2, X_3) = 0$ showing finally that $ ( d\langle J, {\mathcal K}_\subW\rangle - dB_1 ) |_{\mathcal{U}_{\B_1}} = 0$.

Next, we show that the characteristic foliation of $\pi_\red^1$ is determined by the common level sets of the $l$ (independent) Casimirs $\{\bar J_1,...,\bar J_l\}$. In fact, let us denote $\textup{dim} (Q) = n$, $\textup{rank} (D) = d$ and $\textup{dim} (G) = m$. Then, by the dimension assumption and using that the $G$-action is free, we have $\textup{rank} (S) = l = m+d-n$ and $\textup{dim} (\M/G) = n+d-m$. 
Since the distribution ${\mathcal U}_{\B_1}$ is generated by the hamiltonian vector fields given by $G$-invariant functions (see \eqref{eq:def_U_B}), we have that  $\textup{rank} ({\mathcal U}_{\B_1}) = n+d-m$.  Now, using that
 the characteristic distribution of $\pi_\red^\B$, denoted by $\C_\red^\B$, is the projection of ${\mathcal U}_{\B_1}$ to the quotient $\M / G$, we get 
$$
\textup{rank} (\C_\red^\B) = \textup{rank} \: T\rho({\mathcal U}_{\B_1}) = \textup{rank} ({\mathcal U}_{\B_1}) - \textup{rank} (\S) = 2(n-m) = 2 \: \textup{dim} (Q/G).
$$
Then, the rank of the annihilator of $\C_\red^\B$ is given by $\textup{rank}(T^*(\M/G)) - \textup{rank} (\C^\B_\red) = l$. Since the 1-forms $\{d \bar J_1,...,d\bar J_l\}$ are independent, they generate this annihilator and thus the distribution $\C_\red^\B$ is exactly the integrable distribution with foliation given by the level sets of the reduced functions $\bar J_i$ induced by the horizontal gauge momenta $J_i$, for $i=1,...,l$.

\end{proof}

\begin{remark}
 Observe that Proposition \ref{L:B_1} explaining the first step of gauging $\pi_\nh$ is ``purely geometric'' is the sense that, given any global $G$-invariant basis $\{\zeta_1,...,\zeta_l\}$  of $\Gamma(\g_S)$, the associated 2-form $B_1$ (given in \eqref{Eq:B1}) defines a reduced Poisson bivector field $\pi_\red^1$ whose characteristic foliation is given by the associated functions $\mathcal{J}_i = {\bf i}_{\zeta_i}\Theta_\subM$ (considered as $G$-invariant functions).  This foliation is independent of the dynamics.  In fact, the reduced nonholonomic vector field $X_\red$ is not necessarily tangent to the foliation.  If the basis of $\g_S$ happens to be given by horizontal gauge symmetries, then the resulting foliation will be tangent to the dynamics, which is the case that we are interested in. 
\end{remark}

Following Diagram \eqref{Diagram}, we proceed with the second step of the gauge transformation by $B$.  By performing a gauge transformation of $\pi_1$ by the (basic) 2-form $\mathcal{B}$ we maintain the characteristic foliation but we modify the leafwise symplectic form by the 2-form $\mathcal{B}$.  As a result, we obtain a bivector field with an almost symplectic foliation.

\begin{proposition} \label{Prop:gauge2}
  Under the hypothesis of Theorem \ref{T:GlobalB}, the gauge transformation of $\pi_1$ by the 2-form $\mathcal{B} = - \langle J, \mathcal{K}_\V\rangle -  \frac{1}{2}\kappa_\g\wedge {\bf i}_{P_\V(X_{\emph\nh})} [\mathcal{K}_\subW + d\mathcal{A}_S^i \otimes \eta_i]$ 
 induces a $G$-invariant bivector field $\pi_\B$ on $\M$ such that the induced reduced bracket $\pi_{\emph\red}^\B$ on $\M/G$ satisfies:
 \begin{enumerate}
   \item[$(i)$]  $\pi_{\emph\red}^\B$  is a regular $(- d\bar{\mathcal{B}})$-twisted Poisson bracket.
  \item[$(ii)$] Its characteristic foliation coincides with the one given by $\pi_{\emph\red}^1$, that is, it is given by the common level sets of the functions $\{\bar{J}_1,...,\bar{J}_l\}$.  
 \end{enumerate}
\end{proposition}

\begin{proof}
First, we observe that since the 2-form $\mathcal{B}$ is semi-basic with respect to the bundle $\tau_\subM:\M\to Q$, then $\mathcal{B}$ defines a gauge transformation of $\pi_1$ inducing a new bivector field $\pi_\B$. From Remark \ref{R:remarks_def_of_B} we see that the gauge transformation of $\pi_1$ by the 2-form $\mathcal{B}$ is the same as performing a gauge transformation of $\pi_\nh$ by the 2-form $B = B_1 + \mathcal{B}$.  

Second, since $\mathcal{B}$ is basic with respect to the principal bundle $\rho:\M\to \M/G$, then the following diagram commutes:
$$
   \xymatrix{(\M, \pi_1)  \ar[d]^{\rho} \ar[r]^{\mathcal{B}} & (\M, \pi_\B)  \ar[d]^{\rho} \\
              (\M/G, \pi_\red^1) \ar[r]^{\bar{\mathcal{B}}} & (\M/G, \pi_\red^\B). }
$$
That is, $\pi_\red^\B$ is also obtained as the gauge transformation of $\pi_\red^1$ by the 2-form $\bar{\mathcal{B}}$. If we denote by $(L_\mu, \Omega_\mu)$ the symplectic foliation associated to $\pi_\red^1$, then  $\pi_\red^\B$ has the same characteristic foliation as $\pi_\red^1$ but it carries an almost symplectic form on each leaf given by $\tilde\Omega_\mu = \Omega_\mu + \iota^*_\mu \bar{\mathcal B}$, where $\iota_\mu: L_\mu\to \M/G$ is the natural inclusion (see Remark \ref{R:Gauge}$(i)$).  Hence, on each leaf $d\tilde \Omega_\mu = d\iota^*_\mu\bar{\mathcal B}$ and thus $\pi_\red^\B$ is a twisted Poisson bracket by the 3-form $(-d\bar{\mathcal B})$.  
\end{proof}

\noindent {\it Proof of Theorem \ref{T:GlobalB}.} 
By Props. \ref{L:B_1} and \ref{Prop:gauge2} we conclude that the gauge transformation of $\pi_\nh$ by $B$ gives the $G$-invariant almost Poisson bracket $\pi_\B$ described in Prop. \ref{Prop:gauge2}. 

Next, we show that $B$ defines a {\it dynamical gauge transformation}, i.e., ${\bf i}_{X_\nh} B |_\C = 0$.  
First, we observe that ${\bf i}_{X_\nh} B ({\mathrm Y}_i) = - {\bf i}_{{\mathrm Y}_i} B (X_\nh) = -\Lambda_i(X_\nh) = - \pounds_{(\eta_i)_\subM} \Theta_\subM (X_\nh) = 0$ since $\eta_i$ is a horizontal gauge symmetry.  Since $\mathrm{Y}_i$, $i=1,...,l$ generate $\Gamma(\S)$, it remains to prove that ${\bf i}_{X_\nh} B ({\mathrm X}) = 0$ for all ${\mathrm X}\in \Gamma(\mathcal{H})$ (recall that $\C= {\mathcal H} \oplus {\mathcal S}$).   Using that $\mathcal{H}$ is $G$-invariant, by linearity, it is enough to check that ${\bf i}_{X_\nh} B ({\mathrm X}) = 0$ for any $\rho$-projectable vector field  ${\mathrm X}$  taking values in $\mathcal{H}$.

From the definition of $\Lambda_i$ in \eqref{Eq:Lambda}, and using that $\Theta_\subM (X_\nh)$ is a $G$-invariant function and $J_i$ are first integrals, we have that 
\begin{equation}\label{Eq:LambdaTheta}
0=  \Lambda_i (X_\nh) =  d\Theta_\subM({\mathrm Y}_i,X_\nh) = - \Theta_\subM ( [  {\mathrm Y}_i , X_\nh]) - X_\nh (J_i) + {\mathrm Y}_i (\Theta_\subM (X_\nh)) = - \Theta_\subM ( [ {\mathrm Y}_i , X_\nh]), 
\end{equation}
If we consider a (local) basis of  $\rho_Q$-projectable sections of $H$ given by $\{X_\alpha\}$ so that $\{X_\alpha, Y_j\}$ defines a local basis of section of $D$, then, the fact that $X_\nh$ is a {\it second order differential equation}, is expressed, for $m\in \M$, as  $T\tau_\subM(X_\nh(m)) = v^\alpha X_\alpha + v^j Y_j$ where $\kappa^\sharp(v) = m$ for $v=v^\alpha X_\alpha + v^j Y_j$. Since the coordinates $v^i$ and $v^\alpha$ are $G$-invariant, we get, from \eqref{Eq:LambdaTheta}, that 
\begin{equation*}
\begin{split}
0 & = \langle m, T\tau_\subM ([\mathrm{Y}_i, X_\nh])\rangle\\
& =  v^k v^j (\kappa (Y_k, [Y_i,Y_j])) + v^k v^\alpha ( \kappa (Y_k, [Y_i,X_\alpha]) + \kappa (X_\alpha, [Y_i,Y_k]) ) + v^\alpha v^\beta ( \kappa (X_\beta, [Y_i,X_\alpha]) ) ,
\end{split}
\end{equation*}
concluding that the coefficients of the above second order polynomial are skew-symmetric, in particular
\begin{equation}
\label{eq:kappa_zeros}
\kappa (X_\beta, [Y_i,X_\alpha]) = -\kappa(X_\alpha, [Y_i, X_\beta]) \quad \mbox{and} \quad 
\kappa (Y_j, [Y_i,X_\alpha]) + \kappa (X_\alpha, [Y_i,Y_j]) = 0.
\end{equation}

Now, we denote by $\sigma = [\mathcal{K}_\subW + d^\C \mathcal{A}_S^i\otimes \eta_i]$ the $\g$-valued 2-form on $\M$.  Consider a (local)  basis of sections $\{\xi_a\}$ of $\g_W$ so that $\{ \eta_i, \xi_a \}$ is a basis of of $\g \times Q \rightarrow Q$. Then we recall that $\mathcal{K}_\subW = d\tilde \epsilon^a \otimes \xi_a$ where $\tilde \epsilon^a$ are the constraint 1-forms on $\M$ so that $\tilde\epsilon^a((\xi_b)_\subM)=\delta^a_b $.  It follows that, for $\mathrm{X}_1,\mathrm{X}_2,\mathrm{X}_3\in T\M$ with $X_i = T\tau_\subM(\mathrm{X}_i)$ in $TQ$,  $\langle \kappa_\g(\mathrm{X}_1), \sigma(\mathrm{X}_2,\mathrm{X}_3)\rangle = \kappa(X_1, P_V ([X_2,X_3]) )$.  Therefore, \eqref{eq:kappa_zeros} implies that, for $\mathrm{Y}\in\Gamma(\S)$,
\begin{equation}\label{Eq:Proof-permutations}
\begin{split}
(i) \ & \ \langle \kappa_\g(\mathrm{X}_1), \sigma(\mathrm{Y}, \mathrm{X}_2)\rangle = - \langle \kappa_\g(\mathrm{X}_2), \sigma(\mathrm{Y}, \mathrm{X}_1)\rangle \mbox{ for } \mathrm{X}_1, \mathrm{X}_2 \in \Gamma(\mathcal{H}) \mbox{ $\rho$-projectable}, \\
(ii) & \ \langle \kappa_\g(\mathrm{Y}_1), \sigma(\mathrm{Y}, \mathrm{X})\rangle = - \langle \kappa_\g(\mathrm{X}), \sigma(\mathrm{Y}, \mathrm{Y}_1)\rangle \mbox{ for } \mathrm{X}\in \Gamma(\mathcal{H}) \mbox{ $\rho$-projectable},\mathrm{Y}_1 \in \Gamma(\S).
\end{split}
\end{equation}

Second, by the definition of $\kappa_\g$ \eqref{Def:kappag} and using again that  $X_\nh$ is a second order differential equation, we observe that, for $\chi\in\g$,
\begin{equation}\label{Eq:Proof-kappaJ}
\langle \kappa_\g(X_\nh(m)), \chi\rangle = \kappa(T\tau_\subM(X_\nh(m)) , \chi_Q\rangle = \langle m, \chi_Q\rangle = \langle J, \chi\rangle(m).
\end{equation}

Using the fact that $\sigma|_{\mathcal H} = \mathcal{K}_\V |_{\mathcal H}$ and also \eqref{Eq:Proof-permutations}(i) and \eqref{Eq:Proof-kappaJ}, we get for $\mathrm{X}\in \mathcal{H}$,
\begin{equation*}
\begin{split}
 {\bf i}_{X_\nh} B ({\mathrm X}) &= \langle J, \sigma (X_\nh, \mathrm{X}) \rangle - \langle J , \mathcal{K}_\V (X_\nh,{\mathrm X}) \rangle - {\textstyle \frac{1}{2}}\kappa_\g \wedge  {\bf i}_{P_\V(X_\nh)} \sigma (P_{\mathcal H} (X_\nh), \mathrm{X} )  \\
 &= \langle J, \sigma (P_\V(X_\nh), \mathrm{X}) \rangle - {\textstyle \frac{1}{2}}\langle \kappa_\g(P_{\mathcal H} (X_\nh)), \sigma( P_\V(X_\nh), \mathrm{X} ) \rangle+ {\textstyle \frac{1}{2}}\langle \kappa_\g(\mathrm{X} ), \sigma( P_\V(X_\nh), P_{\mathcal H} (X_\nh)) \rangle \\
 &=  \langle \kappa_\g(X_\nh), \sigma (P_\V(X_\nh), \mathrm{X}) \rangle - \langle \kappa_\g(P_{\mathcal H} (X_\nh)), \sigma( P_\V(X_\nh), \mathrm{X} ) \rangle \\
 & = \langle \kappa_\g(P_{\mathcal V} (X_\nh)), \sigma( P_\V(X_\nh), \mathrm{X} ) \rangle = - \langle \kappa_\g(\mathrm{X}), \sigma( P_\V(X_\nh),P_{\mathcal V} (X_\nh)  ) \rangle = 0,
\end{split}
\end{equation*}
where in the last equality we used \eqref{Eq:Proof-permutations}$(ii)$.
%
%

\hfill $\square$

\begin{remark}\label{R:B:a_(alphabeta)} 
\begin{enumerate} 
 \item[$(i)$]  The 2-form $B$ in \eqref{Eq:GlobalB} is not the unique 2-form generating a reduced bracket $\{\cdot, \cdot\}_\red^\B$ satisfying   $(i)$ and $(ii)$ of Theorem \ref{T:GlobalB}. In fact, any 2-form $\tilde B$ defining a  dynamical gauge transformation that generates a bivector field $\pi_{\tilde{\B}}$  satisfying conditions $(i)$ and $(ii)$ of Theorem \ref{T:GlobalB} can be written as
\begin{equation*}
\label{form_B_basis_hgm_solution_general}
\tilde B = B +  \mathfrak{B},
\end{equation*}
where  $B$ is given in \eqref{Eq:GlobalB} and $\mathfrak{B}$ is a basic 2-form (with respect to $\rho:\M\to \M/G$) satisfying ${\bf i}_{X_\nh} {\mathfrak B} = 0$ (see Remark \ref{R:Gauge}$(iv)$).
\item[$(ii)$] The splitting $TQ=H\oplus S \oplus W$ in \eqref{Eq:Splitting:H+S+W} is not canonical.  First, the choice of the vertical complement of the constraints $W$ is not canonical, however, it can be checked that the 2-forms $B_1$ and $\mathcal{B}$ remain invariant for different choices of $W$.  On the other hand, we have also the choice of a principal connection giving rise to the horizontal space $H$.  If we denote by $B_1$ and $\bar{B}_1$ the 2-forms given in \eqref{Eq:B1} but computed for two different principal connections so that $TQ= H\oplus V$ and $TQ = \bar{H}\oplus V$ respectively, then $B_1$ and $\bar{B}_1$ differ from a basic term (with respect to $\rho:\M\to \M/G$) that is closed and thus their induced brackets $\{\cdot, \cdot \}_\red^1$ described in Prop.~\ref{L:B_1} are Poisson with the same foliation. Moreover, since $\mathcal{B}$ is already basic (no matter the horizontal space chosen), then the conclusions relative to the integrability of the brackets $\{\cdot, \cdot \}_\red^\B$ on $\M/G$ computed for different choices of $H$ remain the same.    
\item[$(iii)$] The almost Poisson bivector fields $\pi_1$ and $\pi_\B$ have the functions $\{\Ham_\subM, J_1,...,J_l\}$ in involution while the nonholonomic bracket $\pi_\nh$ does not.  
\end{enumerate}
\end{remark}

Following Remark \ref{R:B:a_(alphabeta)}$(i)$, when $\textup{rank}(H) = 1$, any basic 2-form (with respect to $\rho:\M\to \M/G$) is zero and thus $B$ coincides with $B_1$, in decomposition \eqref{Eq:GlobalB}.  On the other hand, if $\textup{rank}(H)=2$, then imposing the condition that ${\bf i}_{X_\nh} {\mathfrak B} = 0$ implies that $\mathfrak{B} = 0$. Therefore, from Theorem \ref{T:GlobalB} we obtain 

\begin{corollary}\label{C:l+1} Under the hypothesis of Theorem \ref{T:GlobalB}, the following holds:
 \begin{enumerate}
  \item[$(i)$] If $\textup{rank}(H)=1$, then $B_1$ is the unique (semi-basic with respect to $\tau_\subM: \M \to Q$) $G$-invariant 2-form such that the gauge related bracket $\pi_1$ verifies that $\pi^\sharp_1(dJ_i) = - (\eta_i)_\subM$.  Moreover, $B_1$ induces a dynamical gauge transformation and thus the Poisson bracket $\{\cdot,\cdot\}_{\emph\red}^1$ describes the reduced nonholonomic dynamics and has symplectic leaves of dimension  2. 
  \item[$(ii)$]   If $\textup{rank}(H)=2$, then $B= B_1 +\mathcal{B}$ is the unique (semi-basic with respect to $\tau_\subM: \M \to Q$) $G$-invariant 2-form inducing a dynamical gauge transformation so that $\pi_\B$ satisfies that $\pi^\sharp_\B(dJ_i) = - (\eta_i)_\subM$. In this case, the reduced bracket $\{\cdot,\cdot\}_{\emph\red}^\B$ (describing the reduced nonholonomic dynamics) is twisted Poisson with almost symplectic leaves of dimension 4. 
 \end{enumerate}
\end{corollary}

In other words, following Prop.~\ref{prop:Lambda}, if $\textup{rank}(H)=1$, then $B_1$ is the unique 2-form satisfying that \eqref{Eq:BonSi} for all $i=1,...,l$. In this case, the dynamical condition comes for free.  On the other hand, if $\textup{rank}(H)=2$ the 2-form $B=B_1 + \mathcal{B}$ defined in \eqref{Eq:GlobalB} is the unique 2-form verifying simultaneously the dynamical condition and \eqref{Eq:BonSi}.

 It is worth noticing that the main example treated in this paper (Sec.~\ref{Example_BallOnSurface}) belongs to the special case $(i)$ (away from the singularities), as well as the examples of a ball in a cylinder \cite{Bal-14,BGM-96} and the solid of revolution rolling on a plane (\cite{CDS-2010,Bal-16,LGN-JM-16} and Sec.~\ref{Sec:Solids}). On the other hand, the Chaplygin ball (\cite{Chap-03,BorMa2001,Naranjo2008} and Sec.~\ref{Example_Chaplygin_ball}) and the snakeboard (\cite{Bloch:Book,Bal-14} and Sec.~\ref{Example_Snake}) belong to the case $(ii)$.

 \begin{remark}[Chaplygin systems]  If $\textup{rank}(H) = \textup{rank}(D)$ then $S = \{0\}$ and $\mathcal{K}_\subW$ is exactly the principal curvature $\mathcal{K}_\V$ (see also \cite{Bal-14}).   In this case, there are no horizontal gauge momenta and, since $P_\V(X_\nh) = 0$ we obtain that $B_1=-\mathcal{B} = \langle J, \mathcal{K}_\V \rangle$, i.e., $B=0$. Therefore,  $\pi_\red$ is a $(-d\langle J, \mathcal{K}_\V\rangle)-$twisted Poisson bracket on $\M/G$ with a $2.\textup{dim}(Q/G)$-dimensional foliation; that is, since $\textup{dim}(\M/G)=2\textup{dim}(Q/G)$ then $\pi_\red$ is determined by a 2-form $\Omega$ as in Remark \ref{R:Gauge}$(i)$ such that $d\Omega =d\langle J, \mathcal{K}_\V\rangle$.  Therefore we recover also the known theory for Chaplygin system \cite{BS-93,Koiller1992}.
 \end{remark}
 
 \begin{remark}[Horizontal symmetries] Suppose that all horizontal gauge symmetries are given by constant sections $\eta_1, ..., \eta_l$.  Then one can prove that  $B_1 = \langle J , d^\C\mathcal{A}_\V\rangle$ and, moreover, using Lemma \ref{L:Invariance} one shows that $B_1$ is basic with respect to $\rho:\M\to \M/G$ (in fact, $d^\C\mathcal{A}_\V$ is basic). Then, we conclude $B_1(X,Y)  = \langle J, d\mathcal{A}_\V \rangle(P_{\mathcal H}(X), P_{\mathcal H}(Y))$ while $\mathcal{B} =  - \langle J, \mathcal{K}_\V\rangle,$ and therefore $B=0$.

 \end{remark}

\subsection{Coordinate approach}
\label{S:coordinate_approach_gauge}



In this section, we choose (local) basis adapted to the splitting $TQ = H \oplus S \oplus W$ in order to write the 2-form $B$ inducing the dynamical gauge transformation in coordinates.   As we did in Sec. \ref{section_system_for_B}, here we also assume that the $G$-action is {\it free} and proper and that the system admits $l=\textup{rank}(\g_S)$ $G$-invariant horizontal gauge symmetries $\{ \eta_1,...,\eta_l\}$.

Recall, from \eqref{Eq:Splitting:H+S+W}, the decomposition of the tangent bundle $TQ = H \oplus V$, where the $G$-invariant horizontal distribution $H \subset D$ defines a principal connection with respect to the principal bundle $\rho_Q : Q \to Q/G$.
Let us consider a $G$-invariant (local) basis $\{X_\alpha\}$ of sections  of $H$.

Use that $\g\times Q = \g_S \oplus \g_W$ (see \eqref{splitting_g_times_Q}) and consider the basis of sections $\{\eta_i,\xi_a\}$ of $\g\times Q \to Q$ where $\eta_i\in \Gamma(\g_S)$ are the horizontal gauge symmetries and $\xi_a\in\Gamma(\g_W)$, for $i=1,...,l$ and $a=l+1,...,m$.
 Then, we obtain the (local) basis  of sections of $TQ$ adapted to the splitting $TQ = H \oplus S \oplus W$ with its dual basis given by
\begin{equation}
\label{eq:basis_TQ}
\mathfrak{B}_{TQ} = \{X_\alpha, Y_i :=(\eta_i)_Q , Z_a := (\xi_a)_Q\} \qquad \mbox{and} \qquad \mathfrak{B}_{T^*Q}= \{X^\alpha, Y^i, \epsilon^a\}. 
\end{equation}

For the kinetic energy metric $\kappa$, we use the the following notation: $\kappa_{\alpha\beta} = \kappa(X_\alpha, X_\beta)$, $\kappa_{ij} = \kappa(Y_i, Y_j)$, $\kappa_{ai} = (Z_a, Y_i)$ and $\kappa_{a\alpha} = \kappa(Z_a, X_\alpha)$. If we denote by $(p_\alpha, p_i, p_a)$ the coordinates on $T^*Q$ induced by the basis $\mathfrak{B}_{T^*Q}$, the constraint manifold $\M\subset T^*Q$ is expressed as
$$
\M=\{ (q; p_\alpha, p_i, p_a) \in T^*Q \ : \  p_a = \kappa_{ai} v^i + \kappa_{a\alpha} v^\alpha\},
$$
where $p_i = \kappa_{ij} v^j + \kappa_{i\alpha} v^\alpha$ and $p_\a = \kappa_{i\a} v^i + \kappa_{\alpha\beta} v^\beta$ with $(v^\alpha, v^i, v^a)$  the coordinates on $TQ$ associated with the basis $\mathfrak{B}_{TQ}$. It is straightforward to see that $p_a$ depends linearly on $p_i$ and $p_\alpha$ and then we use coordinates $(q; p_i,p_\alpha)$ to denote an element on $\M$ and, when we write $p_a$ we really mean $\iota^*(p_a)$, where $\iota:\M\to T^*Q$ is the natural inclusion.

On $\M$ we consider the 1-forms $\tau_\subM^*X^\alpha = \tilde X^\alpha, \ \tau_\subM^* Y^i = \tilde Y^i$ and $\tau_\subM^* \epsilon^a = \tilde \epsilon^a$ (as usual, $\tau_\subM:\M \to Q$ is the canonical projection), and obtain the following dual basis on $\M$:
\begin{equation}\label{Eq:basisT*M}
\mathfrak{B}_{T^*\M}= \{\tilde X^\alpha, \tilde Y^i, \tilde \epsilon^a, dp_\alpha, dp_i\} \qquad \mbox{and} \qquad
\mathfrak{B}_{T\M}= \{\tilde X_\alpha, \tilde Y_i, \tilde Z_a, \partial_{p_\alpha}, \partial_{p_i}\} .
\end{equation}


Observe that the Liouville 1-form $\Theta_\subM$ on $\M$ is written in these coordinate as $\Theta_\subM = p_\alpha \tilde X^\alpha + p_i\tilde Y^i +\iota^*(p_a) \tilde Z^a$ and thus the horizontal gauge momenta $J_1,...,J_l$ (that are assumed $G$-invariant) are just 
\begin{equation}\label{Eq:pi=Ji}
J_i = \langle J^{\nh}, \eta_i \rangle = p_i.
\end{equation}


\begin{remark}
\label{R:Y_i_generador_inf}  Following Lemma \ref{L:Invariance}, since $\tilde X_\alpha$ are $\rho$-projectable, then the coordinates $p_\alpha$  are $G$-invariant. Since we also assume that $J_i=p_i$ are $G$-invariant,  the vector fields $\tilde Y_i$ and $(\eta_i)_\subM$ coincide. Indeed, by $G$-invariance we have $d p_\alpha ((\eta_i)_\subM) = d p_i ((\eta_i)_\subM) = 0$, and, by duality of the basis \eqref{Eq:basisT*M}, we also have $d p_\alpha (\tilde Y_i) = d p_i (\tilde Y_i) = 0$. Analogously, $\tilde Z_a = (\xi_a)_\subM$ also holds.
\end{remark}

%
%
%
Denoting by $\{\eta^i, \xi^a\}$ the basis of sections of $\g^*$ dual to $\{\eta_i, \xi_a\}$, we have  

\begin{lemma}
\label{lemma:K_W_in_sections} For $m=(q,p_i,p_\alpha,p_a) \in \M$, 
\begin{enumerate}
 \item[$(i)$] ${\mathcal K}_\subW = d^\C \tilde{\epsilon}^a \otimes \xi_a$ and $\langle J, \mathcal{K}_\subW\rangle = p_a d^\C \tilde \epsilon^a$.
 \item[$(ii)$] $\kappa_\g |_{\mathcal H} = \kappa_{\alpha i}\tilde X^\alpha \otimes \eta^i + \kappa_{\alpha a}\tilde X^\alpha \otimes \xi^a$.  
\end{enumerate}
\end{lemma}

\begin{proof}
$(i)$ The $\g$-valued $1$-form $\mathcal{A}_\subW$ given in  (\ref{definition_A_W}) is written as 
 $\mathcal{A}_\subW = \tilde \epsilon^a \otimes \xi_a =  h_a^I \tilde \epsilon^a \otimes \chi_I$ and thus we get $d^\C \mathcal{A}_\subW |_\C =  d (\tilde \epsilon^a \otimes \xi_a) |_\C = (d \tilde \epsilon^a \otimes \xi_a +  d h_a^I \wedge \tilde \epsilon^a \otimes \chi_I )|_\C = d \tilde \epsilon^a \otimes \xi_a |_\C.$   Now, since the $\g^*$-valued function $J:\M\to \g^*$ is written as $J = p_i\otimes \eta^i + p_a \otimes \xi^a$, using \eqref{Def:JK} we get that $\langle J, \mathcal{K}_\subW\rangle = p_a d^\C \tilde \epsilon^a$.
 
$(ii)$  Using that $(\eta_i)_Q = Y_i$ and $(\xi_a)_Q = Z_a$, from the definition of $\kappa_\g$ in \eqref{Def:kappag} one can directly verify this formula.
\end{proof}


We are interested now in writing the 2-form $B$ in suitable coordinates. Let us denote by $C_{\alpha \beta}^{\mbox{\tiny{$D$}}}$, $C_{\alpha i}^{\mbox{\tiny{$D$}}}$ and $C_{ij}^{\mbox{\tiny{$D$}}}$ the structure functions relative to the basis \eqref{eq:basis_TQ}, i.e., $C_{\alpha \beta}^\gamma = X^\gamma ([X_\alpha, X_\beta])$ and so on, for index ${\mbox{\scriptsize{$D$}}}$ including indices $\alpha, i, a$. 

\begin{proposition}\label{P:B-Coord}
Consider \ the \ nonholonomic \ system \ given \ by \ the \ triple \  $(\M, \pi_{\emph\nh}, \Ham_\subM)$ with a $G$-symmetry given by a free and proper action verifying the dimension assumption \eqref{dimension_assumption}. Assume that there are horizontal gauge symmetries $\{\eta_1,...,\eta_l\}$ that form a basis of sections of $\g_S$ with associated $G$-invariant horizontal gauge momenta $\{J_1,...,J_l\}$.  Then, in the basis  \eqref{Eq:basisT*M}, we get
\begin{enumerate}
 \item[$(i)$] the 2-form $B_1$ defined in \eqref{Eq:B1} is written as 
 $$
 B_1 = - \frac{1}{2} p_A C_{ij}^A \tilde Y^i \wedge \tilde Y^j - p_AC_{i \alpha}^A  \tilde Y^i \wedge \tilde X^\alpha  - p_A C_{\alpha\beta}^A \tilde X^\alpha \wedge \tilde X^\beta.
 $$
 \item[$(ii)$] The 2-form $\mathcal{B}$ given in \eqref{Eq:mathcalB} is written as  \ $\mathcal{B} = \frac{1}{2} (p_A C_{\alpha\beta}^A  + v^i\kappa_{\alpha A}C_{i \beta}^A )\tilde X^\alpha \wedge \tilde X^\beta$.
 \item[$(iii)$] The 2-form $B$ inducing the dynamical gauge transformation of Theorem \ref{T:GlobalB} is
\begin{equation} \label{Eq:B-Coord}
B = - \frac{1}{2} p_A C_{ij}^A \tilde Y^i \wedge \tilde Y^j - p_AC_{i \alpha}^A  \tilde Y^i \wedge \tilde X^\alpha + \frac{1}{2} v^i\kappa_{\alpha A}C_{i \beta}^A \tilde X^\alpha \wedge \tilde X^\beta,
\end{equation}
\end{enumerate}
where in all cases the index $A$ only includes indices $i$ and $a$.

\end{proposition}

\begin{proof}
First, observe that $d\tilde Y^i |_\C = - (\frac{1}{2}C_{jk}^i\tilde Y^j\wedge\tilde Y^k + C_{j\alpha}^i\tilde Y^j \wedge \tilde X^\alpha + \frac{1}{2}C_{\alpha\beta}^i\tilde X^\alpha \wedge\tilde X^\beta)$ and, analogously, $d\tilde \epsilon^a |_\C = -( \frac{1}{2}C_{jk}^a\tilde Y^j\wedge\tilde Y^k + C_{j\alpha}^a\tilde Y^j \wedge \tilde X^\alpha + \frac{1}{2}C_{\alpha\beta}^a\tilde X^\alpha \wedge\tilde X^\beta)$.  Then, in the basis \eqref{Eq:basisT*M}, we have that $\mathcal{A}_\subS^i = \tilde Y^i$ and thus using Lemma \ref{lemma:K_W_in_sections} we get, 
\begin{equation*}
 \begin{split}
  B_1 & = \langle J,\mathcal{K}_\subW \rangle + \langle J, d^\C \mathcal{A}_S^i \otimes \eta_i\rangle = p_a d^\C \tilde \epsilon ^a + p_id^\C \tilde Y^i \\
  & = - \frac{1}{2}(p_iC_{jk}^i + p_a C_{jk}^a)\tilde Y^j\wedge\tilde Y^k - (p_iC_{j\alpha}^i + p_aC_{i\alpha}^a)\tilde Y^j \wedge \tilde X^\alpha - \frac{1}{2}(p_iC_{\alpha\beta}^i + p_aC_{\alpha\beta}^a) \tilde X^\alpha \wedge\tilde X^\beta.
 \end{split}
\end{equation*}
To compute $\mathcal{B}$ we observe that the principal curvature is given by $\mathcal{K}_\V = d^\mathcal{H} {\mathcal A}_\V$. Then, since $\mathcal{A}_\V = \tilde Y^i \otimes \eta_i + \tilde \epsilon^a \otimes \xi_a$, we get $\langle J, \mathcal{K}_\V \rangle = \langle J , d^\mathcal{H} \mathcal{A}_\V \rangle =  p_i d^\mathcal{H} \tilde Y^i + p_a d^\mathcal{H} \tilde \epsilon^a  = - \frac{1}{2}(p_iC_{\alpha\beta}^i + p_aC_{\alpha\beta}^a) \tilde X^\alpha \wedge\tilde X^\beta. $
On the other hand, 
$\mathcal{K}_\subW + d^\C\mathcal{A}_S^i \otimes \eta_i =  d^\C \tilde \epsilon^a \otimes \xi_a + d^\C \tilde Y^i \otimes \eta_ i$ and then, using that $T\tau_\subM (P_\V (X_\nh)) = v^iY_i$ and Lemma \ref{lemma:K_W_in_sections} we get that 
$$
( \kappa_\g \wedge {\bf i}_{P_\V(X_\nh)} [ \mathcal{K}_\subW + d^\C\mathcal{A}_S^i \otimes \eta_i] )_{\mathcal{H}} =   ( \kappa_{\alpha i}\tilde X^\alpha \wedge {\bf i}_{P_\V(X_\nh)} d\tilde Y^i + \kappa_{\alpha a}\tilde X^\alpha \wedge {\bf i}_{P_\V(X_\nh)} d\tilde \epsilon^a  )_{\mathcal H} = -   v^i\kappa_{\alpha A}C_{i \beta}^A \tilde X^\alpha \wedge \tilde X^\beta.$$
Finally, since $B = B_1 + \mathcal{B}$ we obtain the desired result.  

\end{proof}


\begin{remark}\label{R:B:LuisJames}
%
%
In \cite{LGN-JM-16}, Garcia-Naranjo and Montaldi find, in suitable coordinates, a dynamical gauge transformation which turns first integrals into Casimirs for the reduced bracket based on the decomposition of the tangent bundle by $TQ = D \oplus D^\perp$, where $D^\perp$ is the orthogonal complement with respect to the kinetic energy metric.  For that purpose, the authors do not rely on the dimension assumption. On the other hand, as we saw in Theorem \ref{T:GlobalB}, this assumption is fundamental in order to conclude the integrability properties of $\pi_\red^1$ and $\pi_\red^\B$.  

Although our global formula for $B$ relies on the complement $W$ being vertical (hence we cannot set $W$ to be $D^\perp$ in general),
we observe that our local formulation \eqref{Eq:B-Coord} of the 2-form $B$ does not depend on this property of $W$. In fact,  consider  any decomposition $TQ = D \oplus R$, where $R$ is not necessarily vertical.  Then $TQ = H \oplus S \oplus R$, and we may consider a local basis adapted to this splitting as in \eqref{eq:basis_TQ} and \eqref{Eq:basisT*M}.  In this case, the 2-form $B$ will be given (locally) by the formula \eqref{Eq:B-Coord} but with the indices $A$ including $\alpha,i,a$. 
If we set $R=D^\perp$, then $p_a = 0$ and we recover the local formula given in \cite{LGN-JM-16}.

 We point out that taking the complement of the constraints $W$  as vertical is essential to allows us to work with a principal connection and the so-called $\W$-curvature. This intrinsic formulation not only leads to a geometric understanding of the gauge transformation and characterization of the almost symplectic foliations but it also simplifies  
 computations when working with examples since, in general, it is not necessary to compute the structure functions $C_{IJ}^K$ nor orthogonal complements. This will be illustrated in the next sections.

On the other hand, in \cite{LGN-JM-16} the 2-form $B$ is constructed from a locally defined 3-form contracted with $X_\nh$, so its compatibility with the dynamics is manifest.


%
\end{remark}

Next, we write the corresponding bivectors $\pi_1$ and $\pi_\B$ on $\M$ and their corresponding reductions $\pi_\red^1$ and $\pi_\red^\B$ on $\M/G$.
Following with the notation \eqref{Eq:basisT*M} and taking into account that $\tilde X_\alpha$ are $\rho$-projectable vector fields and that the functions $p_\alpha$ and $p_i$ are $G$-invariant, we denote by 
\begin{equation}\label{Eq:d-s:basisTMG}
{\mathfrak B}_{\mbox{\tiny{$T(\M/G)$}}} = \{{\mathcal X}_\alpha:=T\rho(\tilde{X}_\alpha), \partial_{\bar p_\alpha}, \partial_{\bar p_i} \} \qquad \mbox{and} \qquad 
{\mathfrak B}_{\mbox{\tiny{$T^*(\M/G)$}}} = \{{\mathcal X}^\alpha, d \bar p_\alpha, d \bar p_i \}, 
\end{equation}
the dual basis on $T(\M/G)$ and $T^*(\M/G)$, respectively, where $\bar p_\alpha$ and $\bar p_i$ are the associated reduced functions on $\M/G$, i.e. $p_\alpha = \rho^* \bar p_\alpha$ and $p_i = \rho^* \bar p_i$. 
One can check, following Lemma \ref{L:Invariance}$(ii)$ and since the basis \eqref{eq:basis_TQ} is taken to be $G$-invariant, that $C_{\alpha\beta}^{\mbox{\tiny{$D$}}}$ and $C_{j\beta}^A$  are $G$-invariant functions for index {\scriptsize{$D$}} including indices $\alpha,i,a$ and index $A$ including only $i,a$. Hence we denote by $\bar{C}_{\alpha\beta}^{\mbox{\scriptsize{$D$}}}$, $\bar{C}_{j\beta}^A$ the functions on $\M/G$ such that   $\rho^*\bar{C}_{\alpha\beta}^{\mbox{\scriptsize{$D$}}} = C_{\alpha\beta}^{\mbox{\scriptsize{$D$}}}$ and $\rho^*\bar{C}_{j\beta}^A = C_{j\beta}^A$ respectively.

\begin{corollary} Under the hypothesis of Prop.~\ref{P:B-Coord} and working in the basis \eqref{Eq:basisT*M} and \eqref{Eq:d-s:basisTMG} on $M$ and $\M/G$ respectively, we have
\begin{enumerate}
 \item[$(i)$] The gauged related bivector field $\pi_1$ described in Prop.~\ref{L:B_1} is given by $ \pi_1 = \tilde Y_i\wedge \partial_{p_i} + \tilde X_\alpha \wedge \partial_{p_\alpha} -\frac{1}{2} p_\gamma C_{\alpha\beta}^{\gamma} \partial_{p_\alpha}\wedge \partial_{p_\beta}$
 and the reduced Poisson bivector field has the form 
 $$
 \pi_{\emph\red}^1 = \mathcal{X}_\alpha \wedge \partial_{\bar{p}_\alpha} + \frac{1}{2}\bar{p}_\gamma \bar{C}_{\alpha\beta}^{\gamma} \partial_{\bar{p}_\alpha}\wedge \partial_{\bar{p}_\beta},
 $$ 
 with symplectic form $\Omega_\mu$ on the leaf $L_\mu$ given by $\Omega_\mu = \iota^*_\mu ( \mathcal{X}^\alpha \wedge d\bar p_{\alpha} + \frac{1}{2}  \bar{p}_\gamma \bar{C}_{\alpha\beta}^{\gamma} \mathcal{X}^\alpha\wedge \mathcal{X}^\beta ) =  - \iota^*_\mu d(\bar p_\alpha \mathcal{X}^\alpha)$.

\item[$(ii)$] The dynamically gauge related bivector field $\pi_\B$ described in Theorem \ref{T:GlobalB}, is given by 
 $$
 \pi_\B = \tilde Y_i\wedge \partial_{p_i} + \tilde X_\alpha \wedge \partial_{p_\alpha} + \frac{1}{2} \left( p_{\mbox{\tiny{$D$}}} C_{\alpha\beta}^{\mbox{\tiny{$D$}}} + v^i\kappa_{\alpha A} C_{j \beta}^A  \right) \partial_{p_\alpha}\wedge \partial_{p_\beta},
 $$
 where indices {\scriptsize{$D$}} include indices $\alpha,i,a$ while indices $A$ include $i,a$.    Moreover,  the reduced bivector field $\pi_{\emph\red}^\B$ also described in Theorem \emph{\ref{T:GlobalB}} is written in the basis \eqref{Eq:d-s:basisTMG} as
 $$
 \pi_{\emph\red}^\B = \mathcal{X}_\alpha \wedge \partial_{\bar p_\alpha} + \frac{1}{2}\left( \bar p_{\mbox{\tiny{$D$}}} \bar C_{\alpha\beta}^{\mbox{\tiny{$D$}}} + v^i\kappa_{\alpha A} \bar{C}_{j \beta}^A  \right) \partial_{\bar p_\alpha} \wedge \partial_{\bar p_\beta},
 $$
  On each leaf $L_\mu$ the almost symplectic form is given by 
$$
\tilde \Omega_\mu = \Omega_\mu + \iota^*_\mu(\bar{\mathcal{B}}) =  \iota^*_\mu \left( -  d(\bar p_\alpha \mathcal{X}^\alpha) + \bar{\mathcal{B}} \right),
$$
where $\bar{\mathcal B}$ is the 2-form on $\M/G$ such that $\rho^*\bar{\mathcal B}= \mathcal{B}$.
 \end{enumerate}

\end{corollary}

\begin{proof}
First, observe that, since the vector fields $\tilde X_\alpha$ are $\rho$-projectable we have that $[\tilde X_\alpha, \tilde Y_i] \in \Gamma(\V)$ and thus $C_{\alpha i}^\beta=0$. On the other hand,  by integrability of the vertical distribution we also have $[\tilde Y_i, \tilde Y_j] \in \Gamma(\V)$ which implies that $C_{ij}^\beta = 0$.  Then, we have that \begin{equation*}
\Omega_\subC = \left( \tilde X^\alpha \wedge d p_\alpha + \tilde Y^i \wedge d p_i + \frac{1}{2} ( p_A C_{ij}^A \tilde Y^i \wedge \tilde Y^j + p_{\mbox{\tiny{$D$}}} C_{\alpha \beta}^{\mbox{\tiny{$D$}}} \tilde X^\alpha \wedge \tilde X^\beta) + p_A C_{i\alpha}^A \tilde Y^i \wedge \tilde X^\alpha   \right)|_\C,
 \end{equation*}
where, index $A$ includes only the indices $i$ and $a$ while ${\mbox{\scriptsize{$D$}}}$ includes $\alpha$ too.
Then, from Prop.~\ref{P:B-Coord}$(i)$ we get that 
$\Omega_\C + B_1|_\C = \tilde X^\alpha \wedge dp_\alpha + \tilde Y^i \wedge dp_i + \frac{1}{2} p_{\gamma} C_{\alpha\beta}^{\gamma} \tilde X^\alpha \wedge \tilde X^\beta$. Then, using \eqref{Eq:GaugedNH} we get the local form for $\pi_1$.  
Next, we observe that 
 $\Omega_\C + B|_\C = \left( \tilde X^\alpha \wedge dp_\alpha + \tilde Y^i \wedge dp_i + \frac{1}{2} ( p_{\mbox{\tiny{$D$}}} C_{\alpha\beta}^{\mbox{\tiny{$D$}}} + v^i\kappa_{\alpha A} C_{j \beta}^A  ) \tilde X^\alpha \wedge \tilde X^\beta\right)|_\C$ and thus we get $\pi_\B$. By projecting the bivectors $\pi_1$ and $\pi_\B$ we get the local formula of $\pi_\red^1$ and $\pi_\red^B$. 
\end{proof}

\section{Examples revisited}
\label{ExamplesRevisited}

\subsection{The snakeboard}\label{Example_Snake}

The snakeboard is a variation of a standard skateboard permitting the axis of the wheels to rotate by the effect of a human rider creating a torque allowing the board to spin about a vertical axis, see \cite{Bloch:Book} and \cite{Bal-Fer, Oscar-11}. We denote by $m$ the mass of the board, by $R$ the distance from the center of the board to the pivot point of the wheel axle, by $J$ the inertia of the rotor and by $J_1$ the inertia of the each wheel.
The configuration manifold is $Q=SE(2) \times S^1 \times S^1$, with coordinates $q=(x,y,\theta, \phi,\psi)$, 
and the Lagrangian is the kinetic energy of the system given by
\begin{equation*}
L(q,\dot{q}) = \frac{1}{2} m (\dot{x}^2 + \dot{y}^2 + r^2 \dot{\theta}^2) + \frac{1}{2} J \dot{\psi}^2 + J \dot{\psi} \dot{\theta} + J_1 \dot{\phi}^2.
\end{equation*}

Following \cite{Bloch:Book}, the non-sliding condition induce the constraint 1-forms 
\begin{equation*}
\epsilon^1 = dx + R \cot \phi \cos \theta d\theta \qquad \text{and} \qquad
\epsilon^2 = dy + R \cot \phi \sin \theta d\theta ,
\end{equation*}
for $\phi \neq 0, \pi$, and thus the constraint distribution $D$ is given by
\begin{equation}
\label{D_snakeboard}
D = \textup{span} \{ X_\theta := \partial_{\theta} - R \cot \phi \cos \theta \partial_x - R \cot \phi \sin \theta \partial_y , \, X_\phi := \partial_\phi , \, Y_\psi := \partial_{\psi} \}.
\end{equation}

The free and proper action of the Lie group $G=\R^2 \times S^1$ on $Q$ given at each $(x,y,\theta,\phi,\psi) \in Q$ and $(a,b,\alpha)\in G$, by $(x+a,y+b,\theta,\phi,\psi+\alpha)$ defines a symmetry of the nonholonomic system.  Representing by $\g\simeq \R^2\times \R$ the Lie algebra of $G$, with $(1,0;0)_Q = \partial_x$ , $(0,1;0)_Q= \partial_y$ and $(0,0;1)_Q =Y_\psi$,  the distribution $S= D\cap V$ is given by $S=\textup{span} \{ Y_\psi\}$ and we can choose the vertical complement of the constraints to be $W = \textup{span}\{ \partial_x , \partial_y \}$.  Hence, we also see that $\{ \xi_1=(1,0;0), \xi_2=(0,1;0), \eta = (0,0;1)\}$ is an adapted basis to the splitting $\g =  \g_W \oplus \g_S$ where $\g_W \simeq \R^2$ and  $\g_S \simeq \R$.
Therefore, we have an adapted basis of $TQ = D\oplus W$ and its dual basis given by 
$$
\mathfrak{B}_{TQ} = \{ X_\theta, X_\phi, Y_\psi , \partial_x, \partial_y \} \qquad \text{and} \qquad
\mathfrak{B}_{T^*Q} =  \{ d \theta, d \phi, d \psi, \epsilon^1, \epsilon^2 \} ,
$$
with associated coordinates $(p_\theta,p_\phi,p_\psi, p_x, p_y)$ on $T^*Q$.
The constraint manifold $\mathcal{M} \subset T^*Q$ is given by
\begin{equation*}
\mathcal{M} = \{ (q, p_\theta,p_\phi,p_\psi, p_x, p_y) \  : \ p_x = - F(\phi) \cos\theta  (p_\theta - p_\psi) , \ p_y = - F(\phi)  \sin\theta  (p_\theta - p_\psi)\},
\end{equation*}
where $F(\phi) = mR\frac{\sin\phi \cos\phi }{m R^2 - J \sin^2 \phi }$. 
Then, following \eqref{Eq:basisT*M}, we denote by 
$$
\mathfrak{B}_{T^*\M}= \{\tilde d\theta, \tilde d\phi, \tilde d\psi, \tilde \epsilon^1, \tilde \epsilon^2, dp_\theta, dp_\phi, dp_\psi\} \qquad \mbox{and} \qquad
\mathfrak{B}_{T\M}= \{\tilde X_\theta, \tilde X_\phi, \tilde Y_\psi, \tilde \partial_x, \tilde \partial_y, \partial_{p_\theta}, \partial_{p_\phi}, \partial_{p_\psi}\},
$$
the respectively dual basis where $\tilde d\theta = \tau_\subM^*d\theta$ and analogously with the other elements of the basis. 

The lifted $G$-action to $T^*Q$ leaves the manifold $\M$ invariant and thus it is induced a smooth reduced manifold $\M/G$ diffeomorphic to $\mathbb{T}^2 \times \R^3$ and also $Q/G \simeq \mathbb{T}^2$. 

Now we observe that the function $J_\eta = {\bf i}_{\tilde Y_\psi} \Theta_\subM =  p_\psi$ is a horizontal gauge momentum of the nonholonomic system, see e.g., \cite{Bloch:Book}. Therefore, by Theorem \ref{T:GlobalB} (and Cor.~\ref{C:l+1}) we conclude that there is  a (unique)  regular reduced bracket $\pi_\red^\B$ on $\M/G$ that describes the dynamics and is a twisted Poisson bracket with almost symplectic leaves of dimension 4 (note that $\textup{rank}(H) = 2$).  
In order to compute the 2-form $B$,  we start considering the principal connection $\mathcal{A}_\V = \tilde\epsilon^1 \otimes \xi_1 +\tilde \epsilon^2\otimes \xi_2 + \tilde{d\psi} \otimes \eta$ and then, for $X,Y\in T\M$
$$
\langle J, \mathcal{K}_\subW \rangle = p_x d^\C \tilde \epsilon^1 + p_x  d^\C \tilde\epsilon^2 \quad \mbox{and} \quad \langle J,\mathcal{K}_\V\rangle(X,Y) = \langle J, \mathcal{K}_\subW \rangle (P_{\mathcal H}X, P_{\mathcal H}Y),.
$$
since $d^{\mathcal{H}}(\tilde{d\psi}) = 0$. Moreover, since 
$$
\langle J , \mathcal{K}_\subW \rangle = - R \frac{ F(\phi)}{ \sin^2 \phi} ({p_\theta}-{p_\psi}) \tilde d\theta \wedge \tilde d\phi
$$
is a basic 2-form with respect to the bundle $\rho: \M \to \M/G$ we conclude that $\langle J,\mathcal{K}_\V\rangle = \langle J, \mathcal{K}_\subW \rangle$.  Therefore, using again that $d^\C\mathcal{A}_\subS^\psi = d^\C(\tilde{d\psi}) = 0$ and that ${\bf i}_{P_\V(X_\nh)}\mathcal{K}_\subW = 0$ (since $\mathcal{K}_\subW$ is semi-basic w.r.t. $\rho:\M\to \M/G$), we obtain that  
\begin{equation*}
 B_1 = \langle J,\mathcal{K}_\subW\rangle \quad \mbox{and} \quad \mathcal{B} = - \langle J, \mathcal{K}_\subW \rangle
 \end{equation*}
and therefore $B = 0$. Hence, the reduced bracket $\pi_\red$ induced by the nonholonomic bracket $\pi_\nh$ is twisted Poisson in agreement with \cite{Bal-14}.  Moreover, from Theorem \ref{T:GlobalB}, $\pi_\red$ is a twisted Poisson bracket for the 3-form $(-d\bar{\mathcal B})$ where $\rho^*\bar{\mathcal B}= \langle J, \mathcal{K}_\subW\rangle$ with (a 4-dimensional) almost symplectic foliation described by the level sets of $J_\eta = p_\psi$, as it was already observed in \cite{Bal-14} (see also Prop. \ref{Prop:JacB}).


\subsection{The Chaplygin ball}
\label{Example_Chaplygin_ball}

An important example of a rigid body with nonholonomic constraints is the Chaplygin sphere which consists of an inhomogeneous sphere of radius $R$ whose geometric center coincides with its center of mass and which is allowed to roll without sliding over a horizontal plane. This example was first studied by Chaplygin in \cite{Chap-03} but the hamiltonian structure of the reduced equations of motion has been proved only in 2002 by Borisov and Mamaev \cite{BorMa2001}. Later the Hamiltonization of this example has been geometrically understood using gauge transformations \cite{PL2011,Naranjo2008}. Next we recall briefly the geometrical framework of this example following \cite{Bal-16,PL2011,Naranjo2008} and we write the 2-form $B$ following Diagram \eqref{Diagram} and Theorem \ref{T:GlobalB}.

The configuration manifold, given by $Q=SO(3)\times\R^2$, is determined by the coordinates $(g,x,y)$ where $g\in SO(3)$ is the orthogonal matrix that represents the orientation of the
ball by relating the orthogonal frame attached to the body with the one fixed in space and $(x,y)$ is the projection of the  center of mass of the ball to the horizontal plane.


If $\bm{\omega}=(\omega_1,\omega_2,\omega_3)$ and  $\bm{\Omega} = (\Omega_1, \Omega_2, \Omega_3)$ denote the angular velocity of the ball with respect to the space frame and body frame respectively, then $\bm{\omega} = g \bm{\Omega}$ and the non-sliding constraints can be written as 
$$
\dot{x} = R \langle \vecbeta , \bm{\Omega} \rangle \quad \text{and} \quad \dot{y} = - R \langle \vecalpha , \bm{\Omega} \rangle,
$$
where $\vecalpha$, $\vecbeta$ denote the first and second rows of the matrix $g$.

We assume that the body frame is aligned with the principal axes of inertia of the body.  Thus, the {\it inertia tensor} $\I$ is represented by a diagonal matrix with positive entries $I_1,I_2,I_3$ (which are the principal moments of inertia of the ball). The Lagrangian is the total kinetic energy, which after left-trivialization $TQ \simeq Q \times \mathfrak{so}(3) \times T \R^2 \simeq Q \times \R^3 \times \R^2$, takes the form,
\begin{equation*}
L(g,x,y,\bm{\Omega},\dot{x},\dot{y}) = \frac{1}{2} \bm{\Omega}^t \I \bm{\Omega} + \frac{1}{2} m( \dot x^2 + \dot y^2).
\end{equation*}
where $\bm{\Omega}^t$ denotes the transpose of $\bm{\Omega}$. 
Denoting by $\{X_1^L, X_2^L,X_3^L\}$ the left-invariant frame of $SO(3)$, then the constraint distribution $D$ is given by
\begin{equation*}
D = \textup{span} \{ X_1:=  X_1^L + R \beta_1 \partial_x - R \alpha_1 \partial_y , X_2:=X_2^L + R \beta_2 \partial_x - R \alpha_2 \partial_y, X_3:= X_3^L  + R \beta_3 \partial_x - R \alpha_3 \partial_y \}.
\end{equation*}

The nonholonomic system admits a symmetry group $G = \{(h,a) \in SO(3) \times \R^2 : h \bm{e_3} = \bm{e_3} \} \simeq SO(2) \ltimes \R^2$ which acts on a point $(g,x,y) \in Q$ as
$(h,a) \cdot (g,x,y) = (h g, \tilde h (x,y)^t + a)$, where $\tilde h$ is the $2\times 2$ matrix defined by $h = {\mbox{\scriptsize{$\left(\begin{array}{cc} \tilde h & 0 \\[-6pt] 0 & 1 \end{array}\right)$}}} $.  Using the identification $\g \simeq \R \times \R^2$ and denoting by $\bm{X} = (X_1^L, X_2^L, X_3^L )$ and $\vecgamma = (\gamma_1, \gamma_2, \gamma_3)$ the third row of the matrix $g$, we compute the infinitesimal generators and thus vertical space $V$ is given, at each $q\in Q$, by
\begin{equation*}
V_q = \textup{span} \{ (1;0,0)_Q = \langle \bm{\gamma} , \bm{X} \rangle - y \partial_x + x \partial_y, (0;1,0)_Q = \partial_x, (0;0,1)_Q = \partial_y \},
\end{equation*}
Hence, the system verifies the dimension assumption \eqref{dimension_assumption} and we choose the vertical complement of the constraints to be $W := \textup{span} \{ \partial_x, \partial_y \}$. 
Then, following \eqref{Eq:basisT*M} we get the dual basis (adapted to the decomposition $TQ= D\oplus W$)
$$
\mathfrak{B}_{TQ} = \{X_1, X_2, X_3, \partial_x, \partial_y\} \qquad \mbox{and} \qquad \mathfrak{B}_{T^*Q} = \{ \lambda_1,\lambda_2,\lambda_3, \epsilon^x , \epsilon^y \},
$$
where  $\{\lambda_1,\lambda_2,\lambda_3\}$ is the dual basis  of $\{X_1^L, X_2^L, X_3^L\}$ (the left Maurer-Cartan 1-forms) and 
\begin{equation*}
\epsilon^x = dx - R \langle \vecbeta, \bm{\lambda} \rangle \quad \text{and} \quad \epsilon^y = dy + R \langle \vecalpha, \bm{\lambda} \rangle 
\end{equation*}
are the constraint 1-forms, for $\bm{\lambda} = (\lambda_1, \lambda_2, \lambda_3)$. Let us denote by  $(M_1, M_2, M_3,p_x, p_y)$ the coordinates on $T^*Q$ associated to this basis.
Then the constraint manifold $\M = Leg(D) \subset T^*Q$ is given by
\begin{equation*}
\M = \{ (g,x,y; M_1, M_2, M_3 ,p_x, p_y) : \: p_x = mR\langle \vecbeta, \vecOm \rangle, \ \ \ p_y = -mR \langle \vecalpha, \vecOm\rangle \} ,
\end{equation*}
where $\bm{M} =  (\I + mr^2\textup{Id}) \vecOm + mr^2 \langle \vecgamma , \vecOm \rangle \vecgamma$ for $\bm{M}= (M_1, M_2, M_3)$ and $\textup{Id}$ the $3\times 3$ identity matrix.  

The $G$-action on $T^*Q$ leaves the manifold $\M$ invariant and thus reduced dynamics takes place in $\M/G \simeq \mathbb{S}^2\times \R^3$ with (redundant) coordinates given by $(\gamma_1, \gamma_2, \gamma_3, M_1, M_2,M_3).$ As it was proven in \cite{Bal-14}, the reduced bracket $\pi_\red$ is not Poisson (nor twisted Poisson) since the 3-form $d\langle J, \mathcal{K}_\subW\rangle$ is not basic with respect to the principal bundle $\rho:\M\to \M/G$ (see Prop.~\ref{Prop:Jac}), however this example is hamiltonizable via a gauge transformation, see \cite{Naranjo2008,PL2011}. Next, we recover the dynamical gauge transformation using our theory.  

In particular, observe that $S:=D\cap V$ is given, at each $q\in Q$, by $S_q = \textup{span}\{Y_{(q)}:=\langle \vecgamma, \bm{X}  \rangle\}$ while $(\g_S)_q = \textup{span}\{\eta_{(q)} := (1;-y,x)\}$ and $\g_W = \textup{span}\{\xi_x := (0;1,0), \xi_y:=(0;0,1)\}$.  Observe that $\eta_Q = Y$ and that the function $J_\eta = \langle \bm{\gamma} , \bm{M}\rangle$ is a $G$-invariant horizontal gauge momentum with horizontal gauge symmetry given by $\eta$.  
Therefore by Theorem \ref{T:GlobalB} we conclude that there is a (unique) dynamical gauge transformation by the 2-form $B$ defined in \eqref{Eq:GlobalB} generating a regular twisted Poisson bracket on $\M/G$ with a 4-dimensional foliation given by the level sets of $J_\eta$ (observe that $\textup{rank}(H) = 2$, see Corollary \ref{C:l+1}). In order to write the 2-form $B$, we first  consider the principal connection $A_V:TQ \to \g$ given by $A_V = \epsilon^v\otimes \eta + \epsilon^x \otimes \xi_1 + \epsilon^y \otimes \xi_2,$
where $\epsilon^v = \langle \bm{\gamma}, \bm{\lambda}\rangle$. Then, we define the horizontal space $H : = \textup{Ker}A_V$ so that $TQ = H \oplus S \oplus W$. We also consider the $\g$-valued 1-forms $\mathcal{A}_\V := \tau_\subM^* A : T\M \to \g$ and $\mathcal{A}_\subS:T\M \to \g$ given by $\mathcal{A}_\subS = \tilde \epsilon^v\otimes \eta$
where $\tilde \epsilon^v = \tau_\subM^*\epsilon^v$. Following \eqref{Eq:B1} and \eqref{Eq:mathcalB}, we get that 
\begin{equation*}
B_1  = \langle J, \mathcal{K}_\subW \rangle + J_\eta d^\C\! \tilde \epsilon^v \quad \mbox{and} \quad 
\mathcal{B} = -\langle J, \mathcal{K}_\V\rangle -  {\textstyle \frac{1}{2}} (\kappa_\g\wedge {\bf i}_{P_\V(X_\nh)} [\mathcal{K}_\subW + d^\C \tilde \epsilon^v \otimes \eta ] )_{\mathcal H} .
\end{equation*}
Observe that, for $U_1, U_2\in T\M$,  $\mathcal{K}_\V (U_1, U_2)= [ \mathcal{K}_\subW+ d^\C \epsilon^v \otimes \eta](P_{\mathcal H}U_1, P_{\mathcal H}U_2) $.   Now, using the identity $d\bm{\gamma} = \bm{\gamma} \times \bm{\lambda}$ and the  notation $d\bm{\lambda} = (d\lambda_1, d\lambda_2, d\lambda_3)$ we observe that $d^\C \tilde \epsilon^v = - \tau^*_\subM\langle \bm{\gamma},d \bm{\lambda}\rangle = \rho^* ( \gamma_1 d\gamma_2\wedge d\gamma_3 +  \gamma_2 d\gamma_3\wedge d\gamma_1 + \gamma_3 d\gamma_1\wedge d\gamma_2 )$ showing that $d^\C \tilde\epsilon^v$ is a basic form with respect to $\rho:\M\to \M/G$ and therefore 
\begin{equation}\label{Eq:ExChap:Bglobal}
B (U_1, U_2) =  \langle J, \mathcal{K}_\subW \rangle (U_1, U_2) -  \langle J,\mathcal{K}_\subW\rangle(P_{\mathcal H}U_1, P_{\mathcal H}U_2) -   \langle \bm{\gamma} , \bm{\Omega} \rangle \langle \kappa_\g (P_{\mathcal H}U_1)  ,  \mathcal{K}_\subW (Y , P_{\mathcal H}U_2) \rangle,
\end{equation}
where we also use that $P_{\V}(X_\nh) = \tilde\epsilon^v(X_\nh) \tilde Y$ for $\tilde Y = \eta_\subM$.

Next, we compute the 2-form $B$ given in \eqref{Eq:ExChap:Bglobal} in local coordinates.  First, using Lemma \ref{lemma:K_W_in_sections} we write $\langle J, \mathcal{K}_\subW \rangle = R^2 m \langle \bm{\Omega}, d\tilde{\bm{\lambda}}  \rangle - R^2 m \langle \vecgamma, \bm{\Omega} \rangle \langle \bm{\gamma}, d\tilde{\bm{\lambda}}  \rangle,$ where $\tilde \lambda_i = \tau_\subM^*\lambda_i$.
Now let us choose a basis of generators of the distribution $H$ given, at each $q\in Q$ such that $\gamma_3 \neq 0$, by $\mathbb{X}_1:= X_1 - \gamma_1Y$ and $\mathbb{X}_2 := X_2 - \gamma_2 Y$. Then we have a (local) basis of $\Gamma(TQ)$ given by $\{\mathbb{X}_1, \mathbb{X}_2, Y, \partial_x, \partial_y\}$ and its dual basis given by $\{\mathbb{X}^1, \mathbb{X}^2, \epsilon^v, \epsilon^x, \epsilon^y\}$ for $\mathbb{X}^1 := \gamma_3^{-1}(-\gamma_1\lambda_3 + \gamma_3\lambda_1)$ and $\mathbb{X}^2 := \gamma_3^{-1}(\gamma_2\lambda_3+\gamma_3\lambda_2)$.
Therefore, following the notation of \eqref{Eq:basisT*M}, i.e., $\tilde{\mathbb{X}}^i= \tau_\subM^*\mathbb{X}^i$ and $\tilde{\mathbb{X}}_i , \tilde Y \in \mathfrak{X}(\M)$ such that $T\tau_\subM(\tilde{\mathbb{X}}_i) = \mathbb{X}_i, T\tau_\subM(\tilde Y) = Y$,  from \eqref{Eq:ExChap:Bglobal} we see that
\begin{equation*}
 \begin{split}
  \ & B(\tilde{\mathbb{X}}_1,\tilde Y) =  \langle J, \mathcal{K}_\subW \rangle(\tilde{\mathbb{X}}_1,\tilde Y) = R^2 m \langle \bm{\Omega}, d\bm{\lambda}  \rangle (\mathbb{X}_1,Y)= -R^2m (\bm{\gamma}\times \bm{\Omega})_{2} \mbox{ \ and \ } B(\tilde{\mathbb{X}}_2,Y) = - R^2m (\bm{\gamma}\times \bm{\Omega})_{1} \\
\ &  B(\tilde{\mathbb{X}}_1, \tilde{\mathbb{X}}_2) = - \langle \bm{\gamma} , \bm{\Omega} \rangle [\kappa(\mathbb{X}_1, \partial_x) d\epsilon^x + \kappa(\mathbb{X}_1, \partial_y) d\epsilon^y]( Y  , \mathbb{X}_2) = - R^2m \langle \bm{\Omega}, \bm{\gamma}\rangle \gamma_3.  
 \end{split}
\end{equation*}
Then, since $B = B(\tilde{\mathbb{X}}_i,\tilde{Y})\tilde{\mathbb{X}}^i \wedge \tilde \epsilon^v + B(\tilde{\mathbb{X}}_1, \tilde{\mathbb{X}}_2)\tilde{\mathbb{X}}^1\wedge \tilde{\mathbb{X}}^2$, it is straightforward to obtain that $B = R^2m\langle \bm{\Omega}, d\tilde{\bm{\lambda}}\rangle$ coinciding with the result in \cite{Naranjo2008}. 

%
%
%
%

Using the almost symplectic structure on each leaf, it is possible to compute a conformal factor and extend it to the whole foliation obtaining a conformal factor for the twisted Poisson bracket as it was done in \cite{Bal-Fer}.


\subsection{A solid of revolution rolling on a plane}\label{Sec:Solids}

Following \cite{Bal-16, CDS-2010} we consider a strictly convex body of revolution rolling without sliding over a horizontal plane which is described by $\{ z = 0 \}$. We denote by $m$ the mass of the body and by $\I=(I_1,I_1,I_3)$ the moment of inertia with respect to an orthonormal frame $(\bm{e_1},\bm{e_2},\bm{e_3})$ attached to the body where $\bm{e_3}$ is oriented along the axis of symmetry of the body. This set of examples includes the celebrated Routh sphere \cite{Cushman-RouthSph, Routh-55} and the axisymmetric ellipsoid \cite{Bol-Kil-Kaz-14, BorMa2002b}.  In this case, we will follow the geometric considerations and notations from \cite{Bal-16,CDS-2010}.  This group of examples is known to be hamiltonizable by a gauge transformation by a 2-form (see \cite{Bal-16,LGN-JM-16}). In this section we recover the 2-form $B$ using our theory.

Let us denote by $\mathbf{x}=(x,y,z) \in \R^3$ the coordinates of the center of mass of the body and by $g$ the orthogonal matrix indicating its orientation. During this section we maintain, when possible, the same notation as in Sec.~\ref{Example_Chaplygin_ball}. The configuration manifold of the free mechanical system is given by $Q_0 = \R^3 \times SO(3)$ with coordinates $(\mathbf{x},g)$ and the Lagrangian $L : TQ_0 \to \R$ (of mechanical type) is given by
\begin{equation}
\label{Lagrangian_body}
L_0(\mathbf{x}, g , \dot{\mathbf{x}} , \bm{\Omega} ) = \frac{1}{2} m \langle \dot{\mathbf{x}} , \dot{\mathbf{x}} \rangle + \frac{1}{2} \langle \I \bm{\Omega} , \bm{\Omega} \rangle - m {\bf g} \langle \mathbf{x} , \bm{e_3} \rangle ,
\end{equation}
where $\bm{\Omega} = (\Omega_1,\Omega_2,\Omega_3)$ is the angular velocity in the body frame, $\langle \cdot,\cdot \rangle$ denotes the standard scalar product in $\R^3$ and ${\bf g}$ is the acceleration of gravity.   In this case, denoting by $\bm{\gamma} = (\gamma_1 , \gamma_2 , \gamma_3 )$ the third row of the matrix $g$, the configuration manifold $Q$ (representing the solid on the plane) is given by 
\begin{equation*}
Q = \{ (\mathbf{x}, g) \in \R^3 \times SO(3) : z = - \langle \bm{s} , \bm{\gamma} \rangle \}, 
\end{equation*}
where $\bm{s} : S^2 \to \bm{S}$ denotes the inverse of the Gauss map, for $\bm{S}$ the surface describing the body; that is
\begin{equation*}
\bm{s} = \bm{s}(\bm{\gamma}) = (\varrho(\gamma_3)\gamma_1, \varrho(\gamma_3)\gamma_2, \zeta(\gamma_3)),
\end{equation*} 
with $\varrho: (−1, 1) \to \R$ and $\zeta: (−1, 1) \to \R$ smooth functions depending on the parametrization of the body of revolution, see \cite[Sec. 6.7.1]{CDS-2010}.

As usual, we denote by $\{ X_1^L , X_2^L , X_3^L ,\partial_{x} , \partial_{y} \}$ the local basis of $TQ$ where  $\{X_1^L , X_2^L , X_3^L\}$ are the left-invariant vector fields on $SO(3)$ with $(\bm{\Omega} , \dot{x}, \dot{y} )$ representing the corresponding velocity coordinates. Then, the non-sliding constraints are written as $\dot{x} = - \langle \bm{\alpha} , \bm{\Omega} \times \bm{s} \rangle$ and $\dot{y} = - \langle \bm{\beta} , \bm{\Omega} \times \bm{s} \rangle$ (recall that $\bm{\alpha}$ and $\bm{\beta}$ are the first and second rows of the matrix $g$) and thus the constraint 1-forms are given by 
\begin{equation*}
\epsilon^x = dx + \langle \bm{\alpha} , \bm{\lambda} \times \bm{s} \rangle, \quad
\epsilon^y = dy + \langle \bm{\beta} , \bm{\lambda} \times \bm{s} \rangle,
\end{equation*}
where $\bm{\lambda} =  (\lambda_1, \lambda_2, \lambda_3)$ and $\{\lambda_1, \lambda_2, \lambda_3\}$ is the dual basis of $\{ X_1^L , X_2^L , X_3^L\}$.  Then, the constraint distribution $D$ is generated by the vector fields $X_i := X_i^L + (\bm{\alpha} \times \bm{s})_i \partial_{x} + (\bm{\beta} \times \bm{s})_i \partial_{y} + (\bm{\gamma} \times \bm{s})_i \partial_{z}$, for $i=1,2,3$, that is 
\begin{equation*}
D = \textup{span} \{ X_1, X_2, X_3 \}.
\end{equation*}

\noindent {\bf The symmetries and the vertical complement $W$}.  Consider the action of the Lie group $SE(2) \times S^1$ on $Q$ given, for each $(a,b,\phi,\theta) \in SE(2) \times S^1$ and $(x, y, g) \in Q$, by
$$
(a,b,\phi,\theta) \cdot (x, y, g) = (R_\phi (x,y)^t + (a,b)^t, \hat{R_\phi} g \hat{R_{-\theta}}),
$$
where $R_\phi$ is a $2 \times 2$ rotation matrix of angle $\phi$ and $\hat{R_\phi}$ and $\hat{R_{-\theta}}$ denote $3 \times 3$ orthogonal matrices representing rotations about the $z$-axis of angles $\phi$ and $-\theta$ respectively. This $G$-action defines a symmetry of the nonholonomic system, \cite{CDS-2010}.

The Lie algebra of $G$ denoted by $\g$ is identify with $\g \simeq \R^2 \times \R \times \R$ with the infinitesimal generator relative to the $S^1$-action given by $((0,0), 0;1)_Q = - X_3^L - y \partial_{x} + x \partial{y}$, while relative to the $SE(2)$-action are
$$
((1,0),0;0)_Q =  \partial_{x}, \qquad ((0,1),0,0)_Q = \partial_{y}, \qquad ((0,0),1,0)_Q =  \langle \vecgamma, \bm{X}^L \rangle - y \partial_{x} + x \partial_{y}.
$$

We observe that the rank of the vertical space $V$ is not constant. In fact, for  $\bm{\gamma}= (0,0,\pm 1)$ we have $\textup{rank}(V)= 3$. However, we see that the dimension assumption is satisfied and we can also compute the (generalized) distribution $S  = D \cap V = span \{ \mathcal{Y}_1:= - X_3, \mathcal{Y}_2:=\langle \vecgamma, \bm{X} \rangle \}$, where $\bm{X} = (X_1,X_2,X_3)$. Moreover, following \cite{Bal-16} we choose the vertical complement $W$ defined by $W=\textup{span} \{ \partial_{x},  \partial_{y} \}$.  Hence the bundle $\g_\subW \to Q$ is generated by the sections $\xi_1 = ((1,0),0;0)$ and $\xi_2 = ((0,1),0,0)$ while the bundle $\g_S \to Q$ is generated by the sections $\zeta_1 = P_{\g_S}(((0,0), 0;1))$ and $\zeta_2 = P_{\g_S}(((0,0), 1;0))$ where $P_{\g_S}: \g \to \g_\subS$ is the projection associated to the splitting $\g\times Q = \g_S \oplus \g_W$. Observe that $(\zeta_1)_Q = \mathcal{Y}_1$, and $(\zeta_2)_Q = \mathcal{Y}_2$.   

\begin{remark} \label{R:Solids:VertSym}
The bundle $\g_W \to Q$ is generated by the Lie algebra $\R^2$ and thus we see that this vertical complement satisfies the {\it vertical symmetry condition} \cite{Bal-16}. Even though this condition may simplify some computations, we are not using it here to build the 2-form $B$.    
\end{remark}

Now, we set a basis of $TQ$ adapted to the splitting $TQ=D\oplus W$ and its corresponding dual basis, 
$$
\mathfrak{B}_{TQ} = \{X_1, X_2,X_3, \partial_{x}, \partial_{y}\} \quad \mbox{and} \quad \mathfrak{B}_{T^*Q} = \{\lambda_1, \lambda_2,\lambda_3, \epsilon^x, \epsilon^y\},
$$
with coordinates on $T^*Q$ given by $(M_1, M_2,M_3, p_x, p_y)$. Using the kinetic energy metric given in \eqref{Lagrangian_body} we compute the constraint manifold $\M = \kappa^\sharp(D)$,
$$
\M = \{ (x,y, g , M_1, M_2, M_3 , p_x, p_y) : p_x = m \langle \bm{\alpha} , \bm{s} \times \bm{\Omega} \rangle , \quad p_y = m \langle \bm{\beta} , \bm{s} \times \bm{\Omega} \rangle \}, 
$$
where ${\bm{M}} = \I \bm{\Omega} + m \bm{s} \times (\bm{\Omega} \times \bm{s})$ for $\bm{M} =(M_1, M_2, M_3)$.

The lifting of the $G$-action to $T^*Q$ leaves the manifold $\M$ invariant and thus the restricted action on $\M$ is given, at each  $(a,b,\phi,\theta) \in G$ and $((x,y), g, \bm{M}) \in \M$, by
\begin{equation*}
\label{body_action_on_M}
(a,b,\phi,\theta) \cdot ((x,y), g, \bm{M}) = (R_\phi (x,y)^T + (a,b)^T, \hat{R_\phi} g \hat{R_\theta}, \hat{R_\theta} \bm{M}).
\end{equation*}

Since the $G$-action on $\M$ is proper the quotient $\M / G$  is a stratified differential space. Reducing by stages, observe that $SE(2)$ is a normal subgroup of $G$ and the $SE(2)$-action is free and proper. Then, the quotient $\M / SE(2)$ is a manifold which is diffeomorphic to $S^2 \times \R^3$ with coordinates $(\vecgamma, \bm{M})$. The $S^1$-action on $\M / SE(2)$ is not free and we describe the resulting differential space $\M / G$ using {\it invariant theory}, see \cite{Bal-16,CDS-2010}. The ring of $S^1$-invariant polynomials in  $S^2 \times \R^3$ is generated by
\begin{equation*}
\begin{split}
\tau_1 = & \gamma_3, \quad \tau_2 = \gamma_1 M_2 - \gamma_2 M_1, \quad \tau_3 = \gamma_1 M_1 + \gamma_2 M_2, \\
\tau_4 = & M_3, \quad \tau_5 = M_1^2 + M_2^2,
\end{split}
\end{equation*}
and the quotient space $\M /G$ is represented by the following semi-algebraic set of $\R^5$
\begin{equation*}
\M/G = \{ (\tau_1,\tau_2,\tau_3,\tau_4,\tau_5) \in \R^5 \: : \: |\tau_1|\leq 1, \tau_5 \geq 0, \tau_2^2 + \tau_3^2 = (1-\tau_1^2) \tau_5 \}.
\end{equation*}
Thus, the 1-dimensional {\it singular stratum} of $\M / G$ associated to $S^1$-isotropy type is given by
\begin{equation*}
\{ (\pm 1, 0, 0, \tau_4, 0) \in \M/G \: | \: \tau_4 \in \R \},
\end{equation*}
and corresponds to the configuration where the body of revolution is spinning over one of the two poles which remains fixed on the plane.
Since there are no other isotropy types, the regular stratum is the complementary 4-dimensional manifold $\M_{reg} / G$, where $\M_{reg}$ denotes the submanifold where the action is free.

We recall from \cite{Bal-16} that the reduced nonholonomic bracket $\{\cdot, \cdot\}_\red$ is not Poisson on $\M/G$, however, this example is known to admit a dynamical gauge transformation so that the resulting bracket $\{\cdot, \cdot\}_\red^\B$ is Poisson.  In what follows, we apply Theorem \ref{T:GlobalB} to recover the corresponding 2-form $B$.

\noindent {\bf First Integrals}.   First, observe that the sections $\mathcal{Y}_1 = -X_3$ and $\mathcal{Y}_2= \langle \bm{\gamma}, {\bf X} \rangle $ generating $S$ do not induce horizontal gauge symmetries since, for $\tilde{\mathcal Y}_i = (\zeta_i)_\subM$, the $G$-invariant functions
$$
\mathcal{J}_1 = {\bf i}_{\tilde{\mathcal{Y}}_1} \Theta_\subM = -M_3 \qquad \mbox{and} \qquad \mathcal{J}_2 = {\bf i}_{\tilde{\mathcal{Y}}_2}\Theta_\subM = \langle \bm{\gamma},{\bf M}\rangle   
$$
are not first integrals of $X_\nh$.  However, the system admits two $G$-invariant horizontal gauge momenta $J_1$ and $J_2$, that depend linearly (over $Q$) on $\mathcal{J}_1$ and $\mathcal{J}_2$, that is
$$
\left(  \begin{smallmatrix}  J_1\\ J_2  \end{smallmatrix} \right)  = F  
\left(  \begin{smallmatrix}  \mathcal{J}_1\\ \mathcal{J}_2  \end{smallmatrix} \right) \qquad \mbox{where} \quad F=F(\gamma_3) =
\left(  \begin{smallmatrix}  f_1(\gamma_3) & \ g_1(\gamma_3) \\ f_2(\gamma_3) & \ g_2 (\gamma_3)   \end{smallmatrix} \right)
$$
where $(f_1, g_1)$ and $(f_2, g_2)$ are the two independent solutions on $Q/G$ of a linear system of ordinary  differential equations on $Q/G$, that is, $f_i, g_i\in C^\infty(Q)^G$ (recall that $\textup{dim}(Q/G) =1$, see \cite{Bal-16} to see the detailed ODE system).

Then, on $\M_{reg}$, we are under the hypothesis of Theorem \ref{T:GlobalB}. Moreover, in this case, the horizontal space associated to a principal connection has rank 1 and thus, following Corollary \ref{C:l+1}, the gauge transformation by $B_1$ (defined in \eqref{Eq:B1}) generates a reduced bracket $\{\cdot, \cdot\}_\red^1$ that describes the (reduced) dynamics and is Poisson with 2-dimensional symplectic leaves.

\noindent {\bf The dynamical gauge transformation}. In order to use the formula \eqref{Eq:B1} to write the 2-form $B=B_1$ on $\M_{reg}$, we recall from \cite{Bal-14} (and Lemma \ref{lemma:K_W_in_sections}) that 
\begin{equation*}
  \langle J, \mathcal{K}_\subW \rangle  = p_x d^\C \tilde \epsilon^x + p_y d^\C \tilde \epsilon^y = m\varrho \langle \bm{\gamma}, \bm{s}\rangle \langle \bm{\Omega},d\tilde{\bm{\lambda}}\rangle +  \mathcal{Q} \langle \bm{\gamma}, d\tilde{\bm{\lambda}}\rangle + \mathcal{P} d\tilde{\lambda}_3,
\end{equation*}
where $\mathcal{Q}$ and $\mathcal{P}\in C^\infty(\M)^G$ are given by $\mathcal{Q} = −m( \varrho^2 \langle \bm{\Omega},\bm{\gamma} \rangle + \varrho' c_3)$ and $\mathcal{P}  = m(L \varrho \langle \bm{\Omega},\bm{\gamma} \rangle + L' c_3)$ for $L = L(\gamma_3) = \varrho . \gamma_3 -\zeta$,  $c_3$ the third component of $\bm{\gamma} \times (\bm{\Omega} \times \bm{s})$ and $L' = \frac{d}{d\gamma_3} L$, $\varrho' = \frac{d}{d\gamma_3} \varrho$ (where also, as usual, $\tilde{\lambda}_i = \tau^*_\subM \lambda_i$ and thus $\tilde{\bm{\lambda}} = \tau^*_\subM \bm{\lambda}$).

Then it remains to compute  the 2-form $\langle J, d^\C\mathcal{A}_\subS^1 \otimes \eta_1+ d^\C\mathcal{A}_\subS^2 \otimes \eta_2 \rangle$ where $\eta_i = f_i \zeta_1 + g_i \zeta_2$ are the corresponding horizontal gauge symmetries and  $\mathcal{A}_\subS^i((\eta_j)_\subM) = \delta_j^i$ for each $i,j=1,2$.  The problem with this example (and also with the example in Section \ref{Example_BallOnSurface}) is that the horizontal gauge symmetries can not be explicitly written except in a few particular cases.  So, we will write the 2-form $\langle J, d\mathcal{A}_\subS^i \otimes \eta_i\rangle$ in the basis $\{\mathcal{Y}_1, \mathcal{Y}_2\}$. Let us denote by $\mathcal{Y}^1, \mathcal{Y}^2$ the 1-forms such that $\mathcal{Y}^i(\mathcal{Y}_j) = \delta_j^i$ and $\mathcal{Y}^i|_W =0$.

\begin{lemma} \label{L:B_1-Basis}
If the functions $(f_i, g_i)$ satisfy the system of ordinary differential equations $\left(  \begin{smallmatrix}  f'\\ g'  \end{smallmatrix} \right)  = \Phi  \left(  \begin{smallmatrix}  f\\ \ \\ g  \end{smallmatrix} \right)$  on $Q_{reg}/G$ then 
\begin{equation}\label{Ex:Solids}
B_1 = \langle J, \mathcal{K}_\subW\rangle - \mathcal{J}_i \Phi_{ij} d\gamma_3 \wedge \tilde{\mathcal{Y}}^j.
\end{equation}
\end{lemma}
\begin{proof}
 Let us denote by $T=T(\gamma_3)$ the inverse transpose of the matrix $F$.  Then $(\mathcal{A}^1_\subS, \mathcal{A}_\subS^2)  = (\mathcal{Y}^1, \mathcal{Y}^2 ) T^t$ for $T^t$ denoting the transpose of $T$.   Then, denoting by $T'$ the matrix with elements $(T_{ij})' = \frac{d}{d\gamma_3} T_{ij}$, we get
\begin{equation*}
 \begin{split}
  \langle J, d\mathcal{A}_\subS^i \otimes \eta_i\rangle |_\C & =  (J_1, J_2) \, (d^\C\!\mathcal{A}_\subS^1 , d^\C\!\mathcal{A}_\subS^2 ) = (\mathcal{J}_1, \mathcal{J}_2) F^t d^\C(T \left(  \begin{smallmatrix}  \tilde{\mathcal{Y}}^1\\ \tilde{\mathcal{Y}}^2  \end{smallmatrix} \right) ) \\
 & =  \left( \mathcal{J}_i d\tilde{\mathcal Y}^i +  (\mathcal{J}_1, \mathcal{J}_2) F^t T' \left(  \begin{smallmatrix}  d\gamma_3\wedge\tilde{\mathcal{Y}}^1\\ d\gamma_3\wedge\tilde{\mathcal{Y}}^2  \end{smallmatrix} \right)   \right)\!|_\C =  \left( \mathcal{J}_i d\tilde{\mathcal Y}^i -  (\mathcal{J}_1, \mathcal{J}_2) \Phi \left(  \begin{smallmatrix}  d\gamma_3\wedge\tilde{\mathcal{Y}}^1\\ d\gamma_3\wedge\tilde{\mathcal{Y}}^2  \end{smallmatrix} \right)   \right)\!|_\C
\nonumber
 \end{split}
\end{equation*}
where $\tilde{\mathcal Y}^i = \tau_\subM^*{\mathcal Y}^i$ is the corresponding 1-form on $\M$. 
\end{proof}

Following \cite[Sec.~3.6]{Bal-16} we have that $(\mathcal{J}_1, \mathcal{J}_2 ) \Phi = (\mathcal{Q}, \mathcal{P})$ and then  \eqref{Ex:Solids} becomes
$$
\langle J, d^\C\mathcal{A}_\subS^i \otimes \eta_i\rangle = {\mathcal J}_i d^\C\tilde{\mathcal Y}^i -\mathcal{Q} d\gamma_3\wedge \tilde{\mathcal Y}^1 - \mathcal{P}d\gamma_3 \wedge \tilde{\mathcal Y}^2.
$$
Next, we observe that  the sections $\mathcal{Y}_i$ can be written as $\mathcal{Y}_1 = (\zeta_1)_Q = (P_{\g_\subS}((0,0),0;1))_Q=  ((0,0),0;1)_Q + h_1^1 (\xi_1)_Q + h_1^2(\xi_2)_Q$ and analogously with $\mathcal{Y}_2 =  ((0,0),1;0)_Q + h_2^i (\xi_i)_Q$ for functions $h_i^j\in C^\infty(Q)$. Hence, from Lemma \ref{L:Invariance} we conclude that $[\mathcal{Y}_1, \mathcal{Y}_2]$ and $[X,\mathcal{Y}_i]$ are sections in $W$ for $X\in \Gamma(H)$, $\rho$-projectable, and thus we obtain that $d\mathcal{Y}^i|_D=0$ or equivalently $d\tilde{\mathcal{Y}}_i|_\C = 0$.  Moreover, taking $\mathcal{Y}^1 = {\textstyle \frac{1}{1-\gamma_3^2}} (\gamma_3 \langle \bm{\gamma},\bm{\lambda}\rangle - \lambda_3)$ and $\mathcal{Y}^2 = {\textstyle \frac{1}{1-\gamma_3^2}} ( \langle \bm{\gamma},\bm{\lambda}\rangle - \gamma_3\lambda_3)$ we get that 
$$
\langle J, d^\C\mathcal{A}_\subS^i \otimes \eta_i\rangle = - \mathcal{Q}\langle \bm{\gamma}, d\tilde{\bm{\lambda}}\rangle - \mathcal{P} d\tilde{\lambda}_3.
$$

Finally, following Theorem \ref{T:GlobalB} we conclude that 
\begin{equation}\label{Ex:Solids:B}
B = B_1 = \langle J, \mathcal{K}_\subW \rangle +\langle J, d^\C\mathcal{A}_S^i \otimes \eta_i\rangle =  m\varrho \langle \bm{\gamma}, \bm{s}\rangle \langle \bm{\Omega},d\tilde{\bm{\lambda}}\rangle
\end{equation}
induces a Poisson bracket $\{\cdot, \cdot\}_\red^1$ on $\M_{reg}/G$ that describes the (reduced) dynamics, recovering the previous results in \cite{Bal-16, LGN-JM-16}. 
In fact,  since the 2-form $B$ given in \eqref{Ex:Solids:B} is well defined in the whole manifold $\M$, then the reduced bracket $\{\cdot,\cdot\}_\red^\B$ can be defined on the differential space $\M/G$ (moreover, it was proven in \cite{Bal-16}, that $\{\cdot,\cdot\}_\red^\B$ is a Poisson bracket on the differential space $\M/G$ using formulation given Prop.~\ref{Prop:JacB}).

\begin{remark}
 The authors in \cite{Bol-Kil-Kaz-14} analyze the possibility of having a bracket in the reduced space with a 4-dimensional almost symplectic foliation given by the level set of only one of the horizontal gauge momenta.  In order to have a 4-dimensional foliation, we need to reduced only by the group $SO(3)$  in which case $S$ is generated by $\mathcal{Y}_2$. However, as we saw, $\zeta_2$ does not generate a horizontal gauge symmetry, in fact we need $\zeta_1$ coming with the $S^1$-symmetry.  That is why, the dynamics cannot be described by a bracket with a 4-dimension foliation only one of the first integrals as Casimir functions.  
\end{remark}

\subsection{A homogeneous ball rolling in the interior side of a cylinder}

Consider the mechanical system formed by a homogeneous ball rolling in the interior side of a vertical circular cylinder \cite{Bal-14,BGM-96,BKM-02,Marle-03,Ramos-04}.
This example shares many properties with the mechanical system studied in Sec.~\ref{Example_BallOnSurface}. In this case, we will only explain the theory and the computations can be found in \cite[Sec.~7.4]{Bal-14} where it is shown that the system admits a dynamical gauge transformation inducing a Poisson bracket on the reduced space $\M/G$ with the horizontal gauge symmetries becoming Casimirs.

Indeed, this nonholonomic system has a symmetry given by a free and proper action of a Lie group $G$ satisfying the dimension assumption so that $\textup{rank}(S)=2$. Since the mechanical system has also 2 horizontal gauge momenta \cite{BGM-96} which are $G$-invariant we can apply Theorem \ref{T:GlobalB}.  Moreover, we can also observe from \cite{Bal-14} that, since $\textup{rank}(D)=3$ then $\textup{rank}(H)=1$ and by Corollary \ref{C:l+1}$(i)$, we conclude that there is a unique dynamical gauge transformation given by the 2-form $B=B_1$ (defined in \eqref{Eq:B1}) such that the reduced bracket $\{\cdot,\cdot\}_\red^B$ on $\M/G$ is Poisson with a 2-dimensional sympletic foliation defined by the level sets of both (reduced) horizontal gauge momenta. Following \eqref{Eq:B1} we recover the 2-form $B$ computed in \cite{Bal-14}.

\section{Example: The homogeneous ball in a convex surface of revolution} \label{Example_BallOnSurface}

In this section we  study the hamiltonization problem of the mechanical system formed by a homogeneous ball rolling without sliding on a convex surface of revolution  \cite{BKM-02, FGS-05, Hermans-95, Ramos-04, Routh-55, Zenkov-95}. 
In \cite{BKM-02}, it was shown that, in local coordinates, the reduced system is described by a Poisson bracket
(after time reparametrization). The properties of this Poisson bracket have also been studied by Ramos \cite{Ramos-04} 
who observed, by a dimensional argument, that reparametrization of time is not necessary (see also \cite{FGS-05}).
However, the geometry underlying the existence of such a Poisson structure was not studied yet. 

In particular, we study the geometric background of this nonholonomic system following Sec.~\ref{Sec:geometric_approach_nh} and afterwards we show, using Theorem \ref{T:GlobalB}, that
there exists an almost Poisson bracket $\{ \cdot, \cdot \}_\B$ describing the dynamics of the nonholonomic system
such that, after reduction by symmetries, it induces a Poisson bracket $\{ \cdot, \cdot \}_\red^\B$ on $\M/G$.  Moreover, we end this section writing explicitly the 2-form $B$ inducing a dynamical gauge transformation and the Poisson bracket on $\M/G$.

\subsection{The model}
Consider the motion of a homogeneous ball of mass $m$ and radius $R$ rolling without sliding under the influence of gravity on the interior side of a convex surface of revolution $\Sigma \subset \R^3$  with vertical axis of symmetry (parallel to the gravity force). We denote by $(x,y)$ the coordinates of the projection of the center of mass of the ball to the plane $z=0$. The homogeneity means that the inertia tensor of the ball has the form $\mathbb{I} = I \cdot \textup{Id}$, where $I$ is a positive constant and $\textup{Id}$ the $3 \times 3$ identity matrix. This mechanical system has two kinds of constraints: the holonomic constraint imposing the motion of the ball on the surface $\Sigma$ and the nonholonomic one given by the non-sliding condition.  

We fix an orthonormal frame 
in the ambient 3-dimensional space (space frame), a moving orthonormal frame 
attached to the ball (body frame) and denote by $g$ the orthogonal matrix relating both frames. Therefore the configuration space $Q$ of the mechanical system is $\Sigma \times SO(3) \simeq \R^2 \times SO(3)$ with coordinates $(x,y,g)$.
Since the ball is homogeneous the system has a $SO(3)$-symmetry given by the freeness in choosing the body frame and an extra $S^1$-symmetry induced by the axial symmetry of the surface $\Sigma$ where the ball moves.


The nonholonomic constraints are described by the following equation relating the angular velocity $\vec{\omega}$ and the velocity of the center of mass $\vec{v}$,
\begin{equation}
\label{NH_constr_vector}
\vec{\omega} \times \vec{a} = - \vec{v},
\end{equation}
where $\times$ denotes the usual vector product in $\R^3$ and $\vec{a}$ is a vector joining the center of mass of the ball with the contact point of the ball with the surface. 
Denoting by $\vec{n}$ the exterior unit normal vector to the surface, we have that $\vec{a} = R \vec{n}$. On the other hand, following Borisov, Mamaev and Kilin \cite{BKM-02}, the angular momentum with respect to the contact point is given by
$
\vec{M} = I \vec{\omega} + m R^2 \vec{n} \times ( \vec{\omega} \times \vec{n} ) = I \vec{\omega} + m R^2 \vec{\omega} - m R^2 (\vec{\omega} \cdot \vec{n} ) \vec{n}.
$

The equations of motion are found considering Newton's second law for translations and rotations,
\begin{equation}
\label{Newton's_law}
m \dot{\vec{v}} = \vec{N} + \vec{F} , \qquad
I \dot{\vec{\omega}} = R \vec{n} \times \vec{N},
\end{equation}
where $\vec{N}$ denotes the reaction force at the contact point and $\vec{F}$ the external force applied to the center of mass. The force $\vec{F}$ is the gradient of the potential energy $U(x,y) = m {\bf g} z$, where ${\bf g}$ denotes the acceleration of gravity and $z = z(x,y)$ indicates the height (vertical position) of the center of mass of the ball which is constrained to lie on the surface $\Sigma$. Eliminating $\vec{N}$ from (\ref{Newton's_law}) and using \eqref{NH_constr_vector} one gets (see \cite{BKM-02}),
\begin{equation}
\label{equations_motion_physics}
\dot{\vec{M}} = m R^2 \dot{\vec{n}} \times ( \vec{\omega} \times \vec{n} ) + \vec{M_F},  \qquad
\dot{\vec{r}} + R \dot{\vec{n}} = - R \vec{\omega} \times \vec{n},
\end{equation}
where $\vec{M_F}$ denotes the moment of force (torque) of $\vec{F}$. 
The equations of motion in \eqref{equations_motion_physics} do not depend on the orientation of the ball, so they can be considered as the equations of motion after reduction by the $SO(3)$-symmetry. Moreover, following \cite{BKM-02}, the equations  (\ref{equations_motion_physics}) are valid on any smooth surface, however we consider a smooth convex surface of revolution with vertical axis of symmetry so that we have the extra $S^1$-symmetry.

\subsection{Geometric aspects of the system and the nonholonomic bivector field}
\noindent{\bf The constraints and the nonholonomic bivector field.}
Let $\Sigma$ denote the convex surface of revolution described by the center of mass of the ball ${\bf x} = (x,y,z)$. We suppose that $\Sigma$ is parametrized in Cartesian coordinates by
\begin{equation}
\label{equation_Sigma}
\Sigma = \{ {\bf x} \in \R^3 : z = \phi(x^2+y^2) \} ,
\end{equation}
where the smooth function $\phi : \R^+ \to \R$ describes the profile curve of the surface $\Sigma$. To guarantee smoothness and convexity of the surface, we assume that $\phi$ verifies that $\phi'(0^+)=0$, $\phi'(s)>0$ and $\phi''(s)>0$ when $s>0$.  
To ensure that the ball has only one contact point with the surface we ask the curvature of $\phi(s)$ to be at most $1/R$. 
Denoting by $\vec{n}=\vec{n}(x,y) = (n_1,n_2,n_3)$ the exterior unit normal to the surface $\Sigma$, due to the smoothness of $\phi$ we have that $n_3$ is never zero and
\begin{equation}
\label{formulas_ni_phi'}
\frac{n_1}{n_3} = 2x \phi ' , \quad \frac{n_2}{n_3} = 2y \phi ' , \quad n_3 = - \frac{1}{(1 + 4 (x^2 + y^2) (\phi ')^2 )^{1/2}}.
\end{equation}

Then the configuration manifold $Q$ is identify with $Q \simeq \R^2\times SO(3)$ with coordinates $q =(x,y,g)$ and let us 
denote by $\bm{\omega} = (\omega_1,\omega_2,\omega_3)$ the angular velocity associated to the right invariant frame $\{ X_1^R, X_2^R, X_3^R \}$ of $T(SO(3))$.  Hence, from \eqref{Lagrangian_body} we obtain the mechanical Lagrangian $L : TQ \rightarrow \R$ given by 
\begin{eqnarray}
L (q;\dot x, \dot y, \bm{\omega}) &=& L_0 ({\bf x},g,\dot{\bf x}, \bm{\omega})|_{TQ} = \left( \textstyle{\frac{1}{2}} m \langle \dot{\bf x}, \dot{\bf x}\rangle + \textstyle{\frac{1}{2}} \langle \I\bm{\omega}, \bm{\omega}\rangle - m{\bf g} z\right)|_{TQ} \\ \nonumber 
&=& \frac{m}{2 n_3^2} \left( (1 - n_2^2) \dot{x}^2 + 2 n_1 n_2 \dot{x} \dot{y} + (1 - n_1^2) \dot{y}^2 \right) + \frac{I}{2} (\omega_1^2 + \omega_2^2 + \omega_3^2) - m {\bf g} \phi(x^2 + y^2), \label{Lagrangian}
\end{eqnarray}
where we used (\ref{equation_Sigma}) and (\ref{formulas_ni_phi'}).

The non-sliding constraints (\ref{NH_constr_vector}) are written, in terms of the coordinates $(\dot x,\dot y, \omega_1,\omega_2,\omega_3)$ of $TQ$, as $\dot{x} = - R (\omega_2 n_3 - \omega_3 n_2)$ and $\dot{y} = - R (\omega_3 n_1 - \omega_1 n_3)$,
defining the associated constraint 1-forms
\begin{equation*}
\epsilon^1 = dx - R (n_2 \rho_3 - n_3\rho_2 ), \qquad
\epsilon^2 = dy - R (n_3 \rho_1 - n_1\rho_3 ),
\end{equation*}
where we denoted by $\{ \rho_1, \rho_2, \rho_3 \}$ the right Maurer-Cartan 1-forms on $SO(3)$, dual to right invariant vector fields $\{ X_1^R, X_2^R, X_3^R \}$. The constraint distribution $D$ on $Q$, defined by the annihilator of $\epsilon^1$ and $\epsilon^2$, is given by
\begin{equation} \label{Eq:ExBall:D}
D = \textup{span} \left \{ Y_x := \partial_{x} + \frac{n_2}{R n_3} X_n - \frac{1}{R n_3} X^R_2, \:\: Y_y := \partial_{y} - \frac{n_1}{R n_3} X_n + \frac{1}{R n_3} X^R_1 , \:\: X_n \right  \},
\end{equation}
where $X_n := \langle \vec{n} , {\bf X}^R \rangle$ where ${\bf X}^R = (X_1^R, X_2^R, X_3^R)$ and $\langle \cdot, \cdot\rangle$ denotes, as usual, the canonical inner product.  
Consider now the dual basis of $TQ$ and $T^*Q$ given by
\begin{equation}
\label{Eq:ExBall:BasisTQ}
\mathfrak{B}_{TQ} = \{ Y_x , Y_y, X_n, Z_1, Z_2 \} \quad \text{and} \quad \mathfrak{B}_{T^*Q} = \{ dx, dy, \beta_n, \epsilon^1, \epsilon^2 \},
\end{equation} 
where $\beta_n = \langle \vec{n} , \bm{\rho} \rangle$ for $\bm{\rho} = (\rho_1, \rho_2, \rho_3)$ and where $Z_1$ and $Z_2$ are vector fields defined by
\begin{equation}
\label{formula_Zi}
Z_1 := \frac{1}{R n_3} X^R_2 - \frac{n_2}{R n_3} X_n, \quad
Z_2 := - \frac{1}{R n_3} X^R_1 + \frac{n_1}{R n_3} X_n.
\end{equation}
We denote by $(\dot{x},\dot{y},\omega_n,v_1,v_2)$ and $(p_x,p_y,M_n,M_1,M_2)$ the coordinates in $TQ$ and $T^*Q$ associated to the frames \eqref{Eq:ExBall:BasisTQ}, respectively, where $\omega_n = \langle \vec{n} , \bm{\omega} \rangle$ denotes the normal component of the angular velocity in the space frame $\bm{\omega}$.

\begin{remark}
\label{rmk:choice_Z_i}
Later we will see that the vector fields $Z_1$ and $Z_2$ induce  a $G$-invariant vertical complement of the constraints.
\end{remark}

The constraint manifold $\M = Leg(D)  \subset T^*Q$ defined in (\ref{definition_M})
is given by
\begin{equation}
\M = \Big\{ (x, y , g , p_x , p_y, M_n, M_1 , M_2) \: : \: M_1 = - {\textstyle \frac{I}{E}}p_x , \: M_2 = - {\textstyle \frac{I}{E}} p_y  \Big\},
\end{equation}
where $E := I + m R^2$. Following \eqref{eq:basis_TQ}, we consider a basis of $T^*\M$ given by 
\begin{equation}
\label{basis_of_T*M}
\mathfrak{B}_{T^* \M} = \Big\{ \tilde{dx}, \tilde{dy}, \tilde{\beta}_n, \tilde{\epsilon}^1, \tilde{\epsilon}^2, d p_x, d p_y, d M_n \Big\},
\end{equation}
where $\tilde{dx} = \tau_\subM^* dx$, $\tilde{dy} = \tau_\subM^* dy$, $\tilde{\beta_n} = \tau_\subM^* \beta_n$, $\tilde{\epsilon}^1 = \tau_\subM^* \epsilon^1$, $\tilde{\epsilon}^2 = \tau_\subM^* \epsilon^2$ with, as usual, $\tau_\subM : \M \rightarrow Q$ the canonical projection.  Its corresponding dual basis is given by 
\begin{equation}
\label{Ex:basis_of_TM}
\mathfrak{B}_{T \M} = \Big\{ \tilde{Y_x}, \tilde{Y_y}, \tilde{X_n}, \tilde{Z_1}, \tilde{Z_2}, \partial_{p_x}, \partial_{p_y}, \partial_{M_n} \Big\}, 
\end{equation}
where, as usual, we denote with a tilde the vector fields on $\M$ to distinguish them from their corresponding vector fields on $Q$.
Hence, from \eqref{definition_C}, the constraint distribution $\C$ on $\M$ is given by $\C = \textup{span} \Big\{ \tilde{Y_x}, \tilde{Y_y}, \tilde{X_n}, \partial_{p_x}, \partial_{p_y}, \partial_{M_n} \Big\}.$

From \eqref{definition_pi_nh} and \eqref{nonholonomic_dynamics}, we obtain the following Proposition. 

\begin{proposition}
\label{Prop:ExBall:pi_nh}
The nonholonomic bivector field $\pi_{\emph{\nh}}$ on $\M$ is given by
\begin{equation*} 
\begin{split}
\pi_{\emph{\nh}} = & \ \tilde{Y_x} \wedge \partial_{p_x} + \tilde{Y_y} \wedge \partial_{p_y} + \tilde{X_n} \wedge \partial_{M_n} \\
& + M_n D_{xy}^n \partial_{p_x} \wedge \partial_{p_y}
+ {\textstyle \frac{I}{E}} ( p_x D_{yn}^x + p_y D_{yn}^y ) \partial_{p_y} \wedge \partial_{M_n} 
- {\textstyle \frac{I}{E}} ( p_x D_{xn}^x + p_y D_{xn}^y ) \partial_{M_n} \wedge \partial_{p_x},
\end{split}
\end{equation*}
where $D_{xy}^n$, $D_{xn}^x$, $D_{yn}^y$, $D_{xn}^y$ and $D_{yn}^x$ are basic functions on $\M$ (with respect to the bundle $\tau_\subM : \M \rightarrow Q$) given by
\begin{equation}
\label{formulas_Dxyn}
\begin{split}
D_{xy}^n &= \frac{1}{R n_3} \left( n_1^x + n_2^y +\frac{1}{R}  \right),  \qquad \qquad \qquad \qquad
D_{xn}^x = - D_{yn}^y = R \left( - n_1^y n_3 + n_1 n_3^y + \frac{n_1 n_2}{R n_3} \right) , \\
D_{xn}^y &= R \left( n_1^x n_3 - n_1 n_3^x - \frac{n_1^2 + n_3^2}{R n_3} \right), \qquad \qquad \ \ 
D_{yn}^x =  R \left( - n_2^y n_3 + n_2 n_3^y + \frac{n_2^2 + n_3^2}{R n_3} \right) ,
\end{split}
\end{equation}
with $n_i^x$ and $n_i^y$ the partial derivatives of $n_i$, $i=1,2,3$, with respect to $x$ and $y$, respectively; i.e., $n_i^x =\partial_xn_i$ and $n_i^y = \partial_yn_i$.
\end{proposition}

\begin{proof}
First observe that the Liouville 1-form on $\M$  is written in the adapted basis (\ref{basis_of_T*M}) as $\Theta_\subM = \iota^*\Theta_Q =  p_x \tilde{dx} + p_y \tilde{dy} + M_n \tilde{\beta_n} -\frac{I}{E} p_x \tilde{\epsilon}^1 -\frac{I}{E} p_y \tilde{\epsilon}^2.$ Hence, $\Omega_\C = \Omega_\subM |_\C$ is given by
\begin{equation*}
\Omega_\subC = \left( - d p_x \wedge \tilde{dx} - d p_y \wedge \tilde{dy} - d M_n \wedge \tilde{\beta_n} - M_n d \tilde\beta_n + \frac{I}{E} p_x  d \tilde{\epsilon}^1 + \frac{I}{E} p_y d \tilde{\epsilon}^2 \right) |_\subC.
\end{equation*}

Now, using that $\langle \vec{n} , n^x \rangle = \langle \vec{n}, n^y \rangle = 0$, we get 
$d \beta_n |_D = D_{xy}^n  dx \wedge dy |_D,$ 
and
\begin{equation*}
d \epsilon^1 |_D = - (D_{xn}^n dx \wedge \beta_n  + D_{yn}^x dy \wedge \beta_n  ) |_D \quad \mbox{and} \quad d\epsilon^2 |_D =- (D_{xn}^y dx \wedge \beta_n + D_{yn}^y dy \wedge \beta_n ) |_D.
\end{equation*}
Therefore, the 2-section $\Omega_\subC$ becomes
\begin{equation}
\label{Omega_C}
\begin{split}
\Omega_\subC = \: & ( - d p_x \wedge \tilde{dx} - d p_y \wedge \tilde{dy} - d M_n \wedge \tilde{\beta_n} \\
&- M_n D_{xy}^n \tilde{dx} \wedge \tilde{dy} - {\textstyle \frac{I}{E}} ( p_x D_{xn}^x + p_y D_{xn}^y ) \tilde{dx} \wedge \tilde{\beta_n} - {\textstyle \frac{I}{E}}  ( p_x D_{yn}^x + p_y D_{yn}^y ) \tilde{dy} \wedge \tilde{\beta_n} ) |_\subC.
\end{split}
\end{equation}
Finally, the nonholonomic bivector is computed using \eqref{definition_nh_bracket} and \eqref{definition_pi_nh} and we get the desired expression.

\end{proof}

The hamiltonian 
\begin{equation*}
\Ham_\subM = {\textstyle \frac{R^2}{2 E}} ( (1-n_1^2) p_x^2 + (1-n_2^2) p_y^2 - 2 p_x p_y n_1 n_2  ) +  {\textstyle \frac{1}{2 I}} M_n^2 + m a_g \phi(x^2 + y^2) ,
\end{equation*}
induced by the Lagrangian \eqref{Lagrangian} defines the nonholonomic vector field $X_\nh = - \pi_\nh^\sharp (d \Ham_\subM)$ that is given by
\begin{equation}
\label{X_nh_formula1}
\begin{split}
X_\nh &= \dot{x} \tilde{Y_x} + \dot{y} \tilde{Y_y} + \omega_n \tilde{X_n} \\
&+ \left( \dot{y} M_n D_{xy}^n + (p_x n_1 + p_y n_2) (p_x n_1^x + p_y n_2^x ) + \omega_n {\textstyle \frac{I}{E}} ( p_x D_{xn}^x + p_y D_{xn}^y)  - 2m {\bf g} x \phi'  \right) \partial_{p_x} \\
&+ \left( - \dot{x} M_n D_{xy}^n + (p_x n_1 + p_y n_2) (p_x n_1^y + p_y n_2^y ) + \omega_n {\textstyle \frac{I}{E}} ( p_x D_{yn}^x + p_y D_{yn}^y )  - 2m {\bf g} y \phi' \right) \partial_{p_y} \\
&+ {\textstyle \frac{I}{E}} \left( - \dot{x} ( p_x D_{xn}^x + p_y D_{xn}^y)  - \dot{y} ( p_x D_{yn}^x + p_y D_{yn}^y)  \right) \partial_{M_n},
\end{split}
\end{equation}
where we used the following expressions relating velocities and momenta given by the Legendre transformation,
\begin{equation}
\label{velocities-momenta}
\dot{x} = {\textstyle \frac{R^2}{E}}  \left( p_x (1 - n_1^2) - p_y n_1 n_2 \right) , \qquad
\dot{y} = {\textstyle \frac{R^2}{E}} \left( p_y (1 - n_2^2) - p_x n_1 n_2 \right) , \qquad
\omega_n = {\textstyle \frac{M_n}{I}}.
\end{equation}

Previous works treating this example, such as \cite{BKM-02, FGS-05, Hermans-95, Ramos-04, Zenkov-95}, computed  the equations of motion from physical principles as recalled in the beginning of this section.
Dynamical properties of $X_\nh$ were studied in \cite{Hermans-95} where it was shown that the motion is quasi-periodic on tori of dimension at most three. That result was obtained using reconstruction from the dynamics of the reduced system by the symmetries.

\noindent{\bf Reduction by symmetries.} 
Consider the compact Lie group $G = S^1 \times SO(3)$, where $SO(3)$ defines a right action and $S^1$ a left action on $Q = \R^2 \times SO(3)$. More precisely, the action by an element $(\varphi, h) \in S^1 \times SO(3)$ on $(x,y,g) \in Q$ is given by
\begin{equation*}
(\varphi , h) \cdot (x, y, g) = (R_\varphi (x,y)^t, \hat R_\varphi \, g \, h ),
\end{equation*}
where $R_\varphi$ denotes the $2 \times 2$ rotation matrix of angle $\varphi$ and $\hat R_\varphi$ denotes the $3 \times 3$ rotation matrix of angle $\varphi$ with respect to the $z$-axis. It is straightforward to prove that $G$ is a symmetry of the nonholonomic system.

The Lie algebra $\g$ of $G$ is isomorphic to $\R \times \R^3$ with the infinitesimal generator with respect to the $S^1$-action given by
\begin{equation*}
U_0 := (1; \mathbf{0})_Q = - y \partial_{x} + x \partial_{y}+ X^R_3,
\end{equation*}
and with respect to  the $SO(3)$-action are given by
\begin{equation*}
(0;\mathbf{e}_i)_Q = \alpha_i X^R_1 + \beta_i X^R_2 + \gamma_i X^R_3, \quad \mbox{for } i=1,2,3,
\end{equation*}
where $\mathbf{e}_i$ denotes the $i$-th canonical basis vector of $\R^3$ and, as usual $\vecalpha = (\alpha_1,\alpha_2,\alpha_3)$, $\vecbeta=(\beta_1,\beta_2,\beta_3)$ and $\vecgamma=(\gamma_1,\gamma_2,\gamma_3)$ for the rows of the matrix $g \in SO(3)$.
Consequently the vertical (generalized) distribution $V_q = T_q(Orb(q))$ is given by 
\begin{equation}\label{Eq:ExBall:V}
V = \textup{span} \{ U_0, X_1^R, X_2^R, X_3^R \}, 
\end{equation}
and then the $G$-symmetry satisfies the dimension assumption (\ref{dimension_assumption}). We observe that the rank of $V$ is 3 for $(x,y)=(0,0)$ and it is 4 elsewhere, showing that the action is not free (not even locally free).

Now we describe the reduced space $\M / G$ as a stratified differential space and write the reduced dynamics.
Considering the lifted $G$-action to the invariant manifold $\M$ we get, for $(\varphi , h)\in G$ and $(x, y, g, p_x, p_y, M_n ) \in \M$, that
\begin{equation}
\label{cotangent_action}
(\varphi , h) \cdot (x, y, g, p_x, p_y, M_n ) = (R_\varphi (x,y)^t, \hat R_\varphi \, g \, h , R_\varphi (p_x, p_y)^t,  M_n ).
\end{equation}

Since the actions of $S^1$ and $SO(3)$ commute, the reduction of $\M$ by the symmetry group $G = S^1 \times SO(3)$ is performed by stages as in \cite{FGS-05, Hermans-95} (see also Sec.~\ref{Sec:Solids}).
The reduction by $SO(3)$ induces the smooth manifold $\M / SO(3)$ and results in the elimination of the coordinate $g$ of $\M$. Furthermore, from \eqref{cotangent_action} we see that $S^1$ acts on $\M / SO(3)$ by 
$$
\varphi \cdot (x,y,p_x,p_y,M_n) = (R_\varphi (x,y)^t , R_\varphi (p_x,p_y)^t, M_n).
$$

Since $(0,0,0,0, M_n)$ is a fixed point for any rotation $R_\varphi$, the $S^1$-action is not free and the reduction is performed using {\em invariant theory} as in \cite{FGS-05, Hermans-95, Ramos-04}. The $S^1$-invariant polynomials in $\M / SO(3)$ for this action are given by
\begin{equation}
\label{invariant_p_i}
\begin{split}
p_0 &= p_x^2 + p_y^2 , \qquad
p_1 = x^2 + y^2 , \qquad
p_2 = x p_x + y p_y , \\
p_3 &= x p_y - y p_x , \qquad
p_4 = M_n,
\end{split}
\end{equation}
and they form a basis of $C^\infty(\M)^G$.   As we saw in Section \ref{symmetries_reduction_proper_actions}, the (proper) $G$-action induces on the reduced space $\M / G$ a differential structure where the invariant polynomials $p_i$, $i=0,...,4$ are considered as ambient coordinates for $\M/G$ in the sense that $\M / G$ is described by the following semi-algebraic subset of $\R ^5$,
\begin{equation*}
\{ (p_0, p_1, p_2, p_3, p_4) \in \R^5 : p_0 \geq 0, \: p_1 \geq 0, \:  p_0 p_1 = p_2^2 + p_3^2 \}.
\end{equation*}

The stratified orbit space $\M / G$ has two strata corresponding to the orbit types. The $1$-dimensional {\it singular stratum} associated to the $S^1$ isotropy is given by 
\begin{equation*}
M_1 = \{ (p_0, p_1, p_2, p_3, p_4) \in \R^5 : p_0 = p_1 = p_2 = p_3 = 0 \},
\end{equation*}
and corresponds to the situation where the ball lies at the origin of the surface and spins about the vertical axis. 
The other $4$-dimensional stratum, called {\it regular stratum}, is the complement of $M_1$ in $\M / G$ and is given by
\begin{equation*}
M_4 = \{ (p_0, p_1, p_2, p_3, p_4) \in \R^5 : p_0 \geq 0, \: p_1 \geq 0, \: p_0 p_1 = p_2^2 + p_3^2, \: p_0^2 + p_1^2 > 0 \}.
\end{equation*}

The manifold $M_4$ corresponds to the orbit space of the submanifold $\M_{reg}$ of $\M$ where the $S^1$-action is free, i.e. $M_4 = \M_{reg}/G$.
It was observed in \cite{FGS-05, Hermans-95} that the regular stratum $\M_{reg}/G$ is diffeomorphic to $S^2 \times \R^2$.
Observe that possible trajectories in the regular stratum include the case when the ball passes through the bottom of the surface $\Sigma$ but with nonzero velocity. On the other hand, from (\ref{velocities-momenta})
we can rewrite the polynomials $p_i$, $i=0,...,4$ in the velocity coordinates, and we observe that trajectories verifying $p_3 = 0$ are those where the ball moves only in the radial direction, i.e. in the intersection of the surface with a vertical plane passing through the origin, while when $p_2 = 0$ the ball moves in a circular trajectory at a constant height. The variables $p_3$ and $p_4$ will be important to compute first integrals.

The reduced dynamics $X_\red$ on $\M/G$ is computed projecting the nonholonomic vector field $X_\nh$ given in \eqref{X_nh_formula1} by the quotient map $\rho : \M \to \M\ / G$ and has been presented in \cite{Hermans-95} (see also \cite{FGS-05}). In our reduced variables $p_i$, $i=0,...,4$, the reduced nonholonomic vector field $X_\red$ is given by
\begin{equation}
\label{Eq:ExBall:RedDyn}
X_\red = \bar{F_0} p_2 \partial_{p_0}  + \bar{F}_1 p_2 \partial_{p_1} + \bar{F_2} \partial_{p_2} + \frac{R (n_3)^2}{E} \bar{F_3} p_2 p_4 \partial_{p_3} + \frac{R^3 I (n_3)^2}{E^2} p_2 p_3 \bar{F_4} \partial_{p_4} ,
\end{equation}
where $n_3, \bar{F}_1, \bar{F_3}, \bar{F_4}$ are basic functions with respect to the map $\M / G \to Q/G$ (thinking of $Q / G$ as a differential space) given by
\begin{equation}
\label{functions_F3_F4}
\begin{split}
n_3 & = n_3(p_1) =  - \frac{1}{ (1 + 4 p_1 \phi'(p_1)^2 )^{1/2}}, \qquad \qquad \quad \bar{F_1} = \bar{F_1}(p_1) = \frac{2 R^2 (n_3)^2}{E}, \\
\bar{F_3} &= \bar{F_3}(p_1) = 2( \phi '(p_1) + 2 p_1 \phi ''(p_1) ) (n_3)^2, \qquad \ 
\bar{F_4} = \bar{F_4}(p_1) = 4(2 \phi '(p_1) ^3 - \phi '' (p_1)) (n_3)^2 ,
\end{split}
\end{equation}
and $\bar{F_0}$ and $\bar{F_2}$ are functions on $\M /G$ (see \cite{FGS-05} for the explicit expressions of $\bar{F}_0$ and $\bar{F}_2$).

As it has been remarked in \cite{FGS-05}, the reduced dynamics is well defined for $p_1=0$ and $p_0 \neq 0$ (on the regular stratum), that is for orbits passing through the origin of the surface with non-zero momentum. 
Note that $X_\red$ is a vector field in $\R^5$ which is tangent to the space $\M / G$ as a stratified space, that is $X_\red$ is tangent to each strata \cite{FGS-05}. In the singular stratum $M_1$ it reduces to the equation $\dot{p_4}=0$ with trivial solution and on the regular stratum it defines a smooth vector field on the smooth manifold $\M_{reg} /G$. 

The equilibria of the reduced equations of motion are of two types: the singular equilibria which are all the points of the singular stratum with $p_4 = M_n$ constant, and the points of the regular stratum verifying $p_2 = 0$ and $\bar{F_2}(p_0,p_1,p_2,p_3,p_4) = 0$, which describe circular motions at constant height, see \cite{FGS-05, Hermans-95, Routh-55,  Zenkov-95}. Hermans \cite{Hermans-95} has shown that away from the equilibrium points all the orbits of the reduced dynamics are periodic. 
The qualitative study of the reduced equilibrium points was started by Routh \cite{Routh-55}, where he gave necessary conditions for their stability by linearizing the dynamics. Afterwards, Zenkov \cite{Zenkov-95} proved that linear stability imply nonlinear (orbital) stability.

\begin{remark}
Since the surface $\Sigma$ is of revolution, the functions $\bar{F_3}$ and $\bar{F_4}$ can be expressed in terms of the principal curvatures $\mu_1$ and $\mu_2$ of $\Sigma$. In fact, since
\begin{equation*}
\mu_1 = - \frac{2 \phi '}{ (1 + 4 p_1 (\phi ')^2)^{1/2} } \quad \mbox{and} \quad \mu_2 = - \frac{2 \phi ' + 4 p_1 \phi ''}{ (1 + 4 p_1 (\phi ')^2)^{3/2} },
\end{equation*}
we obtain, on $Q_{reg}$ (i.e., $p_1 \neq 0$), that $ \bar{F_3}(p_1) = \displaystyle{\frac{\mu_2}{n_3}}$ and $\bar{F_4}(p_1) = \displaystyle{\frac{\mu_1 - \mu_2}{n_3 p_1}}$.

\end{remark}

\noindent {\bf The reduced bracket $\{\cdot, \cdot\}_\red$ on $\M/G$.} \ Following \eqref{definition_reduced_bracket}, the nonholonomic bracket $\{\cdot, \cdot\}_\nh$ (given in Prop.~\ref{Prop:ExBall:pi_nh}) induces a reduced bracket $\{\cdot, \cdot\}_\red$ on $\M/G$.
Next, using Prop.~\ref{Prop:Jac} we show that $\{\cdot, \cdot\}_\red$ is not Poisson.  In order to compute the 3-form $d^\C\langle J, \mathcal{K}_\subW\rangle$ we start by choosing a $G$-invariant vertical complement of the constraints (see Def.~\ref{Def:W}).

Let us consider the vector fields $Z_1$ and $Z_2$ on $Q$ given in \eqref{formula_Zi} and define the distribution $W$ on $Q$ by 
\begin{equation*}
W = \textup{span}\{Z_1, Z_2\}.
\end{equation*}
From \eqref{Eq:ExBall:BasisTQ} we see that $W$  is a vertical complement of the constraints $D$ in $TQ$, i.e.,  $TQ = D\oplus W$ with $W\subset V$. Moreover, it is straightforward to check that $W$ is a $G$-invariant distribution and it is constant rank.   Let us denote by $\xi_1$ and $\xi_2$ the sections of the bundle $\g \times Q \to Q$ given, at each point $q = (x,y,g)$, by
\begin{equation}
\label{sections_of_gw}
\xi_1 |_q = (0;  \mathbf{C}_1\, g) \qquad \mbox{and} \qquad  \xi_2 |_q = (0; \mathbf{C}_2\, g),
\end{equation}
where  $\mathbf{C}_1$, $\mathbf{C}_2$ are the vectors
\begin{equation*} 
\textbf{C}_1 = - \frac{1}{R n_3}
\begin{pmatrix}
n_1 n_2 , & n_2^2 - 1 \: , & n_2 n_3
\end{pmatrix},
\qquad
\textbf{C}_2 = \frac{1}{R n_3}
\begin{pmatrix}
n_1^2 - 1 \: , & n_1 n_2 \: , & n_1 n_3
\end{pmatrix}.
\end{equation*}
Observe that $Z_1 = (\xi_1)_Q$ and $Z_2 = (\xi_2)_Q$ and hence the subbundle $\g_W \to Q$ of $\g\times Q\to Q$ defined by 
\begin{equation*}
\g_W = \textup{span}\{\xi_1, \xi_2\}
\end{equation*}
induces $W$ and is $Ad$-invariant. Analogously, the distribution $\W = \textup{span} \{ (\xi_1)_\subM, (\xi_2)_\subM \}$ is a vertical complement of the constraints $\C$.

\begin{remark}
\label{vertical_symmetry_condition}
Following Rmk.~\ref{R:Solids:VertSym}, we enforce the fact that this examples does not admit a vertical complement $W$ of the constraints satisfying the {\it vertical symmetry condition} since  the Lie group $SO(3)$ is simple and therefore does not have any normal subgroup which could act on the system as a symmetry group, see \cite[Remark 2.4]{Bal-16}.

\end{remark}

\begin{theorem}
\label{reduced_bracket_is_not_Poisson}
The reduced bracket $\{ \cdot , \cdot \}_{\emph{\red}}$ on $\M / G$ is not Poisson.
\end{theorem}

\begin{proof}
We check the Jacobi identity of $\{\cdot, \cdot\}_\red$ on the functions $p_1, p_3, p_4$ using Prop.~\ref{Prop:Jac}.  In order to do that, following \eqref{Def:JK} and Lemma \ref{lemma:K_W_in_sections} we compute
\begin{equation}\label{Ex:Ball:JK}
\begin{split}
  \langle J, \mathcal{K}_\subW\rangle & =  - {\textstyle \frac{I}{E}} p_x d^\C \tilde \epsilon^1 - {\textstyle \frac{I}{E}} p_y d^\C \tilde \epsilon^2 \\
  &  = \ {\textstyle \frac{I}{E}}  (p_x D_{xn}^x +p_yD_{xn}^y ) \tilde dx \wedge \tilde \beta_n + {\textstyle \frac{I}{E}}  (p_xD_{yn}^x + p_yD_{yn}^y) \tilde dy \wedge \tilde \beta_n,
  \end{split}
\end{equation}
where $D_{xn}^x$, $D_{xn}^y$, $D_{yn}^x$ and $D_{yn}^y$ are the functions defined in \eqref{formulas_Dxyn} and hence 
\begin{equation} \label{Ex:Ball:dJK}
\begin{split}
d^\C \langle J , \mathcal{K}_\subW \rangle &  = \left( - {\textstyle \frac{I}{E}}  d^\C p_x \wedge d^\C \tilde{\epsilon}^1 - {\textstyle \frac{I}{E}}  d^\C p_y \wedge d^\C \tilde{\epsilon}^2 + \Psi \right),  \\
& = {\textstyle \frac{I}{E}}  \left( D_{xn}^x d p_x \wedge \tilde{dx}+ D_{xn}^y d p_y \wedge \tilde{dx}  + D_{yn}^x d p_x \wedge \tilde{dy} \wedge + D_{yn}^y d p_y \wedge \tilde{dy}\right)  \wedge \tilde{\beta_n}  + \Psi,
\end{split}
\end{equation}
 where $\Psi$  is a semi-basic 3-form with respect to the bundle $\tau_\subM : \M \rightarrow Q$.

On the other hand,  using Prop.~\ref{Prop:ExBall:pi_nh}, we also get that
\begin{equation*}
\begin{split}
X_{p_1} = \pi_\nh^\sharp(d \rho^* p_1) = & 2(x \partial_{p_x} + y\partial_{p_y} ), \\
X_{p_3} = \pi_\nh^\sharp (d \rho^* p_3) = & - x \tilde Y_y + y \tilde Y_x -
(x D_{xy}^n M_n - p_y ) \partial_{p_x} - (y D_{xy}^n M_n +  p_x ) \partial_{p_y} \\
&  + {\textstyle \frac{I}{E}}  \left(p_x (x D_{yn}^x - y D_{xn}^x) +  p_y(x D_{yn}^y - y D_{xn}^y) \right) \partial_{M_n} ,
\\
X_{p_4} = \pi_\nh^\sharp (d \rho^* p_4) =& - \tilde{X_n}  -{\textstyle \frac{I}{E}} (p_x D_{xn}^x + p_y D_{xn}^y) \partial_{p_x} - {\textstyle \frac{I}{E}} (p_x D_{yn}^x + p_y D_{yn}^y) \partial_{p_y}. 
\end{split}
\end{equation*}
Now, since $\Psi$ is semi-basic, we have that ${\bf i}_{\pi_\nh^\sharp(d\rho^*p_1)} \Psi =0$ and thus we obtain that
\begin{equation*}
d^\C \langle J , \mathcal{K}_\subW \rangle (X_{p_1},X_{p_3},X_{p_4})  =  -\frac{2I}{E} \left( x\, d\tilde \epsilon^1(X_{p_3}, X_{p_4}) + y \, d\tilde \epsilon^2(X_{p_3}, X_{p_4}) \right) 
= -\frac{2I}{E} \left( 2xy D_{xn}^x - x^2 D_{yn}^x + y^2 D_{xn}^y \right). 
\end{equation*}
Moreover, if we suppose that $y = 0$ with $x\neq 0$, then $n_2 =0$ and using \eqref{formulas_Dxyn} we get
$$
d^\C \langle J , \mathcal{K}_\subW \rangle (\pi_\nh^\sharp(d \rho^* p_1), \pi_\nh^\sharp(d \rho^* p_3), \pi_\nh^\sharp(d \rho^* p_4)) |_{y=0}  =  \frac{2I}{E} x^2 n_3(- R n_2^y + 1) |_{y=0},
$$
which is different from zero since $R n_2^y |_{y=0} \neq 1$ (in fact, we have $n_2^y |_{y=0} = - \frac{2 \phi'}{\sqrt{1+4x^2(\phi')^2}} <0$).

\end{proof} 

%

For completeness we compute the reduced bracket $\{ \cdot , \cdot\}_\red$ which is given in our reduced variables $p_i$ by
\begin{equation*}
\begin{split}
  \{p_0, p_1\}_\red &= 4p_2, \quad \{p_0, p_2\}_\red = 2p_0 - 2 p_3p_4D_{xy}^n, \quad \{p_0,p_4\}_\red = -2p_2p_4D_{xy}^n, \quad \{p_0, p_4\}_\red = {\textstyle \frac{I}{E}}  8  p_2 p_3 R n_3^2, \\
\{p_1, p_2\}_\red &= -2p_1, \qquad \{p_2, p_3\}_\red = - p_1 p_4 D_{xy}^n, \qquad \{p_2, p_4\}_\red = -{\textstyle \frac{IR}{E}} (n_3)^2 (\Phi_1(p_1) + 4 p_1 \Phi_2(p_1)) p_3, \\  
\{p_3, p_4\}_\red &  = {\textstyle \frac{IR}{E}}  (n_3)^2 \Phi_1(p_1) p_2,
\end{split}
\end{equation*}
with other combinations being equal to zero and where $\Phi_1$ and $\Phi_2$ are given by
$\Phi_1(p_1) = 2 \phi ' - \frac{1}{Rn_3}$ and 
$\Phi_2(p_1) = \phi '' - \frac{(\phi ')^2}{R n_3}$. We recall that $n_3$ and $D_{xy}^n$ are functions depending only on $p_1$ (that is, they are well defined functions on $Q/G$). 


\subsection{The dynamical gauge transformation}
\label{Sec:Ex:DynGauge}

In this section, we follow Diagram \eqref{Diagram} and Theorem \ref{T:GlobalB} to see that the system admits a dynamical gauge transformation $B$ so that the reduced induced bracket $\{\cdot, \cdot\}_\red^\B$ on $\M_{reg}/G$ is, in fact, Poisson.  Moreover, we write the explicit formulations for the 2-form $B$ and we study the Poisson bracket on the differential space $\M/G$ (giving also explicit formulations). In order to apply Theorem \ref{T:GlobalB} we start studying the horizontal gauge momenta associated to the system. 

\medskip

\noindent {\bf Horizontal gauge momenta}. 
From \eqref{Eq:ExBall:D} and \eqref{Eq:ExBall:V}, the distribution $S = D \cap V$ is given by $S = \textup{span} \{\mathcal{Y}_1 := -y Y_x + x Y_y, \ \mathcal{Y}_2 := X_n \}$ and we observe that its rank is  nonconstant: it is equal to $1$ at $(x,y)=(0,0)$ and equal to $2$ elsewhere. The subbundle $\g_S \to Q$ of $\g \times Q \to Q$ defined in (\ref{defbundle_gS}) is then generated by the sections
\begin{equation}
\label{sections_of_gs}
\zeta_1 |_q = (1; \mathbf{0}) + y \xi_1 |_q - x \xi_2 |_q - (0; \vecgamma) |_q ,
\qquad
\zeta_2 |_q = (0; \vec{n} g) |_q ,
\end{equation}
where we consider the normal $\vec{n}$ as a row vector and  $\xi_1$ and $\xi_2$ are the sections of $\g_W\to Q$ defined in \eqref{sections_of_gw}.  Then we have that $(\zeta_1)_Q = \mathcal{Y}_1$ and $(\zeta_2)_Q = \mathcal{Y}_2$ and we can check that $\g_S$ has constant rank equal to $2$ while the distribution $S$ varies its rank (see Sec.~\ref{section_splittings} and \cite{Bal-16}).

Let us now consider the components of the nonholonomic momentum map in the basis $\{ \zeta_1 , \zeta_2 \}$ of $\Gamma(\g_S)$, i.e.,
\begin{equation}\label{Ex:Ball:J1J2}
\mathcal{J}_1 = \mathbf{i}_{(\zeta_1)_\subM} \Theta_\subM = -y p_x + x p_y = \rho^*p_3 \quad \mbox{and} \quad \mathcal{J}_2 = \mathbf{i}_{(\zeta_2)_\subM} \Theta_\subM = M_n = \rho^*p_4.
\end{equation}
Observe that $\mathcal{J}_1$ and $\mathcal{J}_2$ are $G$-invariant functions, and thus we denote by $\bar{\mathcal J}_1 = p_3$ and $\bar{\mathcal J}_2 = p_4$ the induced functions on $\M/G$.  Then, it is straightforward to see that $\mathcal{J}_1$ and $\mathcal{J}_2$ are not first integrals of the nonholonomic dynamics $X_\nh$ since $X_\red(\bar{\mathcal J}_1) \neq 0$ and $X_\red(\bar{\mathcal J}_2) \neq 0$ for $X_\red$ given in \eqref{Eq:ExBall:RedDyn}.

%
%
%
%

Following the ideas of Sec.~\ref{NonholonomicMomentMap} (and Sec.~\ref{Sec:Solids}), we look for $G$-invariant horizontal gauge momenta $J_\eta$ on $\M$ for  $\eta$ a $G$-invariant section of $\S \rightarrow Q$.  Since an arbitrary ($G$-invariant) section $\eta$ of $\g_S \rightarrow Q$ is written as
$\eta = f \zeta_1 + g \zeta_2$,
for functions $f, g \in C^\infty (Q)^G$, then the nonholonomic momentum map associated to the section $\eta$ is given by
$$
J_\eta = \langle J^\nh , \eta \rangle = f \mathcal{J}_1 + g \mathcal{J}_2, 
$$
That is, we look for functions $\bar{f} = \bar{f}(p_1)$ and $\bar{g} = \bar{g}(p_1)$  on $Q / G$ so that $\bar{J}_\eta = \bar{f} p_3 + \bar{g} p_4$ is a first integral of the reduced dynamics $X_\red$
\begin{equation*}
X_\red (\bar{J}_\eta) = p_3 \bar{f} ' dp_1(X_\red) + p_4 \bar{g} ' dp_1(X_\red) + \bar{f} dp_3 (X_\red) + \bar{g} d p_4 (X_\red) = 0.
\end{equation*}
Using \eqref{Eq:ExBall:RedDyn}, the functions $\bar{f}$ and $\bar{g}$ satisfy the following linear system of ordinary differential equations defined on $Q_{reg}/G$, where $Q_{reg}$ denotes the manifold where the $G$-action on $Q$ is free,
\begin{equation}
\label{EDO_fi}
\bar{f} ' + \bar{g} \frac{R I}{2 E} \bar{F_4} (p_1)  = 0, \qquad \bar{g} ' + \bar{f} \frac{1}{2 R}   \bar{F_3} (p_1)  = 0,
\end{equation}
where $\bar{F_3}$ and $\bar{F_4} \in C^\infty (Q / G)$ are given in (\ref{functions_F3_F4}). The linear differential system (\ref{EDO_fi}) has two solutions in the domain of continuity of the functions $\bar{F_3}$ and $\bar{F_4}$ (that is, starting with two independent initial conditions, the solutions remain independent as long as they exist).  We denote by
$(\bar{f_1} (p_1) , \bar{g_1} (p_1))$ and $(\bar{f_2} (p_1) , \bar{g_2} (p_1))$ two independent solutions of \eqref{EDO_fi}, 
which verify $\bar f_1(p_1) \bar g_2(p_1) - \bar g_1(p_1) \bar f_2 (p_1) \neq 0$ in a interval containing the origin, and we get two first integrals in $\M / G$ of the form
\begin{equation*}
\bar{J}_1 (p_1, p_3, p_4 ) = \bar{f_1}(p_1) p_3 + \bar{g_1}(p_1) p_4 , \qquad
\bar{J}_2 (p_1, p_3, p_4 ) = \bar{f_2}(p_1) p_3 + \bar{g_2}(p_1) p_4 .
\end{equation*}
%
%
%
These first integrals were known since the work of Routh \cite{Routh-55}, see also \cite{FGS-05, Hermans-95}.
Hence, $J_1 = \rho^*\bar{J}_1$ and $J_2= \rho^*\bar{J}_2$ are $G$-invariant horizontal gauge momenta for $X_\nh$ with horizontal gauge symmetries given by
\begin{equation}
\label{gauge_symmestries_eta_i}
 \left(  \begin{smallmatrix}  \eta_1\\[4pt] \eta_2  \end{smallmatrix} \right)  = F  
\left(  \begin{smallmatrix}  \zeta_1\\ \zeta_2  \end{smallmatrix} \right)   \qquad \mbox{where} \quad F=F(p_1) =
\left(  \begin{smallmatrix}  f_1 & g_1 \\ f_2 & g_2    \end{smallmatrix} \right) ,
\end{equation}
where $\zeta_1$ and $\zeta_2$ are given in (\ref{sections_of_gs}). In general these first integrals are not explicitly known except for some particular cases such as the circular paraboloid \cite{BKM-02}.

Therefore, on $\M_{reg}$, the system satisfies the hypothesis of Theorem \ref{T:GlobalB}.  Moreover, since $\textup{rank}(D) -\textup{rank}(S) =1$, then any principal connection determines a horizontal distribution $H$ of rank 1 and thus, by Corollary \ref{C:l+1}, the 2-form $B=B_1$ induces a reduced Poisson bracket $\{\cdot , \cdot\}_\red^1$ on $\M_{reg}/G$ that describes the dynamics and has a 2-dimensional symplectic foliation determined by the level sets of $\bar{J}_1, \bar{J}_2$.

\medskip

\noindent{\bf The explicit expressions of the dynamical gauge transformation and conclusions.}
Next, we use Theorem \ref{T:GlobalB} and, in particular, Corollary \ref{C:l+1} to compute the 2-form $B=B_1$ given in \eqref{Eq:B1}. As a result, we obtain a Poisson bracket $\{\cdot , \cdot\}_\red^1$ on $\M_{reg}$ with two Casimirs induced by the $G$-invariant horizontal gauge momenta.    

Following formula \eqref{Eq:B1} and recalling \eqref{Ex:Ball:JK}, it remains to compute the 2-form $\langle J, d^\C \mathcal{A}_\subS^1 \otimes \eta_1 + d^\C \mathcal{A}_\subS^2\otimes \eta_2\rangle$. We have the same observation as in Example \ref{Sec:Solids}: the horizontal gauge symmetries $\eta_1, \eta_2$ are not explicitly known since the functions $f_i,g_i$ are given as a solution of a linear system of differential equations. Analogously as we did in the previous example, we will work in our basis $S = \textup{span}\{\mathcal{Y}_1,\mathcal{Y}_2\}$ on $Q_{reg}$. 
Denoting by $\{{\mathcal Y}^1, {\mathcal Y}^2\}$ the 1-forms such that ${\mathcal Y}^i({\mathcal Y}_j) = \delta_j^i$ and ${\mathcal Y}^i|_W = 0$ for $i, j=1,2$, from Lemma \ref{L:B_1-Basis}, and using \eqref{EDO_fi} we get
\begin{equation*}
  \langle J, d^\C\mathcal{A}_\subS^i \otimes \eta_i\rangle  = \mathcal{J}_i d\tilde{\mathcal Y}^i + {\mathcal J}_2 {\textstyle\frac{1}{2R}}F_3  d\tilde p_1 \wedge \tilde{\mathcal Y}^1 + \mathcal{J}_1 {\textstyle \frac{RI}{2E}} F_4 d\tilde p_1\wedge \tilde{\mathcal Y}^2,
\end{equation*}
where $\tilde{\mathcal Y}^i = \tau_\subM^*{\mathcal Y}^i$ and $\tilde p_1 = \rho^*p_1$ are the corresponding 1-forms and function on $\M$. 

On $Q_{reg}$, we can define the 1-forms ${\mathcal Y}^1= \frac{-ydx +xdy}{x^2+y^2}$ and ${\mathcal Y}^2=\beta_n$ and obtain that 
$$
d{\mathcal Y}^1 |_D = 0, \qquad d{\mathcal Y}^2 |_D= D_{xy}^ndx\wedge dy|_D, \qquad dp_1\wedge {\mathcal Y}^1 = 2dx\wedge dy \quad \mbox{and} \quad dp_1\wedge {\mathcal Y}^2 = 2(xdx+ydy)\wedge \beta_n.
$$
Therefore, using also \eqref{Ex:Ball:J1J2} we finally obtain  
$$
\langle J, d^\C\mathcal{A}_\subS^i \otimes \eta_i\rangle = M_n (D_{xy}^n + {\textstyle \frac{F_3}{R}} )\tilde dx\wedge \tilde dy + (-yp_x+ xp_y){\textstyle \frac{I}{E}} R F_4 (x\tilde dx + y\tilde dy)\wedge \tilde \beta_n.
$$
Then, from Theorem \ref{T:GlobalB} and Corollary \ref{C:l+1} and using \eqref{Ex:Ball:JK} we obtain

\begin{theorem} \label{T:ExBall:Gauge}
The nonholonomic system  describing the dynamics of a homogeneous ball rolling without sliding on a convex surface of revolution with symmetry group $S^1 \times SO(3)$  has a description of the reduced dynamics on $\M_{reg}/G$ given by a Poisson bracket $\{\cdot, \cdot \}_{\emph\red}^\B$.  This bracket is obtained by the reduction of an almost Poisson bracket $\pi_\B$ dynamically gauge related to $\pi_{\emph\nh}$ by the 2-form
\begin{eqnarray}
B \!\!\!\! & = & \!\!\!\! B_1  =  \langle J, \mathcal{K}_\subW\rangle + \langle J, d^\C\mathcal{A}_\subS^i \otimes \eta_i\rangle \\ \label{B_bola_teorema}
 \!\!\!\!& = & \!\!\!\! M_n ( D_{xy}^n + \frac{F_3}{R}) \tilde dx \wedge \tilde dy + \frac{I}{E}\left( K_{xn} + x (x  p_y - y p_x) R F_4 \right) \tilde dx \wedge \tilde \beta_n + \frac{I}{E} \left( K_{yn}  + y (x  p_y - y  p_x) R F_4 \right) \tilde dy \wedge \tilde \beta_n\nonumber
\end{eqnarray}
where $K_{xn} =  p_x D_{xn}^x +  p_y D_{xn}^y$ and $K_{yn} =  p_x D_{yn}^x +  p_y D_{yn}^y$.  Moreover, this 2-form $B$ is uniquely defined and the Poisson bracket $\{\cdot, \cdot \}_{\emph\red}^\B$ on $\M_{reg}/G$ has 2-dimensional symplectic leaves given by the level sets of the (reduced) horizontal gauge momenta $\bar{J}_1$ and $\bar{J}_2$.
\end{theorem}

\begin{remark} \label{R:Ex:Ball:3form}
Using the coordinates $\dot{x}$, $\dot{y}$ and $\omega_n$  in (\ref{velocities-momenta}), the $2$-form $B$ is also written as
$$
B = \Phi(x,y) \left( \omega_n \: \tilde{dx} \wedge \tilde{dy} + \dot{x} \: \tilde{dy} \wedge \tilde{\beta_n} + \dot{y} \: \tilde{\beta_n} \wedge \tilde{dx} \right),
$$
where the function $\Phi$ is given by $\Phi(x,y) = (1 - 2 \phi ' R n_3) \frac{I}{R^2 n_3}$
and it is $S^1$-invariant because the surface is of revolution ($\phi '$ and $n_3$ are $S^1$-invariant).  Following \cite{LGN-JM-16} we see that $B$ can also be written as a contraction of the $3$-form
\begin{equation*}
\Pi = \Phi(x,y) \: \tilde{dx} \wedge \tilde{dy} \wedge \tilde{\beta_n},
\end{equation*}
with the nonholonomic vector field $X_\nh$, that is, $B = \mathbf{i}_{X_\nh} \Pi$.
\end{remark}

\noindent {\bf The reduced Poisson bracket $\{\cdot, \cdot\}_\red^\B$ on the differential space $\M/G$}. 
From \eqref{B_bola_teorema} (or Rmk.~\ref{R:Ex:Ball:3form}) we see that the 2-form $B$ is well defined in the whole manifold $\M$ and thus, the bivector field $\pi_\B$ is well defined on $\M$. Using the 2-section 
$\Omega_\subC$ in \eqref{Omega_C}, from \eqref{Eq:GaugedNH},  we compute explicitly the bivector field $\pi_\B$
\begin{equation}
\label{pi_B}
\begin{split}
\pi_\B &= \tilde{Y_x} \wedge \partial_{p_x} + \tilde{Y_y} \wedge \partial_{p_y} + \tilde{X_n} \wedge \partial_{M_n} \\
&- M_n \frac{F_3}{R} \partial_{p_x} \wedge \partial_{p_y}
- \frac{I}{E} y (x p_y - y p_x ) R F_4  \partial_{p_y} \wedge \partial_{M_n} 
+ \frac{I}{E} x (x p_y - y p_x ) R F_4  \partial_{M_n} \wedge \partial_{p_x} ,
\end{split}
\end{equation}
which can be compared with $\pi_\nh$ given in Prop.~\ref{Prop:ExBall:pi_nh}.

Since $B$ is $G$-invariant, then the almost Poisson bivector field $\pi_\B$ induces a bracket $\{\cdot, \cdot\}_\red^\B$ on the whole differential space $\M/G$ as in \eqref{eq:reduced_bracket_B} (that, on $\M_{reg}/G$, coincides with the Poisson bracket described in Theorem \ref{T:ExBall:Gauge}).

\begin{theorem}
\label{poisson_in_M/G}
On the differential space $\M /G$ the reduced bracket $\{ \cdot , \cdot \}_{\emph\red}^\B$, induced by \eqref{pi_B}, is Poisson.
\end{theorem}

\begin{proof} Using that $p_0,p_1,p_2,p_3,p_4$, given in \eqref{invariant_p_i}, form a basis of $G$-invariant functions on $\M$,  we have to check that $cyclic [ \{p_i , \{ p_j, p_k\}_\red^\B \}_\red^\B ] = 0$ for all $i,j,k=0,...,4$. 
Following Prop.~\ref{Prop:JacB} we have that 
\begin{equation}
\label{jacobiator_pi}
cyclic [ \{p_i , \{ p_j, p_k\}_\red^\B \}_\red^\B ]\circ \rho = (d^\C \langle J , \mathcal{K}_\subW \rangle - d^\C B) (\pi_\B^\sharp (d \rho^* p_i), \pi_\B^\sharp (d \rho^* p_j), \pi_\B^\sharp (d \rho^* p_k) ),
\end{equation}
for $i,j,k=0,...,4$. First observe that $\pi_\B^\sharp(dJ_i)_{(m)} = -(\eta_i)_\subM(m)$ for all $m\in\M$ and $i=1,2$, and then $\{p_j, \bar{J}_i\}_\red^\B \circ \rho = \{\rho^*p_j, J_i\}_\B = 0$ for all $j=0,...,4$.
Now, since $\bar{J}_1$ and $\bar{J}_2$ depend linearly (on $Q/G$) on $p_3$ and $p_4$, it remains to compute (\ref{jacobiator_pi}) on $p_0$, $p_1$ and $p_2$.
From \eqref{Ex:Ball:dJK} and \eqref{B_bola_teorema} we have that 
$$
d^\C\langle J, \mathcal{K}_\subW\rangle = \Upsilon_1 \wedge \tilde \beta_n + \Psi \qquad \mbox{and} \qquad dB = \Upsilon_2 \wedge \tilde \beta_n + {\textstyle \frac{\Phi(x,y)}{I}} d \tilde M_n \wedge \tilde dx \wedge \tilde dy,
$$
where $\Upsilon_1$ and $\Upsilon_2$ are 2-forms on $\M$, $\Psi$ is the semi-basic 3-form defined in \eqref{Ex:Ball:dJK} and $\Phi$ is the function defined in Remark \ref{R:Ex:Ball:3form}.

On the one hand, the vector field $\pi_\B^\sharp (d \rho^* p_1) = 2 ( x \partial_{p_x} + y \partial_{p_y})$ lies in the kernels of $d M_n \wedge \tilde{dx} \wedge \tilde{dy}$ and $\Psi$.  On the other hand, from \eqref{invariant_p_i}, we see that $p_0, p_1$ and $p_2$ do not depend on the variable $M_n$ and thus, following \eqref{pi_B}, the vector fields $\pi_\B^\sharp(d\rho^*p_i)$ are independent of $\tilde X_n$ (i.e., $\tilde\beta_n (\pi_\B^\sharp(d\rho^*p_i)) = 0$), for $i=0,1,2$, arriving to the fact that 
$$
(d^\C \langle J , \mathcal{K}_\subW \rangle - d^\C B) (\pi_\B^\sharp (d \rho^* p_0), \pi_\B^\sharp (d \rho^* p_1), \pi_\B^\sharp (d \rho^* p_2) ) = 0.
$$ 

\end{proof}

For completeness let us compute the reduced bracket $\{ \cdot , \cdot \}_\red^\B$ on $\M / G$:  
\begin{equation*}
 \begin{split}
  \{p_0, p_1\}_\red^\B & = 4p_2, \qquad \{p_0, p_2\}_\red^\B = 2p_0 - 2 p_3p_4\frac{F_3}{R}, \qquad \{p_0, p_4\}_\red^\B = \frac{I}{E} 2  p_2 p_3 R {F_4}, \\
\{p_1, p_2\}_\red^\B & = -2p_1, \qquad \{p_2, p_3\}_\red^\B =  p_1 p_4 \frac{{F_3}}{R}, \qquad \{p_2, p_4\}_\red^\B = \frac{I}{E}  p_1 p_3 R {F_4},
 \end{split}
\end{equation*}
with other combinations being equal to zero.   We see that the Poisson structure on the singular stratum $M_1$ is trivial.

\noindent{\bf Integrability.}
As mentioned in the introduction, the nonholonomic system treated in this last example has been studied by several authors without reducing a bivector in $\M$ but instead working directly with the properties of the reduced dynamics in $\M / G$. Fass\`{o}, Giacobbe and Sansonetto \cite{FGS-05} used a specific dynamical property of the reduced system in order to find a Poisson bracket in $\M /G$. In fact, outside the equilibrium points, the orbits of the reduced dynamics are the fibers of a (locally trivial) fibration with fiber $S^1$ and the period of the flow is a continuous (and smooth) function of the initial data. Then, each symplectic leaf is densely filled by periodic orbits, which, from \cite{Hermans-95}, reconstruct to tori of dimension at most 3 in $\M$.

\end{document}